\documentclass[acmsmall,dvipsnames,x11names,svgnames,cleveref,screen,nonacm]{acmart}

\hypersetup{breaklinks=true}

\usepackage[utf8]{inputenc}
\usepackage[T1]{fontenc}
\usepackage{microtype}

\usepackage{graphicx}
\usepackage{amsmath}
\usepackage{amsfonts}
\usepackage{adjustbox}
\usepackage{csquotes}
\usepackage{stmaryrd}
\usepackage{mathtools}
\usepackage{todonotes}
\usepackage{xspace}
\usepackage{color}
\usepackage{xcolor}
\usepackage[capitalise]{cleveref}
\usepackage{caption}
\usepackage{subcaption}
\usepackage{yfonts}
\usepackage[official]{eurosym}
\usepackage{wrapfig}
\usepackage{xfrac}
\usepackage{nicefrac}
\usepackage{enumitem}
\usepackage{multirow}
\usepackage{makecell}
\usepackage[twoside]{rotating}
\usepackage{tablefootnote}

\usepackage{tikz}
\usetikzlibrary{patterns,decorations.pathreplacing,arrows,arrows.meta,decorations.pathmorphing,positioning,fit,trees,shapes,shadows,automata,calc}

\usepackage{macros}
\newtheorem*{convention}{Convention}{\itshape}{\rmfamily}
\newtheorem*{claim}{Claim}{\itshape}{\rmfamily}
\AtBeginDocument{\theoremstyle{acmdefinition}\newtheorem{remark}[theorem]{Remark}}

\AtBeginDocument{%
  \providecommand\BibTeX{{%
    \normalfont B\kern-0.5em{\scshape i\kern-0.25em b}\kern-0.8em\TeX}}}

\begin{document}

\title{Weighted Programming}
\subtitle{A Programming Paradigm for Specifying Mathematical Models}

\author{Kevin Batz}
\email{kevin.batz@cs.rwth-aachen.de}
\orcid{0000-0001-8705-2564}
\affiliation{%
    \institution{RWTH Aachen University}
    \city{Aachen}
    \country{Germany}
}

\author{Adrian Gallus}
\email{adrian.gallus@rwth-aachen.de}
\orcid{0000-0002-2176-5075}
\affiliation{%
    \institution{RWTH Aachen University}
    \city{Aachen}
    \country{Germany}
}

\author{Benjamin Lucien Kaminski}
\email{b.kaminski@ucl.ac.uk}
\orcid{0000-0001-5185-2324}
\affiliation{%
	\institution{Saarland University, Saarland Informatics Campus}
	\city{Saarbrücken}
	\country{Germany}
}
\affiliation{%
	\institution{University College London}
	\city{London}
	\country{United Kingdom}
}

\author{Joost-Pieter Katoen}
\email{katoen@cs.rwth-aachen.de}
\orcid{0000-0002-6143-1926}
\affiliation{%
    \institution{RWTH Aachen University}
    \city{Aachen}
    \country{Germany}
}

\author{Tobias Winkler}
\email{tobias.winkler@cs.rwth-aachen.de}
\orcid{0000-0003-1084-6408}
\affiliation{%
    \institution{RWTH Aachen University}
    \city{Aachen}
    \country{Germany}
}

\renewcommand{\shortauthors}{Batz \and Gallus \and Kaminski \and Katoen \and Winkler}



\titlenote{Accepted for publication (\url{https://doi.org/10.1145/3527310}).}


\begin{abstract}
	We study weighted programming, a programming paradigm for specifying mathematical models.
	More specifically, the weighted programs we investigate are like usual imperative programs with two additional features: (1) nondeterministic \emph{branching} and (2) \emph{weighting} execution traces.
	Weights can be numbers but also other objects like words from an alphabet, polynomials, formal power series, or cardinal numbers.
	We argue that weighted programming as a paradigm can be used to specify mathematical models beyond probability distributions (as is done in probabilistic programming).
	
	We develop weakest-precondition- and weakest-liberal-precondition-style calculi \`{a} la Dijkstra for reasoning about mathematical models specified by weighted programs.
	We present several case studies.
	For instance, we use weighted programming to model the \emph{ski rental problem} --- an \emph{optimization problem}.
	We model not only the optimization problem itself, but also the best deterministic online algorithm for solving this problem as weighted programs.
	By means of weakest-precondition-style reasoning, we can determine the \emph{competitive ratio} of the online algorithm \emph{on source code level}.
\end{abstract}

%
\maketitle
%

\section{Introduction and Overview}
\label{sec:introduction}

Weighted programs are usual programs with two distinct features: (1) nondeterministic \emph{branching} and (2) the ability to \emph{weight} the current execution trace.
A prime and very well-studied example of weighted programs are \emph{probabilistic programs} which can branch their execution depending on the outcome of a random coin flip.
For instance, the program $\PCHOICE{C_1}{\tfrac{1}{3}}{C_2}$ weights the trace that executes $C_1$ with probability $\sfrac{1}{3}$ and the trace executing $C_2$ with $1 - \sfrac{1}{3} = \sfrac{2}{3}$.
The weighted outcomes of the two branches are then --- simply put --- summed together.

Besides applications as \emph{randomized algorithms} for speed-up in solving computationally intractable problems, probabilistic programming has over the past decade gained rapidly increasing attention in machine learning.
There, probabilistic programs serve as \emph{intuitive algorithmic descriptions} of complicated probability distributions.
As~\citet{DBLP:conf/icse/GordonHNR14} put it:%
\begin{quote}
\slshape
\enquote{The goal of probabilistic programming is to enable probabilistic modeling \textnormal{\gray{[\dots]}}
to be accessible to the working programmer, who has sufficient domain expertise, but perhaps not enough expertise in probability theory \textnormal{\gray{[\dots]}}.}%
\end{quote}%
\noindent{}%
In this paper, we consider more general weights than probabilities --- in fact: more general than numbers.
We should stress that \emph{we are not the first} to consider weighted programs (see e.g.~\cite{DBLP:conf/mfps/AguirreK20,DBLP:conf/esop/BrunelGMZ14,DBLP:conf/esop/GaboardiKOS21} and see \Cref{sec:related} for detailed comparisons).

Our goal was, however, \emph{not} to merely go from probabilistic to weighted programming, just for the sake of generalization.
Instead, we advocate \emph{weighted programming} as a \emph{programming paradigm for specifying mathematical models}.
In particular, our prime goal is to take a step towards making mathematical modeling more accessible to people with a programming background.
In a nutshell:
		\begin{quote}
		\textbf{Render mathematical modeling accessible to the working programmer},\newline 
		who has sufficient domain expertise, 
		but perhaps not enough expertise\newline in the respective mathematical theory.%
		\end{quote}%
Towards that goal, let us have a look at how such modeling could work in practice.

\subsection*{Weighted Programming as a Paradigm for Specifying Mathematical Models}

As a motivating example, we consider the classical \emph{Ski Rental Problem}~\cite{DBLP:series/txtcs/Komm16}, a classical \emph{optimization problem}, studied also in the context of online algorithms and competitive analysis~\cite{DBLP:books/daglib/0097013}.
A precise \emph{textual} description of the problem is as follows:%
\begin{quotation}
	\itshape%
	\textbf{The Scenario:} A person does not own a pair of skis but is going on a skiing trip for $\varvaclen$ days. 
	At the beginning of each day, the person can chose between two options:
	Either \emph{rent} a pair of skis, costing $1\eur$ for that day; or \emph{buy} a pair of skis, costing $\varbuycost\eur$ (and then go skiing for all subsequent days free of charge).
	\smallskip
	
	\noindent{}%
	\textbf{The Question:} What is the optimal (i.e.\ minimal) amount of money that the person has to spend for a pair of skis for the entire length of the trip?
\end{quotation}%
\begin{wrapfigure}[16]{r}{0.29\textwidth}
  \begin{minipage}{.99\linewidth}
    \abovedisplayskip=-1.125em%
    \belowdisplayskip=0pt%
\begin{align*}
	& \hspace*{-1.7em}\textbf{\textit{The Scenario:}}\\
				\textnormal{\gray{\footnotesize 1:}}\quad &\WHILE{\varvaclen > 0} \\
				%
				%
				\textnormal{\gray{\footnotesize 2:}}\quad & \qquad \ASSIGN{\varvaclen}{\varvaclen - 1}\fatsemi \\
				%
				%
				& \qquad \{\\
				\textnormal{\gray{\footnotesize 3:}}\quad & \qquad\qquad\WEIGH{1}\\
				\textnormal{\gray{\footnotesize 4:}}\quad & \qquad \} \mathrel{\BranchSymbol} \{ \\
				\textnormal{\gray{\footnotesize 5:}}\quad & \qquad \qquad \COMPOSE{\WEIGH{\varbuycost}}{\ASSIGN{\varvaclen}{0}} \\
				& \qquad \} \\
				%
				& \}\\[1em]
				%
	& \hspace*{-1.7em}\textbf{\textit{The Question:}}\\
	& \hspace*{0em}\wp{\mathtt{opt}}{\sone} \eeq ?
\end{align*}%
\normalsize%
  \end{minipage}%
\end{wrapfigure}%
\noindent{}%
Using weighted programming, we can \emph{model the scenario of this optimization problem} in a quite natural, simple, and intuitive way by the weighted program $\mathtt{opt}$ on the right.
Intuitively, this weighted program tests each day whether the vacation is already over (Line 1).
If not, it does the following (Lines 2--5): 
First, it decrements the vacation length by~1 day (Line 2).
It then models the two options that the person has for each day by a \emph{nondeterministic branching} (Line 4).
In the left branch, it realizes the first option: Paying $1\eur$ (Line 3).
In the right branch, it realizes the second option: Paying $y\eur$ and then setting the remaining vacation length to $0$ (Line 5), because with respect to having to pay for a pair of skis (not with respect to the joy of skiing) the vacation has effectively ended.

If we now want to answer \emph{the question} of the optimization problem, we first chose a suitable \emph{semiring}.
In this setting of optimizing (i.e.\ minimizing) incurred cost, the \emph{tropical} semiring $\tropringdef$ comes to mind.
The carrier set of this semiring are the extended natural numbers.
The addition ($\sadd$) in this semiring is taking the minimum of two numbers.
In the program, this is reflected by the fact that in Line 4 we would like to make whatever choice is cheaper for us.
The multiplication ($\smul$) is the standard \emph{addition} of numbers.
In the program, this is reflected, for example, in Line 3, where we add a 1 to the current execution trace.

For actually \emph{answering the question} of the optimization problem, 
we determine a weakest precondition of sorts, but interpreted here in a more general \enquote{quantitative} setting, with respect to post\enquote{condition} $\sone$ --- the multiplicative identity of the semiring; in this case, $\sone$ is the natural number~$0$.
Intuitively, this will for each path multiply together the weights along the path (recall: semiring multiplication is natural number addition).
Then, we sum over the weights of all paths (recall: semiring summation is natural number minimization), thus yielding the accumulated costs along the least expensive path.
As a result, our weakest-precondition-style calculus will yield%
\begin{align*}
	\wp{\mathtt{opt}}{\sone} \eeq n \sadd y~,
\end{align*}
i.e.\ the minimum of the numbers $n$ and $y$.
This is precisely the solution to our optimization problem:
If the trip length $n$ is larger than the cost of buying skis, we should buy skis which will cost us $y\eur$.
If $n$ is smaller than $y$, we should instead rent each day (at cost $1\eur/\textnormal{day}$) which will cost us $n\eur$.

Toward competitive analysis, we can now model the cost of a deterministic online algorithm $\mathtt{onl}$ that solves the ski rental problem and determine $\wp{\mathtt{onl}}{\sone}$.
Then, we can compute the ratio $\wp{\mathtt{onl}}{\sone} / \,\wp{\mathtt{opt}}{\sone}$ to determine the \emph{competitive ratio} of the online algorithm $\mathtt{onl}$.

We stress that program $\mathtt{opt}$ from above is \emph{not} strictly speaking executable.
For that, one would need some sort of scheduler who determinizes the nondeterministic choices.
It is also not immediately clear what weighting the individual execution traces on a physical computer would mean.
Instead, the above weighted program \emph{encodes a mathematical model}, namely an optimization problem, by means of an \emph{algorithmic representation} --- much in the spirit of a probabilistic program that is also not necessarily meant to be executed but instead models a probability distribution.

Lastly, we would like to note that determining weakest preconditions is related to \emph{inference} in probabilistic programming~\cite{DBLP:conf/icse/GordonHNR14}, where one is concerned with, e.g., determining the probability that the probabilistic program establishes some postcondition.

\subsection*{Contributions}
Our main technical contribution is the --- to the best of our knowledge --- first weakest precondition-style reasoning framework for weighted programs, which conservatively extends both Dijkstra's classical weakest preconditions and weakest \emph{liberal} weakest preconditions. 
Our weakest pre calculi capture the semantics of \emph{unbounded and potentially nonterminating loops} effortlessly, while other works explicitly avoid partiality (see \Cref{sec:related}).
Our weakest liberal preweightings even give a nuanced semantics to nonterminating runs in order to reason about such traces as well.
We demonstrate the applicability of our framework by several examples.

To achieve a high degree of generality and applicability, our framework is parameterized by a monoid of \emph{weights} for weighting computations of programs and so-called \emph{weightings} that take over the role of \enquote{quantitative assertions}. We prove well-definedness and healthiness conditions of our calculi, provide formal connections to an operational semantics, and develop easy-to-apply invariant-based reasoning techniques. 

\subsection*{Outline}
\cref{sec:preliminaries} provides preliminaries on monoids and semirings. \cref{sec:wgcl} introduces the syntax and operational semantics of weighted programs. 
We introduce our weakest (liberal) preweighting calculi for reasoning about weighted programs in \cref{sec:weighted-semantics}. 
Invariant-style reasoning for loops is presented in \cref{sec:loop-verification}. In \cref{sec:app}, we demonstrate the efficacy of our framework by means of several examples.
In \Cref{sec:related}, we give an overview of and a comparison to other works that study weighted computations.
We conclude in \Cref{sec:conclusion}.

\section{Monoids and Semirings}
\label{sec:preliminaries}
The weights occurring in our programs are elements from a monoid.
Intuitively, this is because we would like to \enquote{multiply} the weights on a program's computation trace together in order to obtain the total weight of that trace.
In particular, this multiplication should be associative and allow for neutral, i.e. effectless weighting.
In \cref{sec:weightings} further below, we introduce \emph{monoid modules} that are another important ingredient for our theory. 
\begin{definition}[Monoids]
	A \emph{monoid} $\Monoid = (\monoid,\, \monop,\, \monone)$ consists of a \emph{carrier set} $\monoid$, an \emph{operation} $\monop \colon \monoid \times \monoid \to \monoid$, and an \emph{identity} $\monone \in \monoid$, such that for all $a,b,c \in \monoid$,%
	\begin{enumerate}
	\item the operation $\monop$ is associative, i.e.\qquad $a \monop ( b \monop c) \eeq ( a \monop b ) \monop c$,\quad and
	\item $\monone$ is an identity with respect to $\monop$, i.e.\qquad $a \monop \monone \eeq \monone \monop a \eeq a$.
	\end{enumerate}
    The monoid $\Monoid$ is called \emph{commutative} if moreoever $a \monop b = b \monop a$ holds.
    \qedtriangle
\end{definition}%
\noindent{}%
Important examples of monoids are the \emph{words monoid} $(\ab^*,\, \:\langConcat\:,\, \epsilon)$ over alphabet $\ab \neq \emptyset$ and the \emph{probability monoid} $([0,1],\,  \:\cdot\:,\, 1)$ (the latter is commutative).
Another algebraic structure that plays a key role in this paper are \emph{semirings}.
Even though they are not strictly required for our theory, they render the application of our framework easier and more intuitive.
This is because every semiring is a monoid module over itself, which we explain in more detail in \cref{sec:weightings}.
Our definition of semirings is stated below; for an in-depth introduction, we refer to~\cite[Ch.~1, 2]{Droste2009}.
As usual, multiplication $\smul$ binds stronger than addition $\sadd$ and we omit parentheses accordingly.%
\begin{definition}[Semirings]
	A \emph{semiring} $\sringdef$ consists of a \emph{carrier set} $\sring$, an \emph{addition}~${\sadd}\colon \sring \times \sring \to \sring$, a \emph{multiplication}~${\smul}\colon \sring \times \sring \to \sring$, a \emph{zero} $\snull \in \sring$, and a \emph{one} $\sone \in \sring$, 
	such that%
	\begin{enumerate}
	\item $(\sring,\, {\sadd},\, \snull)$ forms a \emph{commutative monoid};
	\item $(\sring,\, {\smul},\, \sone)$ forms a (possibly non-commutative) \emph{monoid};
	\item multiplication \emph{distributes} over addition, i.e.\ for all $a,b,c \in \sring$,%
	 \begin{align*}
		a \ssmul (b \sadd c) \eeq a \smul b \ssadd a \smul c
		\qand
		(a \sadd b) \ssmul c \eeq a \smul c \ssadd b \smul c~; \qqand
	\end{align*}%
	\item multiplication by zero \emph{annihilates} $\sring$, i.e.\ \quad $~\snull \smul a \eeq a \smul \snull \eeq \snull$~.\qedtriangle
	\end{enumerate}%
\end{definition}%
\noindent{}%
Our \enquote{cheat sheet} in \Cref{fig:module-overview} lists various example semirings along with possible applications in a weighted programming context.
Further well-known semirings not considered specifically in this paper include
(i) the \emph{Łukasiewicz} semiring~\cite{lukasiewicz-semiring-2003,lukasiewicz-semiring-2005} motivated by multivalued logics and related to tropical geometry~\cite{DBLP:journals/fss/GavalecNS15},
(ii) the \emph{resolution} semiring~\cite{DBLP:phd/hal/Bagnol14} from proof theory,
(iii) the \emph{categorial} and \emph{lexicographic} semirings~\cite{DBLP:journals/coling/SproatYSR14} used in natural language processing,
(iv) the \emph{thermodynamic} semirings~\cite{DBLP:journals/corr/abs-1108-2874,Marcolli_2014} employed in information theory, and
(v) the \emph{confidence-probability} semiring~\cite{WirschingHuberKoelbl2010}.
We leave the study of applications of weighted programs over these semirings for future work.
Finally, we mention that more complicated semirings can be created from existing ones through algebraic constructions like matrices, tensors, polynomials, or formal power series.

\begin{table}
    \centering
	\caption{Weighted programming cheat sheet.}
	\label{fig:module-overview}
	\raggedright
	\textbf{Optimization} via the \textbf{Tropical semiring} ${(\NatsInf,\, {\min},\, {+},\, {+\infty},\, 0)}$ \\[1ex]
	{
		\scriptsize
		\begin{tabular}{@{}lll@{}}
			\textbf{Weighting $\boldsymbol{\WEIGH{a}}$:} & \multicolumn{2}{l}{\textbf{Branching $\boldsymbol{\sadd}$:}} \\
			\parbox{.31\linewidth}{Accumulate cost $a$} &
			\multicolumn{2}{l}{\parbox{.64\linewidth}{Choose branch that will accumulate \emph{minimal} cost}}
			\\[1ex]
			\textbf{Postweighting $\boldsymbol{f}$:} & \textbf{$\boldsymbol{\wp{C}{0}}$:} & \textbf{$\boldsymbol{\wlp{C}{f}}$:} \\
			\parbox{.31\linewidth}{Cost that is accumulated after program termination; typically choose $f = 0$} &
			\parbox{.32\linewidth}{Minimal accumulated cost amongst all terminating executions of $C$} &
			\parbox{.32\linewidth}{Minimum of $\wp{C}{f}$ and the minimal accumulated cost amongst all non-terminating executions of $C$}
		\end{tabular}
	}

	\vspace{.7ex}

	\textbf{Optimization} via the \textbf{Arctic semiring} ${(\NatsInf_{-\infty},\, {\max},\, {+},\, {-\infty},\, 0)}$ \\[1ex]
	{
		\scriptsize
		\begin{tabular}{@{}lll@{}}
			\textbf{Weighting $\boldsymbol{\WEIGH{a}}$:} & \multicolumn{2}{l}{\textbf{Branching $\boldsymbol{\sadd}$:}} \\
			\parbox{.31\linewidth}{Accumulate cost $a$} &
			\multicolumn{2}{l}{\parbox{.64\linewidth}{Choose branch that will accumulate \emph{maximal} cost}}
			\\[1ex]
			\textbf{Postweighting $\boldsymbol{f}$:} & \textbf{$\boldsymbol{\wp{C}{0}}$:} & \textbf{$\boldsymbol{\wlp{C}{f}}$:} \\
			\parbox{.31\linewidth}{Cost that is accumulated after program termination; typically choose $f = 0$} &
			\parbox{.32\linewidth}{Maximal accumulated cost amongst all terminating executions of $C$} &
			\parbox{.32\linewidth}{Same as $\wp{C}{f}$ if all executions of $C$ terminate, else $+\infty$}
		\end{tabular}
	}

	\vspace{.7ex}

	\textbf{Optimization} via the \textbf{Bottleneck semiring}\tablefootnote{Also known as \emph{max-min-semiring}. The name \emph{Bottleneck semiring} is taken from \cite{semirings-for-breakfast}. Quantitative verification using this semiring is further studied in \cite{quantitative-strongest-post-arxiv,quantitative-strongest-post-oopsla}.} ${(\Reals_{-\infty}^{+\infty},\, {\max},\, {\min},\, {-\infty},\, {+\infty})}$ \\[1ex]
	{
		\scriptsize
		\begin{tabular}{@{}lll@{}}
			\textbf{Weighting $\boldsymbol{\WEIGH{a}}$:} & \multicolumn{2}{l}{\textbf{Branching $\boldsymbol{\sadd}$:}} \\
			\parbox{.31\linewidth}{Restrict capacity of current branch to $a$} &
			\multicolumn{2}{l}{\parbox{.64\linewidth}{Choose branch with accumulate \emph{maximal} capacity}}
			\\[1ex]
			\textbf{Postweighting $\boldsymbol{f}$:} & \textbf{$\boldsymbol{\wp{C}{f}}$:} & \textbf{$\boldsymbol{\wlp{C}{f}}$:} \\
			\parbox{.31\linewidth}{Upper bound after program termination; typically choose $f = +\infty$} &
			\parbox{.32\linewidth}{Maximum bottleneck amongst all \emph{terminating} executions of $C$} &
			\parbox{.32\linewidth}{Maximum bottleneck amongst all executions of $C$}
		\end{tabular}
	}

	\vspace{.7ex}

	\textbf{Model Checking} via the \textbf{Formal languages semiring} ${\bigl(2^{\ab^*},\, {\cup},\, \,{\langConcat}\,,\, \emptyset,\, \set{\varepsilon}\bigr)}$ \\[1ex]
	{
		\scriptsize
		\begin{tabular}{@{}lll@{}}
			\textbf{Weighting $\boldsymbol{\WEIGH{a}}$:} & \multicolumn{2}{l}{\textbf{Branching $\boldsymbol{\sadd}$:}} \\
			\parbox{.31\linewidth}{Append symbol $a$ to current trace} &
			\multicolumn{2}{l}{\parbox{.64\linewidth}{Account for/aggregate behavior of both branches}}
			\\[1ex]
			\textbf{Postweighting $\boldsymbol{f}$:} & \textbf{$\boldsymbol{\wp{C}{\{\varepsilon\}}}$:} & \textbf{$\boldsymbol{\wlp{C}{\{\varepsilon\}}}$:} \\
			\parbox{.31\linewidth}{Language that is appended to each terminated trace; typically choose $f = \{\epsilon\}$} &
			\parbox{.32\linewidth}{Language of all terminating traces of $C$} &
			\parbox{.32\linewidth}{Language of all terminating and nonterminating traces of $C$}
		\end{tabular}
	}

	\vspace{.7ex}

	\textbf{Combinatorics} via the (extended) \textbf{Natural Numbers semiring} ${(\NatsInf,\, {+},\, \,{\cdot}\,,\, 0,\, 1)}$ \\[1ex]
	{
		\scriptsize
		\begin{tabular}{@{}lll@{}}
			\textbf{Weighting $\boldsymbol{\WEIGH{a}}$:} & \multicolumn{2}{l}{\textbf{Branching $\boldsymbol{\sadd}$:}} \\
			\parbox{.31\linewidth}{Make $a$ copies of current path/trace} &
			\multicolumn{2}{l}{\parbox{.64\linewidth}{Sum up number of paths/traces of both branches}}
			\\[1ex]
			\textbf{Postweighting $\boldsymbol{f}$:} & \textbf{$\boldsymbol{\wp{C}{1}}$:} & \textbf{$\boldsymbol{\wlp{C}{f}}$:} \\
			\parbox{.31\linewidth}{Number of copies that is made of each terminated path/trace; typically choose $f = 1$} &
			\parbox{.32\linewidth}{Number of all terminating paths/traces of $C$} &
			\parbox{.32\linewidth}{Same as $\wp{C}{f}$ if all executions of $C$ terminate, else $+\infty$}
		\end{tabular}
	}

	\vspace{.7ex}

	\textbf{Hidden Markov Models} via the \textbf{Viterbi semiring} ${([0,1],\, {\max},\, \,{\cdot}\,,\, {0},\, {1})}$ \\[1ex]
	{
		\scriptsize
		\begin{tabular}{@{}lll@{}}
			\textbf{Weighting $\boldsymbol{\WEIGH{a}}$:} & \multicolumn{2}{l}{\textbf{Branching $\boldsymbol{\sadd}$:}} \\
			\parbox{.31\linewidth}{Let what follows happen with probability $a$} &
			\multicolumn{2}{l}{\parbox{.64\linewidth}{Choose branch of maximal probability}}
			\\[1ex]
			\textbf{Postweighting $\boldsymbol{f}$:} & \textbf{$\boldsymbol{\wp{C}{\iverson{\varphi}}}$:} & \textbf{$\boldsymbol{\wlp{C}{\iverson{\varphi}}}$:} \\
			\parbox{.31\linewidth}{Probability additionally multiplied to each terminated trace; typically choose $f = \iverson{\varphi}$, i.e.\ the indicator function of some event $\varphi$} &
			\parbox{.32\linewidth}{Maximal probability of a terminating execution establishing $\varphi$} &
			\parbox{.32\linewidth}{Maximal probability of a non-terminating execution, or a terminating execution establishing $\varphi$}
		\end{tabular}
	}

	\vspace{.7ex}

	\textbf{Verification/Debugging} via the \textbf{Boolean semiring} ${(\{0,1\},\, {\lor},\, {\land},\, 0,\, 1)}$ \\[1ex]
	{
		\scriptsize
		\begin{tabular}{@{}lll@{}}
			\textbf{Weighting $\boldsymbol{\WEIGH{a}}$:} & \multicolumn{2}{l}{\textbf{Branching $\boldsymbol{\sadd}$:}} \\
			\parbox{.31\linewidth}{Assert predicate $a$} &
			\multicolumn{2}{l}{\parbox{.64\linewidth}{Angelic choice: choose \enquote{most true} branch}}
			\\[1ex]
			\textbf{Postweighting $\boldsymbol{f}$:} & \textbf{$\boldsymbol{\wp{C}{f}}$:} & \textbf{$\boldsymbol{\wlp{C}{f}}$:} \\
			\parbox{.31\linewidth}{Postcondition (a predicate) that should be established after program termination} &
			\parbox{.32\linewidth}{Weakest precondition of $f$, the weakest predicate $g$ so that starting in $g$, program $C$ \emph{can} terminate in state $\tau \models f$} &
			\parbox{.32\linewidth}{Weakest liberal precondition of $f$, the weakest predicate $g$ so that starting in $g$, program $C$ \emph{can} either diverge or terminate in state $\tau \models f$}
		\end{tabular}
	}

	\vspace{.7ex}

	\textbf{Feature Selection} via the \textbf{Why semiring} ${(\text{propositional \emph{positive} DNF}, {\lor}, {\land}, 0, 1)}$ over a finite set of variables $X_1, \ldots, X_n$ (cf. \cite{DBLP:journals/corr/abs-1910-07910}) \\[1ex]
	{
		\scriptsize
		\begin{tabular}{@{}lll@{}}
			\textbf{Weighting $\boldsymbol{\WEIGH{a}}$:} & \multicolumn{2}{l}{\textbf{Branching $\boldsymbol{\sadd}$:}} \\
			\parbox{.31\linewidth}{Use resource $a$} &
			\multicolumn{2}{l}{\parbox{.64\linewidth}{Alternatives: use resources via $C_1$ or via $C_2$}} \\[1ex]
			\textbf{Postweighting $\boldsymbol{f}$:} & \textbf{$\boldsymbol{\wp{C}{\iverson{\varphi}}}$:} & \textbf{$\boldsymbol{\wlp{C}{\iverson{\varphi}}}$:} \\
			\parbox{.31\linewidth}{Combinations of resources used after termination; typically choose $f = \iverson{\varphi}$} &
			\parbox{.32\linewidth}{Possible alternatives (disjunction) of resource sets (conjunction) to reach event $\varphi$} &
			\parbox{.32\linewidth}{Possible alternatives (disjunction) of resource sets (conjunction) to either not terminate or reach event $\varphi$}
		\end{tabular}
	}
\end{table}



\section{Weighted Programs}
\label{sec:wgcl}
For a monoid $\Monoid = (\monoid,\, \monop,\, \monone)$ of weights, we study the $\Monoid$-weighted guarded command language~\mbox{$\Monoid$-$\wgcl$} featuring --- in addition to standard control-flow instructions --- \emph{branching} and \emph{weighting}.
If the monoid $\Monoid$ is evident from the context, we omit the symbol and write just $\wgcl$.

\subsection{Syntax}
\label{sec:syntax}

$\wgcl$ programs $C$ adhere to the grammar
\iffalse
\begin{alignat*}{2}
	C \quad{}\longrightarrow{}
	& \quad &&\ASSIGN{x}{E}
	\tag{assignment} \\
	\qmid & &&\COMPOSE{C_1}{C_2}
	\tag{sequential composition} \\
	\qmid & &&\ITE{\guard}{C_1}{C_2}
	\tag{conditional choice} \\
	\qmid & &&\WHILEDO{\guard}{C_1}
	\tag{while loop} \\
	\qmid & &&\BRANCH{C_1}{C_2}
	\tag{branching} \\
	\qmid & &&\WEIGH{a}
	\tag{weighting} \\
	\gray{\qmid} & && \gray{\SKIP}
	\tag*{\gray{(syntactic sugar for $\WEIGH{\sone}$)}} \\
	\gray{\qmid} & &&\gray{\WCHOICE{C_1}{a}{b}{C_2}}
	\tag*{\gray{(syntactic sugar for $\BRANCH{\COMPOSE{\WEIGH{a}}{C_1}}{\COMPOSE{\WEIGH{b}}{C_2}}$)}}
\end{alignat*}%
\else
\begin{alignat*}{3}
	C \quad{}\longrightarrow{}
	& \quad &&\BRANCH{C_1}{C_2} &&\qmid \WEIGH{a}
	\tag{branching | weighting} \\
	\qmid & &&\ASSIGN{x}{E} &&\qmid \ITE{\guard}{C_1}{C_2}
	\tag{assignment | conditional choice} \\
	\qmid & && \COMPOSE{C_1}{C_2} &&\qmid \WHILEDO{\guard}{C_1}
	\tag{sequential composition | loop} \\
	\gray{\qmid} & && \gray{\SKIP ~\equiv~ \WEIGH{\sone}} &&\gray{\qmid} \gray{\WCHOICE{C_1}{a}{b}{C_2} ~\equiv~ \BRANCH{\COMPOSE{\WEIGH{a}}{C_1}}{\COMPOSE{\WEIGH{b}}{C_2}}}
	\tag*{\gray{(syntactic sugar)}}
\end{alignat*}
\fi
where $x$ is a program variable from a countable set $\Vars$, $E$ is an \emph{arithmetic expression} over $\Vars$, $\guard$ is a Boolean expression (also called \emph{guard}), and $a$ is a \emph{weight} from the monoid's carrier $\monoid$.

Our programs feature \emph{branching} \enquote{$\BRANCH{C_1}{C_2}$} and \emph{weighting} \enquote{$\WEIGH{a}$} of the current computation path where $a \in \monoid$ represents some weight.
For example, we can express a probabilistic choice \enquote{execute $C_1$ with probability $\nicefrac{1}{3}$ and $C_2$ otherwise} as $\WCHOICE{C_1}{\nicefrac{1}{3}}{\nicefrac{2}{3}}{C_2}$ over the monoid $([0,1],\, {\cdot},\, 1)$.
We allow for syntactic sugar in weightings, e.g.\ $\WEIGH{a^x}$ for some $a \in \monoid$, $x \in \Vars$.\footnote{The corresponding program $\COMPOSE{\ASSIGN{i}{x}}{\WHILEDO{i > 0}{\COMPOSE{\WEIGH{a}}{\ASSIGN{i}{i-1}}}}$ requires introducing a fresh loop-variable $i \in \Vars$.}
The symbol $\WeighSymbol$ used in the weight-statement is reminiscent of the corresponding monoid operation in $\Monoid$.
The symbol $\BranchSymbol$ we use for the branching-statement will become evident in \cref{sec:weighted-semantics}.

\subsection{Program States}
\label{sec:program-states}
A program state $\sigma$ maps each variable in $\Vars$ to its value in $\Vals$.
To ensure that the set of program states is countable,\footnote{
    We restrict $\States$ a priori to avoid technical issues; even $\wgcl$ programs over uncountable $\States$ reach just countably many states.
} 
we restrict to states in which at most finitely many variables have a non-zero value.
Intuitively, those that appear in a given program are possibly assigned a non-zero value.
Formally, the set $\States$ of program states is given by
\begin{align*}
	\States \qcoloneqq \setcomp{ 
		\sigma\colon \Vars \to \Vals
	}{
		\Set{ x \in \Vars }{ \sigma(x) \neq 0} \textnormal{ is finite}
	} .
\end{align*}
We overload notation and denote by $\eval{\xi}{\sigma}$ the evaluation of the (arithmetic, Boolean, or weight) expression $\xi$ in $\sigma$, i.e. the value obtained from evaluating $\xi$ after replacing every variable $x$ in $\xi$ by~$\sigma(x)$.
We denote by $\sigma\statesubst{x}{v}$ the \emph{update} of variable $x$ by value $v$ in state $\sigma$. Formally:
%
\begin{align*}
	\sigma\statesubst{x}{v} \qqcoloneqq \lam{y} \begin{cases}
		v & \text{if } y = x , \\
		\sigma(y) & \text{otherwise} .
	\end{cases}
\end{align*}%
%

\subsection{Operational Semantics}
\label{sec:operational-semantics}

To formalize our notion of \emph{weighted computation paths}, we define small-step operational semantics \cite{DBLP:journals/jlp/Plotkin04a} in terms of a weighted computation graph, or rather a weighted \emph{computation forest}.\footnote{A path-based semantics in terms of trees is convenient for technical reasons; it allows to distinguish programs like $\SKIP$ from $\BRANCH{\SKIP}{\SKIP}$. The latter has two terminating computation paths with weight $\sone$ and the former has only one.}
Apart from our special weight operation $\WEIGH{a}$, the operational semantics is standard but we include the details for the sake of completeness.
Intuitively, the computation forest of $\Monoid$-$\wgcl$ contains one tree
$T_{C,\sigma}$ per program $C$ and initial state $\sigma$ representing the computation of $C$ on initial state $\sigma$.%
\begin{definition}[Computation Forest of $\Monoid$-$\wgcl$]
    For the monoid $\monoiddef$, the \emph{computation forest} of $\Monoid$-$\wgcl$ is the (countably infinite) directed weighted graph $\compGraph = (Q,\, \Delta)$, where
    \begin{itemize}
        \item $Q = \left(\wgcl \cup \set{\termState} \right) \times \Sigma \times \Nats \times \set{L,R}^*$ is the set of vertices (called \emph{configurations});
        \item $\Delta \subseteq Q \times \monoid \times Q$ is the set of directed weighted edges (called \emph{transitions}) which is defined as the smallest set satisfying the SOS-rules in Fig.~\ref{fig:sosrules}. \qedtriangle
    \end{itemize}
\end{definition}%
\noindent{}%
We use the notation $\sosTrans{\conf_1}{a}{\conf_2}$ instead of $(\conf_1, a, \conf_2) \in \Delta$.
Intuitively, for a configuration \mbox{$\sosState{C}{\state}{n}{\beta}$}, the component $C$ represents the program that still needs to be executed (thus playing the role of a \enquote{program counter}) and $C = \termState$ indicates termination; 
$\state$~is the current program state (variable valuation); $n$ is the number of computation steps that have been executed so far and $\beta$ is the history of left and right branches that have been taken.
Remembering the number of computation steps and the history of left and right branches ensures that $\compGraph$ is indeed a forest.
Moreover, it is easy to check that $\compGraph$ has no multi-edges, i.e. there is at most one weighted edge between any two configurations.



Note that $\compGraph$ is finitely branching and that the only rules that alter the branching history $\beta$ are \sosRule{l.\ branch} and \sosRule{r.\ branch}.
In particular, \sosRule{if} and \sosRule{else} do not change $\beta$ because they are not \emph{truly} branching:
Indeed, all configurations of the form $\sosState{\ITE{\guard}{C_1}{C_2}}{\state}{n{+}1}{\beta}$ have a \emph{unique} successor configuration which depends on whether or not $\state$ satisfies $\guard$.

An \emph{initial} configuration is of the form $\conf = \sosState{C}{\state}{0}{\epsilon}$ for arbitrary $C$ and $\state$.
For initial configurations we also write $\sosStateAbbr{C}{\state}$ instead of $\sosState{C}{\state}{0}{\epsilon}$.
By definition of $\Delta$, initial configurations have no incoming transitions; they are thus the roots of the trees in the computation forest.
Similarly, a configuration $\sosState{\termState}{\state}{n}{\beta}$ is called \emph{final}.
Final configurations are the leaves of the forest.

We can now define computation paths.
For $\conf \in Q$, let $\succConfs(\conf) \coloneqq \Set{ \conf' \in Q }{ \exists a \in \monoid \colon \sosTrans{\conf}{a}{\conf'} }$ be the set of all possible successor configurations.
A \emph{computation path} of length $n \in \Nats$ is a finite path $\compPath = \conf_0 \conf_1 \ldots \conf_n$ in $\compGraph$, i.e. $\conf_{i+1} \in \succConfs(\conf_{i})$ for all $i = 0,\ldots,n-1$, such that $\conf_0$ is initial.
The set of all such paths is denoted $\pathsOfLengthStartingIn{n}{\conf_0}$.
Since $\compGraph$ has no multi-edges, for two configurations $\conf$ and $\conf' \in \succConfs(\conf)$ we denote the unique weight $a \in \monoid$ such that $\sosTrans{\conf}{a}{\conf'}$ by $\weight{\conf}{\conf'}$.
The \emph{weight} of a computation path $\conf_0 \conf_1 \ldots \conf_n$ is then defined as
\[
	\pathWeight(\conf_0 \conf_1 \ldots \conf_n)
	\ccoloneqq \monbigop_{i=0}^{n-1} \weight{\conf_i}{\conf_{i+1}}
	\eeq \weight{\conf_0}{\conf_1} \monop \weight{\conf_1}{\conf_2} \monop \cdots \monop \weight{\conf_{n-1}}{\conf_n} ~.
\]
Recall that our monoids are not commutative in general, and thus the order of the above product matters.
The \emph{last state} of computation path $\compPath = \conf_0 \conf_1 \ldots \conf_n$ is the program state at configuration $\conf_n$ and is denoted $\lastState(\compPath) \in \States$.
The computation path $\compPath$ is called \emph{terminal} if $\conf_n$ is final.
Given an initial configuration $\conf_0$, we define the set of \emph{terminating computation paths} starting in $\conf_0$ as
\[
    \termPathsStartingIn{\conf_0}
    \qeq
    \bigcup_{n \in \Nats} \Set{\compPath \in \pathsOfLengthStartingIn{n}{\conf_0}}{\compPath \text{ is terminal}} ~.
\]
Note that $\termPathsStartingIn{\conf_0}$ is a countable set for each $\conf_0$.
An \emph{infinite} computation path is an infinite sequence $\conf_0 \conf_1 \ldots$ such that $\conf_0 \ldots \conf_i$ is a computation path for all $i \geq 0$.
%

\begin{figure}[t]
    \centering
    \setlength{\jot}{1ex}
	\begin{adjustbox}{max width=0.95\textwidth}
	\parbox{\linewidth}{
		{\scriptsize
        \begin{gather*}
			%
			\frac{
				\state' = \state\statesubst{x}{\sem{E}(\state)}
			}{
				\sosTrans{
					\sosState{\ASSIGN{x}{E}}{\state}{n}{\beta}
				}{\tinysone}{
					\sosState{\termState}{\state'}{n{+}1}{\beta}
				}
			} ~ \sosRule{assign}
			\qquad
			%
			\frac{}{
				\sosTrans{
					\sosState{\WEIGH{a}}{\state}{n}{\beta}
				}{a}{
					\sosState{\termState}{\state}{n{+}1}{\beta}
				}
			} ~ \sosRule{weight}
			\\ 
			%
			\frac{
				\sosTrans{
					\sosState{C_1}{\state}{n}{\beta}
				}{a}{
					\sosState{\termState}{\state'}{n{+}1}{\beta}
				}
			}{
				\sosTrans{
					\sosState{\COMPOSE{C_1}{C_2}}{\state}{n}{\beta}
				}{a}{
					\sosState{C_2}{\state'}{n{+}1}{\beta}
				}
			} ~ \sosRule{seq.\ 1}
			\qquad
			%
			\frac{
				\sosTrans{
					\sosState{C_1}{\state}{n}{\beta}
				}{a}{
					\sosState{C_1'}{\state'}{n{+}1}{\beta}
				}
				\qand
				C_1' \neq \termState
			}{
				\sosTrans{
					\sosState{\COMPOSE{C_1}{C_2}}{\state}{n}{\beta}
				}{a}{
					\sosState{\COMPOSE{C_1'}{C_2}}{\state'}{n{+}1}{\beta}
				}
			} ~ \sosRule{seq.\ 2}
			\\
			%
			\frac{
				\state \models \guard
			}{
				\sosTrans{
					\sosState{\ITE{\guard}{C_1}{C_2}}{\state}{n}{\beta}
				}{\tinysone}{
					\sosState{C_1}{\state}{n{+}1}{\beta}}
				} ~ \sosRule{if}
			\qquad
			%
			\frac{
				\state \models \neg \guard
			}{
				\sosTrans{
					\sosState{\ITE{\guard}{C_1}{C_2}}{\state}{n}{\beta}}{\tinysone}{\sosState{C_2}{\state}{n{+}1}{\beta}}} ~ \sosRule{else} \\
			%
			\frac{}{
				\sosTrans{
					\sosState{\BRANCH{C_1}{C_2}}{\state}{n}{\beta}
				}{\tinysone}{
					\sosState{C_1}{\state}{n{+}1}{\beta L}}
				} ~ \sosRule{l.\ branch}
			\qquad
			%
			\frac{}{
				\sosTrans{
					\sosState{\BRANCH{C_1}{C_2}}{\state}{n}{\beta}
				}{\tinysone}{
					\sosState{C_2}{\state}{n{+}1}{\beta R}
				}
			} ~ \sosRule{r.\ branch}
			\\
			%
			\frac{
				\state \models \guard
			}{
				\sosTrans{
					\sosState{\WHILEDO{\guard}{C}}{\state}{n}{\beta}
				}{\tinysone}{
					\sosState{\COMPOSE{C}{\WHILEDO{\guard}{C}}}{\state}{n{+}1}{\beta}
				}
			} ~ \sosRule{while}
			\qquad
			%
			\frac{
				\state \models \neg \guard
			}{
				\sosTrans{
					\sosState{\WHILEDO{\guard}{C}}{\state}{n}{\beta}
				}{\tinysone}{
					\sosState{\termState}{\state}{n{+}1}{\beta}
				}
			} ~ \sosRule{break}
        \end{gather*}
	}}\end{adjustbox}
    \caption{Structural operational semantics of \wgcl-programs.}
    \label{fig:sosrules}
\end{figure}


\section{Weighting Transformer Semantics}
\label{sec:weighted-semantics}

Throughout this section, we develop a weakest-precondition-style calculus \`a la \citet{DBLP:journals/cacm/Dijkstra75} for reasoning about weighted programs on source code level.
We start with a recap on Dijkstra's weakest preconditions. We then gradually lift weakest preconditions to weakest pre\emph{weightings}.

\subsection{Weakest Preconditions}
\label{sec:wp-classical}

Dijkstra's weakest precondition calculus is based on \emph{predicate transformers}%
\begin{align*}
	\wpC{C}\colon\quad \Predicates \morespace{\to} \Predicates~, \qquad \textnormal{where}\quad \Predicates \eeq \{0,\, 1\}^\States~\text{is the set of \emph{predicates} over $\States$},
\end{align*}%
which associate to each nondeterministic program $C$ a mapping from predicates to predicates.
Somewhat less common, we consider here an \emph{angelic} setting, where the nondeterminism is resolved to our advantage.
Specifically, the \emph{angelic} weakest precondition transformer $\wpC{C}$ maps a \emph{post}{\-}condition~$\psi$ over final states to a \emph{pre}condition $\wp{C}{\psi}$ over initial states, such that executing the program $C$ on an initial state satisfying $\wp{C}{\psi}$ guarantees that $C$~\emph{can}\footnote{Recall that $C$ is a \emph{nondeterministic} program.}~terminate in a final state satisfying~$\psi$, see also \Cref{fig:wp-set}.
More symbolically, if $\sem{C}_{\sigma}$ denotes the set of all final states reachable from executing $C$ on $\sigma$, then
\begin{align*}
	\sigma \mmodels \wp{C}{\psi} \qqimplies \exists\, \tau \in \sem{C}_{\sigma}\colon \quad \tau \mmodels \psi~.
\end{align*}%
While the above is a \emph{set perspective} on $\wpsymbol$, a different, but equivalent, perspective on $\wpsymbol$ is the \emph{map perspective}, see \Cref{fig:wp-map}.
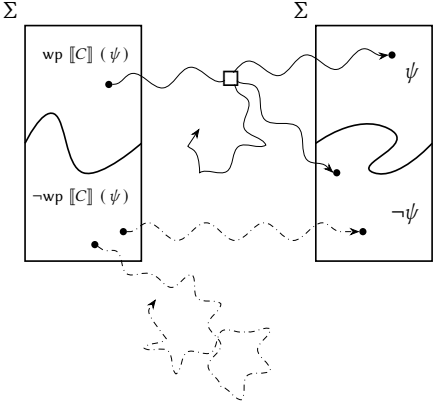
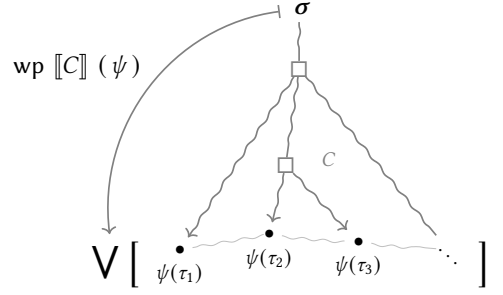
\begin{figure}[t]
	\begin{subfigure}[t]{.45\textwidth}
		\begin{center}
			\begin{adjustbox}{max width=.99\textwidth}
				\begin{tikzpicture}[node distance=0cm,decoration={snake,pre=lineto,pre length=.5mm,post=lineto,post length=.5mm, segment length=10mm, amplitude=1.5mm}]
					\draw[help lines,white] (-4,-2.4) grid (4, 4);
					
					\draw [thick] (-3.5, 0) rectangle (-1.5, 4);
					\coordinate (leftboxtopleft) at (-3.5, 4) {};
					\coordinate (leftboxtopmiddleleft) at (-3.5, 3.5) {};
					\coordinate (leftboxbottommiddleleft) at (-3.5, 0.5) {};
					\node[above left=of leftboxtopleft] (Sigma1) {\Large $\Sigma$};
					\node[right=of leftboxtopmiddleleft] (G) {\footnotesize$\phantom{\neg} \wp{C}{\psi}$};
					\node[right=of leftboxbottommiddleleft,yshift=.6cm,xshift=0.0cm] (nG) {\footnotesize $\neg \wp{C}{\psi}$};
					\draw [thick] plot [smooth,tension=.75] coordinates {(-3.5,2) (-3,2.5) (-2.5,1.5) (-1.5,2)};
					
					\draw [thick] (1.5, 0) rectangle (3.5, 4);
					\coordinate (rightboxtopleft) at (1.5, 4) {};
					\coordinate (rightboxtopmiddleright) at (3.5, 3.5) {};
					\coordinate (rightboxbottommiddleright) at (3.5, 0.5) {};
					\node[above left=of rightboxtopleft] (Sigma2) {\Large $\Sigma$};
					\node[left=of rightboxtopmiddleright,yshift=-.3cm,xshift=-.1cm] (G) {\large $\phantom{\neg} \psi$};
					\node[left=of rightboxbottommiddleright,yshift=.3cm,xshift=-.1cm] (nG) {\large $\neg \psi$};
					\draw [thick] plot [smooth,tension=2] coordinates {(1.5,2) (2.75,2.25) (2.5,1.5) (3.5,2)};
					
					\tikzcomputationarrow{-2}{3}{2.75}{3.5}
					\tikzboxcomputationarrow{0.1}{3.1}{1.8}{1.5}
					
					
					
					\path [draw,-{Stealth[scale=1]}, decoration={snake,pre=lineto,pre length=1mm,post=lineto,post length=1mm, segment length=8.05mm, amplitude=1.2mm}, decorate] (0.1,3.0) -- (0.5,2.0) -- (-0.25,1.5) -- (-0.75,2.0) -- (-0.5,2.25);

					
					
					\tikzpossiblecomputationarrow{-1.75}{.5}{2.25}{0.5}
					
					\path [dashdotted,draw,-{Stealth[scale=1]}, decoration={snake,pre=lineto,pre length=1mm,post=lineto,post length=1mm, segment length=7.05mm, amplitude=1.1mm}, decorate] (-2.25,0.25) -- (-1.5,-0.25) -- (-1,-0.25) -- (-0.5,-0.75) -- (-0.5,-2.25) -- (0,-2.25) -- (0.5,-1.75) -- (0.5,-1.25) -- (0.5,-1.25) -- (0,-1) -- (-1.25,-1.75) -- (-1.75,-1.35) -- (-1.25,-0.65); \node (bottomdivergestart) at (-2.25,0.25) {}; \fill (bottomdivergestart) + (-.05,0.03) circle [radius=1.75pt];
				\end{tikzpicture}
			\end{adjustbox}
		\end{center}
		\caption{\textbf{The set perspective:}
			Starting in $\wp{C}{\psi}$, $C$~can terminate in $\psi$. 
			Starting in $\neg \wp{C}{\psi}$, $C$~either diverges or terminates in $\neg \psi$, but it cannot terminate in $\psi$.}
		\label{fig:wp-set}
	\end{subfigure}
	\hfill
	\begin{subfigure}[t]{.45\textwidth}
		\begin{center}
			\begin{adjustbox}{max width=.99\textwidth}
				\begin{tikzpicture}[node distance=4mm, decoration={snake,pre=lineto,pre length=.5mm,post=lineto,post length=1mm, amplitude=.2mm}]
					\draw[use as bounding box,white] (-4.55,-1.5) grid (3.25, 4.75);
					\node (sigma) at (0, 4) {\Large$\boldsymbol{~\sigma}$};
					\node[gray,inner sep=0pt, outer sep=0pt] (firstplus) at (0,3) {\LARGE$\Box$};
					
					\node[gray,inner sep=0pt, outer sep=0pt] (branch) at (-0.225, 1.4) {\LARGE$\Box$};
					
					\node (tau1) at (-2, 0) {$\bullet$};
					\node (tau2) at (-.5, .25) {$\bullet$};
					\node (tau3) at (1, .125) {$\bullet$};
					\node[inner sep=0pt] (taudots) at (2.5, 0) {$\ddots$};
					
					\node[below of=tau1] {$\psi(\tau_1)$};
					\node[below of=tau2] {$\psi(\tau_2)$};
					\node[below of=tau3] {$\psi(\tau_3)$};

					\node (Exp) at (-3.1, -0.25) {\huge $\boldsymbol{\bigvee}$\Large~$\boldsymbol{\Bigl[}$};
					\node at (3.1, -0.25) {\Large $\boldsymbol{\Bigr]}$};

					\node[gray] (C) at (0.5, 1.5) {$C$};
					
					\draw[gray,decorate,thick] (sigma) -- (firstplus);
					\draw[gray,decorate,->,thick] (firstplus) -- (tau1);
					
					\draw[gray,decorate,thick] (firstplus) -- (branch);
					
					\draw[gray,decorate,->,thick] (branch) -- (tau2);
					\draw[gray,decorate,->,thick] (branch) -- (tau3);
					\draw[gray,decorate,thick] (firstplus) -- (taudots);
					
					\draw[decorate,lightgray] (tau1) -- (tau2) -- (tau3) -- (taudots);
					
					\draw[gray] (sigma) edge[|->,bend right=45,above left,thick] node {\textcolor{black}{\Large$\wp{C}{\psi}$}} (Exp);
				\end{tikzpicture}
			\end{adjustbox}
		\end{center}
		\caption{\textbf{The map perspective:}
			Given initial state~$\sigma$, $\wp{C}{\psi}$ determines all final states $\tau_i$ reachable from executing $C$ on $\sigma$, evaluates $\psi$ in each~$\tau_i$, and returns the disjunction over \mbox{these truth values}.}
		\label{fig:wp-map}
	\end{subfigure}
	\caption{Two perspectives on the angelic weakest precondition of program $C$ with respect to postcondition $\psi$.}
	\label{fig:wp-views}
\end{figure}%
%
%
From this perspective, the postcondition is a function $\psi\colon \States \to \{0,\, 1\}$ mapping program states to truth values.
The predicate $\wp{C}{\psi}$ is then a function that takes as input an initial state $\sigma$, determines for each reachable final state $\tau \in \sem{C}_\sigma$ the (truth) value $\psi(\tau)$, and finally returns the disjunction over all these truth values.
More symbolically,%
\begin{align*}
	\wp{C}{\psi}(\sigma) \qqeq\quad \bigvee_{\mathclap{\tau \in \sem{C}_{\sigma}}} \quad \psi(\tau)~.
\end{align*}%
It is this map perspective which we will now gradually lift to a weighted setting.
For that, we first need to leave the realm of Boolean values in which the predicates live.
Instead of acting on Boolean-valued predicates, our calculus will instead act on more general objects called \emph{weightings}.

\subsection{Weightings and Modules}
\label{sec:weightings}

For probabilistic programs, \citet{DBLP:journals/jcss/Kozen85} and later \citet{DBLP:journals/toplas/MorganMS96} have generalized predicates to real-valued functions $f\colon \States \to \PosRealsInf$ (called \emph{expectations}~\cite{DBLP:series/mcs/McIverM05}) associating a quantity to every program state.
With \emph{weightings}, we generalize further by associating a more general \enquote{quantity} to every program state.
%
Our $\wpsymbol$-style calculus acts on these weightings instead of Boolean-valued predicates.
Weightings form --- just like first-order logic for weakest preconditions --- the \emph{assertion \enquote{language} of weakest preweighting reasoning}.

Let us fix a monoid $\Monoid$ of weights.
As with predicates and expectations, we need notions of addition \emph{and} multiplication operations for our weightings.
The monoid $\Monoid$ constituting our programs' weights, however, only provides a multiplication $\monop$.
We hence require that our weightings form a \emph{$\Monoid$-module} $\Smodule$ which \emph{does} provide an addition.
This is inspired by the probabilistic setting where the program weights --- the probabilities --- are taken from the interval $[0,1]$\footnote{Note that probabilities form a monoid under multiplication.}. Expectations, however, map program states to arbitrary extended reals in $\PosRealsInf$ to reason about, e.g.\ expected values of program variables.
Another advantage of distinguishing between $\Monoid$ and $\Smodule$ in general is explained in \cref{ex:semimodules} further below.

We now define modules formally.
(Monoid)-modules are similar to vector spaces over fields in that they also have a well-behaved \emph{scalar multiplication}.
\begin{definition}[Monoid-Modules]
	\label{def:monoid-module}
	\def\smop{\otimes}
	\def\msop{\otimes}
    Let $\monoiddef$ be a monoid.
    A \emph{(left) $\Monoid$-module} $\smoduledef$ is a \emph{commutative monoid} $(\smodule,\, {\madd},\, {\mnull})$ equipped with a (left) action called \emph{scalar multiplication} $\smop\colon \monoid \times \smodule \to \smodule$, such that for all monoid elements $v,w \in \monoid$ and \mbox{module elements $a,b \in \smodule$},
    \begin{enumerate}
        \item
        the scalar multiplication $\smop$ is \emph{associative}, i.e.\ \quad $(v \monop w) \smop a \eeq v \smop (w \smop a)$~,
        
        \item
        the scalar multiplication $\smop$ is \emph{distributive}, i.e.\ \quad
        $
        v \smop (a \madd b) \eeq (v \smop a) \madd (v \smop b)~,
        $
        
        \item $\monone$ is \emph{neutral} w.r.t.\ $\smop$ and $\mnull$ \emph{annihilates}, i.e.\ \quad $\monone \smop a = a$\qand$v \smop \mnull = \mnull$.
        \hfill$\triangle$
    \end{enumerate}%
\end{definition}%
\noindent{}%
We are now in a position to define \emph{weightings}:%
\begin{definition}[Weightings]
    Given a $\Monoid$-module $\Smodule$, a function $f\colon \States \to \smodule$ associating a weight from $\Smodule$ to each program state is called \emph{weighting}.
    We denote the set of all weightings by~$\Weightings$. Elements of $\Weightings$ are denoted by $f, g, h, \dots$ and variations thereof.
    \hfill$\triangle$
\end{definition}%
\noindent{}%
The structure \mbox{$(\Weightings,\, {\sadd},\, {\mnull},\, {\otimes})$}, where $\sadd$, $\mnull$, and $\otimes$ are lifted pointwise, also forms a $\Monoid$-module.
We refer to $\Weightings$ as the \emph{module of weightings} over $\Smodule$.
We emphasize that all the results developed in this paper apply to the important --- and simpler --- special case where the monoid and the module together form a \emph{semiring}:
The multiplication $\smul$ of a semiring $\sringdef$ is then the left-action $\otimes$ of the multiplicative monoid of $(\sring,\, {\smul},\, \sone)$ to the additive monoid $(\sring,\, {\sadd},\, \snull)$.
For this reason, we write $\odot$ instead of $\otimes$ (as both are associative and it should be clear from the rightmost multiplicant's type) and adopt the following convention:%
\begin{convention}
    In all examples in this paper, unless stated otherwise, both the monoid $\Monoid$ and the $\Monoid$-module $\Smodule$ are given in terms of a semiring $\Sring$ that will be clear from the context.
\end{convention}%
%
%
\noindent{}%
Towards our goal of defining a weakest-precondition-style calculus for weighted programs, we restrict to \emph{naturally ordered}\footnote{More generally, partially ordered modules (where the partial order is compatible with the algebraic structure, \eg addition and left-action are monotone) also work.
However, the natural order is the least (w.r.t.\ $\subseteq$) such partial order.
We employ the natural order for simplicity.} \emph{$\omega$-bicontinuous} modules:
\begin{definition}[Natural Order]
	\label{def:natural-order}
	Given a module $\Smodule$, the binary relation ${\natord} \subseteq \smodule\times\smodule$ given by
    \[
		a \nnatord b \qqiff \exists\, c \in \smodule \colonq a \sadd c \eeq b ~
    \]
	is called the \emph{natural order} on $\Smodule$.
    If $\natord$ is a partial order, we call $\Smodule$ \emph{naturally ordered}.
    \hfill$\triangle$
\end{definition}
\noindent{}%
%
%
%
The unique \emph{least element} of a naturally ordered module is $\mnull$.
We say that $\Smodule$ is \emph{$\omega$-bicontinuous} if (1) both the natural order $\natord$ and the reversed natural order $\rnatord$ are (pointed\footnote{We additionally require existence of a \emph{least} element $\bot$.}) $\omega$-cpos \cite[Sec.\ 2.2.4]{DBLP:books/lib/Abramsky94}, and (2) the operations $\madd$ and $\smop$ are $\omega$-continuous \cite[Sec.\ 2.2.4]{DBLP:books/lib/Abramsky94} functions (w.r.t.\ both $\natord$ and $\rnatord$).
In particular, for $\Smodule$ to be $\omega$-bicontinuous, we require the natural order $\natord$ to possess a \emph{greatest} element $\mTop$.
The definition of naturally ordered $\omega$-bicontinuous \emph{semirings} is completely analogous.
All these properties translate to the aforementioned module of weightings:
$\Weightings$ is naturally ordered\footnote{The natural order on $\Weightings$ is the point-wise lifted natural order on $\Smodule$.}
if the underlying module $\Smodule$ is naturally ordered and joins/meets can be defined pointwise.
For example, the Boolean semiring $\booleanSring$ and the tropical semiring $\tropring$ described in \cref{sec:preliminaries} are $\omega$-continuous.
Please confer \cref{app:preliminaries} for details on the above terms.

\begin{example}[Modules]
    \label{ex:semimodules}
    Let $\ab$ be a non-empty alphabet.
    The structure $\mixedWordLangSmodule{\ab} \coloneqq (2^{\ab^\infty},\, {\cup},\, {\emptyset},\, {\langConcat})$ forms a module over the word monoid $\ab^*$.
    Here, $\ab^\infty \coloneqq \ab^* \cup \ab^\omega$ is the set of all finite and $\omega$-words over~$\ab$, and the subsets of $\ab^\infty$ are subsequently called \emph{$\omega$-potent} languages over $\ab$.
    Our interest in $\mixedWordLangSmodule{\ab}$ stems from the fact that we want to study infinite program runs as in \cref{sec:wlp}.
	We stress that this cannot be achieved by simply defining a \emph{semiring} on $2^{\ab^\infty}$. In fact, even though such a semiring can be defined, its multiplication would \emph{not be $\omega$-cocontinuous} (a counterexample is given in \cref{app:proofs2:mixed-formal-languages-cocontinuity-counterexample}).
	On the other hand, $\mixedWordLangSmodule{\ab}$ \emph{does} form an $\omega$-\emph{bi}continuous module (cf. \cref{app:proofs2:mixed-formal-languages-module}).
	\qedtriangle
%
\end{example}

%
%
%
%
%
%

\subsection{Weakest Preweightings}
\label{sec:wp}

We now define a calculus for formal reasoning about weighted programs \`a la Dijkstra.
In reference to Dijkstra's \emph{weakest precondition calculus} and McIver \& Morgan's \emph{weakest preexpectation calculus}, we name our verification system \emph{weakest preweighting calculus}.

First, we notice that predicates just form a specific semiring, namely $(\Predicates,\, {\vee},\, {\wedge},\, 0,\, 1)$  and thus they are in particular \emph{modules} over their underlying \enquote{Boolean monoid} $(\set{0,1},\, {\wedge},\, 1)$.
We refer to this as the \emph{module of predicates}.
With that in mind, we can now generalize the \emph{map perspective} of weakest preconditions to weakest preweightings, see \Cref{fig:wp-general}. 
\begin{figure}[t]
	\begin{subfigure}[t]{.45\textwidth}
		\begin{center}
			\begin{adjustbox}{max width=.99\textwidth}
				\begin{tikzpicture}[node distance=4mm, decoration={snake,pre=lineto,pre length=.5mm,post=lineto,post length=1mm, amplitude=.2mm}]
							\draw[use as bounding box,white] (-4.55,-0.75) grid (3.25, 4.25);
							\node (sigma) at (0, 4) {\Large$\boldsymbol{~\sigma}$};
							\node[gray,inner sep=0pt, outer sep=0pt] (firstplus) at (0,3) {\LARGE$\oplus$};
							
							\node[gray,inner sep=0pt, outer sep=0pt] (branch) at (-0.225, 1.4) {\LARGE$\oplus$};
							
							\node[gray] (tau1) at (-0.8, 2.2) {$a$};
							\node[gray] (tau1) at (-1.2, 1.6) {$a$};
							\node[gray] (tau1) at (-1.6, 1.0) {$b$};
							
							\node[gray] (tau1) at (-0.3, 2) {$b$};
							\node[gray] (tau1) at (-0.55, 0.85) {$a\vphantom{b}$};
							\node[gray] (tau1) at (0.25, 0.65) {$b$};
							
							\node[gray] (tau1) at (0.45, 2.2) {$b$};
							\node[gray] (tau1) at (1.25, 1.2) {$b$};
							\node[gray] (tau1) at (1.75, 0.6) {$b$};

							\node (tau1) at (-2, 0) {$\bullet$};
							\node (tau2) at (-.5, .25) {$\bullet$};
							\node (tau3) at (1, .125) {$\bullet$};
							\node[inner sep=0pt] (taudots) at (2.5, 0) {$\ddots$};
							
							\node[below of=tau1] {$f(\tau_1)$};
							\node[below of=tau2] {$f(\tau_2)$};
							\node[below of=tau3] {$f(\tau_3)$};

							\node (Exp) at (-3.1, -0.25) {\huge $\boldsymbol{\bigoplus}$\Large~$\boldsymbol{\Bigl[}$};
							\node at (3.1, -0.25) {\Large $\boldsymbol{\Bigr]}$};

							\node[gray] (C) at (0.5, 1.5) {$C$};
							
							\draw[gray,decorate,thick] (sigma) -- (firstplus);
							\draw[gray,decorate,->,thick] (firstplus) -- (tau1);
							
							\draw[gray,decorate,thick] (firstplus) -- (branch);
							
							\draw[gray,decorate,->,thick] (branch) -- (tau2);
							\draw[gray,decorate,->,thick] (branch) -- (tau3);
							\draw[gray,decorate,thick] (firstplus) -- (taudots);
							
							\draw[decorate,lightgray] (tau1) -- (tau2) -- (tau3) -- (taudots);
							
							\draw[gray] (sigma) edge[|->,bend right=45,above left,thick] node {\textcolor{black}{\Large$\wp{C}{f}$}} (Exp);
						\end{tikzpicture}
			\end{adjustbox}
		\end{center}
		\caption{\textbf{Weakest preweightings:}
		Given initial state~$\sigma$, $\wp{C}{f}(\sigma)$ determines the weight $w = a_1 a_2 {\cdots}$ of each path starting in $\sigma$ and terminating in some final state $\tau_i$, scalar-multiplies $w$ to the corresponding postweight $f(\tau)$, and returns the module sum over all so-determined weights.}
		\label{fig:wp-general}
	\end{subfigure}
	\hfill
	\begin{subfigure}[t]{.45\textwidth}
		\begin{center}
			\begin{adjustbox}{max width=.99\textwidth}
				\begin{tikzpicture}[node distance=4mm, decoration={snake,pre=lineto,pre length=.5mm,post=lineto,post length=1mm, amplitude=.2mm}]
					\draw[use as bounding box,white] (-4.55,-0.75) grid (3.25, 3.25);
					\node (sigma) at (0, 4) {\Large$\boldsymbol{~\sigma}$};
					\node[gray,inner sep=0pt, outer sep=3pt] (firstplus) at (0,3) {\large$\tfrac{1}{3}$};
					
					\node[gray,inner sep=0pt, outer sep=3pt] (branch) at (-0.225, 1.4) {\large$\tfrac{1}{2}$};
					
					\node (tau1) at (-2, 0) {$\bullet$};
					\node (tau2) at (-.5, .25) {$\bullet$};
					\node (tau3) at (1, .125) {$\bullet$};
					\node[inner sep=0pt] (taudots) at (2.5, 0) {$\ddots$};
					
					\node[below of=tau1] {$f(\tau_1)$};
					\node[below of=tau2] {$f(\tau_2)$};
					\node[below of=tau3] {$f(\tau_3)$};

					\node (Exp) at (-3.1, -0.25) {\Large \textbf{\textsf{Exp}}~$\boldsymbol{\Bigl[}$};
					\node at (3.1, -0.25) {\Large $\boldsymbol{\Bigr]}$};

					\node[gray] (C) at (0.5, 1.5) {$C$};
					
					\draw[gray,decorate,thick] (sigma) -- (firstplus);
					\draw[gray,decorate,->,thick] (firstplus) -- (tau1);
					
					\draw[gray,decorate,thick] (firstplus) -- (branch);
					
					\draw[gray,decorate,->,thick] (branch) -- (tau2);
					\draw[gray,decorate,->,thick] (branch) -- (tau3);
					\draw[gray,decorate,thick] (firstplus) -- (taudots);
					
					\draw[decorate,lightgray] (tau1) -- (tau2) -- (tau3) -- (taudots);
					
					\draw[gray] (sigma) edge[|->,bend right=45,above left,thick] node {\textcolor{black}{\Large$\wp{C}{f}$}} (Exp);
				\end{tikzpicture}
			\end{adjustbox}
		\end{center}
		\caption{\textbf{Weakest preexpectations:}
		Given initial state~$\sigma$, $\wp{C}{f}(\sigma)$ determines the expected value (with respect to the probability distribution generated by executing $C$ on $\sigma$) of $f$ evaluated in the final states reached after executing $C$ on $\sigma$.}
		\label{fig:wp-prob}
	\end{subfigure}
	\caption{The meaning of weakest preweightings generally and weakest preexpectations specifically.}
	\label{fig:wp-weighted}
\end{figure}%
Instead of a postcondition, we now have a post\emph{weighting} $f\colon \States \to \smodule$ mapping program states to elements from our $\Monoid$-module~$\Smodule$.
The weakest preweighting $\wp{C}{f}$ is then a function that takes as input an initial state $\sigma$, determines the weight $w$ of each path starting in $\sigma$ and terminating in some final state $\tau$, 
scalar-multiplies the path's weight $w$ to the corresponding postweight $f(\tau)$ from the module $\Smodule$,
and finally returns the module sum over all these so-determined weights, see \Cref{fig:wp-general}.

\Cref{fig:wp-prob} depicts how the general weighted setting is instantiated to a probabilistic setting: the postweightings become real-valued functions (expectations), the path weights become the paths' probabilities, and the summation remains a summation, thus obtaining an expected value.

One of the main advantages of Dijkstra's calculus is that the weakest preconditions can be defined by induction on the program structure, thus allowing for \emph{compositional reasoning}.
Indeed, the same applies to our weighted setting.
In the following, we fix an ambient monoid $\Monoid$ of programs weights and an $\omega$-bicontinuous $\Monoid$-module $\Smodule$ that constitutes the habitat of our weightings $\Weightings$.
We now go over each construct of $\wgcl$ and see how a weakest preweighting semantics can be developed and understood analogously to Dijkstra's weakest preconditions.

\subsubsection*{\textbf{Assignment}}
The weakest precondition of an assignment is given by%
\begin{align*}
	\wp{\ASSIGN{x}{E}}{\psi} \eeq \psi\subst{x}{E}~,
\end{align*}%
where $\psi\subst{x}{E}$ is the replacement of every occurrence of variable $x$ in the postcondition $\psi$ by the expression $E$.
For weakest preweightings, we proceed analogously.
That is, we \enquote{replace} every \enquote{occurrence} of $x$ in $f$ by $E$. Since $f$ is actually not a syntactic object, we more formally define
\begin{align*}
	\wp{\ASSIGN{x}{E}}{f} \eeq f\subst{x}{E} \morespace{\coloneqq} \lam{\sigma} f\Bigl( \sigma\statesubst{x}{\eval{E}{\sigma}} \Bigr)~.
\end{align*}%
So the weighting $f$ of the final state reached after executing the assignment $\ASSIGN{x}{E}$ is precisely $f$ evaluated at the state $\sigma\statesubst{x}{\eval{E}{\sigma}}$ --- the state obtained from $\sigma$ by updating variable $x$ to $\eval{E}{\sigma}$.

\subsubsection*{\textbf{Weighting}}

Consider the classical statement $\ASSERT{\varphi}$.
Operationally, when executing $\ASSERT{\varphi}$ on some initial state $\sigma$, we check whether $\sigma$ satisfies the predicate $\varphi$.
If $\sigma \models \varphi$, the execution trace \enquote{passes through} the assertion and potentially proceeds with whatever program comes after the assertion.
If, however, $\sigma \not\models \varphi$, then the execution trace at hand is so-to-speak \enquote{annihilated}.
Intuitively, these two cases can be thought of as multiplying (or weighting) the execution trace either by a multiplicative identity \emph{one} or by an annihilating \emph{zero}, respectively.

Denotationally, the weakest precondition of $\ASSERT{\varphi}$ is given by%
\begin{align*}
	\wp{\ASSERT{\varphi}}{\psi} \eeq \varphi \wedge \psi~.	
\end{align*}%
Indeed, whenever an initial state $\sigma$ satisfies the precondition $\varphi \wedge \psi$, then (a) executing $\ASSERT{\varphi}$ will pass through asserting $\varphi$ and moreover --- since the assertion itself does not alter the current program state --- (b) it terminates in state $\sigma$ which also satisfies the postcondition $\psi$.
Dually, if $\sigma$ does not satisfy $\varphi \wedge \psi$, then either (a) executing $\ASSERT{\varphi}$ does not pass through asserting $\varphi$ or (b)~it does pass through the assertion but $\sigma$ does not satisfy the postcondition $\psi$.

When viewing the above through our monoid and module glasses, $\wedge$ is just the scalar-multipli-cation in the module of predicates.
So in other words, $\wp{\ASSERT{\varphi}}{\psi}$ weights (multiplies) $\psi$ with~$\varphi$.
Therefore, we generalize from \emph{conjunction with a predicate} to (scalar-)\emph{multiplication with a monoid element} $a$ from $\Monoid$ and introduce the statement $\WEIGH{a}$ into the programming language.
Operationally, our execution traces are weighted and the $\WEIGH{a}$ statement scalar-multiplies the current execution trace's weight by an $a$.
Denotationally, the weakest preweighting of $\WEIGH{a}$ is given by%
\begin{align*}
	\wp{\WEIGH{a}}{f} \eeq a \smop f~.
\end{align*}%
\begin{remark}[On Non-commutativity and Notation] Recall that multiplication of weights is generally not commutative --- think, for example, about the word monoid $\ab^*$.
In the light of potential non-commutativity, the flipping of the sides, i.e.\ $\WEIGH{a}$ in the program syntax \emph{versus} $a \smul \cloze{f}$ in the denotational weighting transformer semantics, is on purpose:
Programs are usually read (and executed) in a forward manner.
Assuming that the weights along an execution trace are collected from left to right, from initial to final state, weighting by $a$ is a \emph{right-multiplication}, appending at the end of the current execution trace the weight $a$.

Weakest preweightings, on the other hand, are backward-moving:
The $f$ is a \emph{post}weighting that potentially abstracts or summarizes the effects of subsequent computations. 
Whenever we encounter on our way from the back to the front of a program a weighting by $a$, we thus have to \emph{pre}pend $a$ to the current postweighting $f$, yielding a \emph{left-multiplication} $a \smop f$ in the denotations.
\qedtriangle
\end{remark}

\subsubsection*{\textbf{Branching}}

We now consider the classical angelic nondeterministic choice $\NDCHOICE{C_1}{C_2}$.
Operationally, when \enquote{executing} this choice on some initial state $\sigma$, \emph{either} the program $C_1$ \emph{or} the program~$C_2$ will be executed, chosen nondeterministically.
Hence, the execution will reach either a final state in which executing~$C_1$~on~$\sigma$ terminates or a final state in which executing~$C_2$~on~$\sigma$ terminates (or no final state if both computations diverge).

Denotationally, the \emph{angelic} weakest precondition of $\NDCHOICE{C_1}{C_2}$ is given by%
\begin{align*}
	\wp{\NDCHOICE{C_1}{C_2}}{\psi} \eeq \wp{C_1}{\psi} \morespace{\vee} \wp{C_2}{\psi}~.
\end{align*}%
Indeed, whenever an initial state $\sigma$ satisfies the precondition $\wp{C_1}{\psi} \vee \wp{C_2}{\psi}$ then executing $C_1$ or executing $C_2$ will terminate in some final state satisfying the postcondition $\psi$.

Again viewed through our module glasses, $\vee$ is just the addition of the module of predicates.
So in other words, $\wp{\NDCHOICE{C_1}{C_2}}{\psi}$ unions (adds) $\wp{C_1}{\psi}$ and $\wp{C_2}{\psi}$.
We thus generalize from \emph{disjunction of two predicates} to \emph{addition of two module elements} and introduce the statement $\BRANCH{C_1}{C_2}$ into the programming language.
Operationally, we have the same interpretation as in the classical case:
Either the program $C_1$ can be executed or the program $C_2$.
Denotationally, the weakest preweighting of $\BRANCH{C_1}{C_2}$ is given by%
\begin{align*}
	\wp{\BRANCH{C_1}{C_2}}{f} \eeq \wp{C_1}{f} \mmadd \wp{C_2}{f}~.
\end{align*}%
$\wp{C_i}{f}$ tells us what element we obtain if $C_i$ is executed, and the module addition $\madd$ tells us how to account for the fact that either $C_1$ or $C_2$ could have been executed.

\subsubsection*{\textbf{Conditional Choice}}

We now consider the classical conditional choice $\ITE{\varphi}{C_1}{C_2}$.
Operationally, when executing $\ITE{\varphi}{C_1}{C_2}$ on some initial state $\sigma$, we check whether $\sigma$ satisfies the predicate $\varphi$.
If $\sigma \models \varphi$, the program $C_1$ is executed; otherwise the program~$C_2$.


Denotationally, the weakest precondition of $\ITE{\varphi}{C_1}{C_2}$ (and in fact also the weakest precondition of $\NDCHOICE{\COMPOSE{\ASSERT{\varphi}}{C_1}}{\COMPOSE{\ASSERT{\neg\varphi}}{C_2}}$) is given by%
\begin{align*}
	\wp{\ITE{\varphi}{C_1}{C_2}}{\psi} \eeq \varphi\wedge\wp{C_1}{\psi} \morespace{\vee} \neg\varphi\wedge\wp{C_2}{\psi}~.
\end{align*}%
Indeed, whenever an initial state $\sigma$ satisfies the above precondition then either $\sigma \models \varphi$ and then \mbox{--- since} then $\sigma$ must also satisfy $\wp{C_1}{\psi}$ --- executing $C_1$ will terminate in a final state satisfying~$\varphi$, or $\sigma \not\models \varphi$ and --- since then $\sigma$ must also satisfy $\wp{C_2}{\psi}$ --- executing $C_2$ will terminate in a final state satisfying $\varphi$.

In terms of monoids and modules, $\varphi\wedge \cloze{\psi}$ \emph{could} be viewed as a scalar-multiplication by either $\monone$ (leaving the right operand unaltered) or by $\mnull$ (annihilating the right operand). However, general monoids do not posses an annihilating $\mnull$.
%
%
In order to reenact the desired behavior, we introduce the \emph{Iverson bracket} $\iverson{\guard}$ of a predicate $\guard$, which for a weighting $f \in \Weightings$ defines the weighting%
\[
	(\iverson{\phi} \ivop f)(\sigma) \qeq
	\begin{cases}
		f(\sigma) & \textnormal{if } \sigma \models \guard ~,\\
		\mnull & \textnormal{otherwise} ~.
	\end{cases}
\]
With this notation at hand, we define the weakest preweighting of $\ITE{\varphi}{C_1}{C_2}$ by%
\begin{align*}
	\wp{\ITE{\varphi}{C_1}{C_2}}{f} \eeq \iverson{\varphi} \ivop \wp{C_1}{f} \ssadd \iverson{\neg\varphi} \ivop \wp{C_2}{f}~.
\end{align*}%
%
By convention, $\iverson{\phi} \ivop$ binds stronger than $\madd$.
Depending on the truth value of $\varphi$, the above weakest preweighting thus \emph{selects} either the preweighting $\wp{C_1}{f}$ or the preweighting $\wp{C_2}{f}$.

\subsubsection*{\textbf{Sequential Composition}}
Our composite statement $\COMPOSE{C_1}{C_2}$ is standard.
Operationally, $C_1$ is executed first and then --- provided that $C_1$ terminates --- $C_2$ is executed.
A distinguishing feature of the classical weakest precondition transformer is that it moves \emph{backwards} through the program, and the same applies to our weighting transformer, i.e.
\[
    \wp{\COMPOSE{C_1}{C_2}}{f} \eeq \wp{C_1}{\wp{C_2}{f}} ~.
\]
Indeed, to compute the weakest preweighting of the composition $\COMPOSE{C_1}{C_2}$ w.r.t.\ to some $f \in \Weightings$, we first compute an intermediate weighting $\wp{C_2}{f}$, which we then feed into $\wpC{C_1}$.

\subsubsection*{\textbf{Looping}}
Operationally, a loop $\WHILEDO{\guard}{C}$ is equivalent to the infinite nested conditional 
\[
    \ITE{\guard}{\COMPOSE{C}{\ITE{\guard}{\COMPOSE{C}{\ITE{\guard}{\COMPOSE{C}{\ldots}}{\SKIP}}}{\SKIP}}}{\SKIP}~,
\]
which is the same as saying that $\WHILEDO{\guard}{C} \equiv \ITE{\guard}{\COMPOSE{C}{\WHILEDO{\guard}{C}}}{\SKIP}$.
With the rules for conditional choice and composition as explained above, it is thus reasonable to require that the preweighting $\wp{\WHILEDO{\guard}{C}}{f}$ should be a \emph{fixed point} of the function
\[
    X \morespace{\mapsto} \iverson{\guard} \ivop \wp{C}{X} \mmadd \iverson{\neg \guard} \ivop f
\]
which is indeed just the $\wpsymbol$-characteristic function defined above.
For both the classical weakest precondition transformer as well as for our weighted $\wpsymbol$, we choose the semantics to be the \emph{least} fixed point, which  exists uniquely if the ambient module $\Smodule$ \emph{is $\omega$-continuous} (see \cref{thm:wp-continuous} below).
In the classical Boolean setting, this corresponds to choosing the \emph{strongest} (least) possible predicate that satisfies the fixed point equation.
This ensures that the weakest precondition contains only those initial states where the loop can actually \emph{terminate} in a state satifying the postcondition --- but no such states for which the loop cannot terminate at all.
Taking the least fixed point in the weighted setting generalizes this intuition as we will show in \cref{thm:wlp-operational-semantics} below.

\subsubsection*{\textbf{Properties of $\wpsymbol$}}
Based on the above discussion, we now define $\wpsymbol$ formally, state healthiness and soundness properties, and provide several examples.%
\begin{table}
	\centering%
	\caption{Rules defining the weakest preweighting $\wp{C}{f}$ of program $C$ w.r.t.\ postweighting $f$.}%
	\label{tab:wp-semantics}%
	\renewcommand{\arraystretch}{1.2}%
	\begin{tabular}{l@{\qquad} l}
		$C$ & $\wp{C}{f}$ \\
		\hline
		%
		$\ASSIGN{x}{E}$ &
		$f\subst{x}{E}$ \\
		$\COMPOSE{C_1}{C_2}$ &
		$\wp{C_1}{\wp{C_2}{f}}$ \\
		$\ITE{\guard}{C_1}{C_2}$ &
		$\iverson{\guard} \ivop \wp{C_1}{f} \mmadd \iverson{\neg\guard} \ivop \wp{C_2}{f}$ \\
		$\BRANCH{C_1}{C_2}$ &
		$\wp{C_1}{f} \mmadd \wp{C_2}{f}$ \\
		$\WEIGH{a}$ &
		$a \smop f$ \\
		$\WHILEDO{\guard}{C'}$ &
		$\lfp{X} \iverson{\neg \guard} \ivop f \ssadd \iverson{\guard} \ivop \wp{C'}{X}$
	\end{tabular}%
\end{table}%
\begin{definition}[Weakest Preweighting Transformer]
    The transformer $\wpsymbol\colon \wgcl \to (\Weightings \to \Weightings)$ is defined by induction on the structure of $\wgcl$ according to the rules in \cref{tab:wp-semantics}.
    The function%
    \begin{align*}
    \charwp{\guard}{C}{f} \colon\quad \Weightings \to \Weightings, \quad X \morespace{\mapsto} \iverson{\neg \guard} \ivop f \ssadd \iverson{\guard} \ivop \wp{C'}{X} 
    \end{align*}%
    whose least fixed point defines the weakest preweighting of $\WHILEDO{\guard}{C}$ is called the \emph{$\wpsymbol$-characteristic function} of $\WHILEDO{\guard}{C}$ with respect to postweighting $f$.\qedtriangle
\end{definition}%
\noindent{}%
\begin{theorem}[Well-Definedness of $\wpsymbol$]
	\label{thm:wp-continuous}
	Let the monoid module $\Smodule$ over $\Monoid$ be $\omega$-continuous.
    For all $\Monoid$-$\wgcl$ programs $C$, the weighting transformer $\wpC{C}$ is a well-defined $\omega$-continuous endofunction on the module of weightings over $\Smodule$.
    In particular, if \,$\charwp{}{}{f}$ is the $\wpsymbol$-characteristic function of $\WHILEDO{\guard}{C}$ with respect to postweighting $f$, then
    \[
		\wp{\WHILEDO{\guard}{C}}{f} \qeq \bigjoin_{i \in \Nats} \charwp{}{}{f}^i(\snull) ~.
    \]
\end{theorem}%
\noindent{}%
Our $\wpsymbol$ satisfies the following so-called \emph{healthiness criteria} (see e.g.~\cite{DBLP:journals/jacm/Hoare78,DBLP:conf/lics/HinoKH016,DBLP:journals/entcs/Keimel15,DBLP:series/mcs/McIverM05}) or homomorphism properties \cite{DBLP:books/daglib/0096285}:%
\begin{theorem}[Healthiness]
    \label{thm:wp-properties}
	Let the monoid module $\Smodule$ over $\Monoid$ be $\omega$-continuous.
    For all $\Monoid$-$\wgcl$ programs $C$, the $\wpsymbol$ transformer is 
    \begin{enumerate}
        \item \emph{monotone}, i.e.\ for all  $f, g \in \Weightings$,\qquad  $f\nnatord g$ \qimplies 
        $
        \wp{C}{f} \nnatord \wp{C}{g}
        $,
        \item \emph{strict}, i.e.\ \qquad
        $
        \wp{C}{\snull} \eeq \snull 
        $,
        \item \emph{additive}, i.e.\ for all $f, g \in \Weightings$, \qquad
        $
            \wp{C}{f \mmadd g} \eeq \wp{C}{f} \madd \wp{C}{g}
        $.
        \item Moreover, if $\Monoid$ is commutative, then $\wpsymbol$ is \emph{homogeneous}, i.e.\ for all $a \in \monoid$ and $f\in \Weightings$,
        \[
            \wp{C}{a \smop f} \eeq a \smop \wp{C}{f} ~,
        \]%
        and together with \textnormal{(3)}, $\wpsymbol$ then becomes \emph{linear}.
    \end{enumerate}
\end{theorem}
\noindent
Homogeneity does not hold in general: Consider the formal languages semiring $\finWordLangSring{\set{a, b}}$ and the program $C = \WEIGH{\set{a}}$ with the constant postweighting $f = \sone = \set{\epsilon}$.
Then%
\[
    \wp{C}{\set{b} \langConcat f} \eeq \set{a} \langConcat \set{b} \eeq \set{ab} \morespace{\neq} \set{ba} \eeq \set{b} \langConcat \set{a} \eeq \set{b} \langConcat \wp{C}{f} ~.
\]%
The next theorem states that $\wpsymbol$ indeed generalizes the \emph{map perspective} on classical weakest preconditions as anticipated at the beginning of \cref{sec:wp}.
The operational semantics as well as $\termPathsStartingIn{\sosStateAbbr{C}{\state}}$, $\pathWeight$, and $\lastState$ are defined in \cref{sec:operational-semantics}.%
\begin{theorem}[Soundness of $\wpsymbol$]
    \label{thm:wp-operational-semantics}
	Let the monoid module $\Smodule$ over $\Monoid$ be $\omega$-continuous.
    For all $C \in \wgcl$, $\state \in \States$ and $f \in \Weightings$, 
    \begin{align}
		\label{eq:wp-operational-semantics}
		\wp{C}{f}(\state)
		\qeq \mbigadd_{\compPath \in \termPathsStartingIn{\sosStateAbbr{C}{\state}}}
			\pathWeight(\compPath) \smop f(\lastState(\compPath))
		~.
    \end{align}
\end{theorem}%
\noindent{}%
%
%
\subsubsection*{\textbf{$\wpsymbol$-Annotations}}
In the spirit of Hoare-style reasoning, we will annotate programs as is shown abstractly in \cref{fig:wp-annotations} and concretely in \cref{fig:ex:wp-calc-a}.
Read the annotations \emph{from bottom to top} as follows:
\begin{enumerate}[align=left]
\item[(1)\hspace{.93em}$\annotate{f}$]
	This first annotation states that we start our reasoning from postweighting $f \in \Weightings$.
\item[(2)\hspace{1.075em}$\wpannotate{g}$]
	The superscript $\wpsymbol$ before the annotation indicates that this annotation is obtained from applying $\wp{\cloze{C}}{\cloze{f}}$.
	The program passed into $\wpsymbol$ is the line immediately below this annotation \mbox{--- in} this \mbox{case $C$ ---} and the continuation passed into $\wpsymbol$ is the annotation immediately below the program --- in this case $f$.
	Hence, this annotation states $g = \wp{C}{f}$.
\item[(3)\hspace{.75em}$\bowtieannotate{g'}$] 
	This last annotation states that $g \bowtie g'$, for ${\bowtie} \in \{{\preceq},{=},{\succeq}\}$.
	We thus allow rewriting, or (like the classical rule of \emph{consequence} in Hoare logic) to perform a monotonic relaxation.
\end{enumerate}%
Let us illustrate $\wpsymbol$ by means of two examples.
Recall our convention that the monoid $\Monoid$ and the module $\Smodule$ stem from a semiring $\Sring$ unless stated otherwise.%
\begin{example}
    \label{ex:wp-tropical}
	Let $\Sring = \tropringdef$ be the \emph{tropical} semiring and consider a $\wgcl$-program $C$.
    \cref{thm:wp-operational-semantics} implies that for all states $\state \in \States$,
    \[
        \wp{C}{0}(\state) \qeq \text{ \enquote{minimum weight of all terminating computation paths that start in $\state$}}
    \]
    where the weight of a path is the \emph{usual sum} of all weights along that path.
    Notice that the above \enquote{0} is the map to the \emph{natural number} $0$ and \emph{not} the semiring $\snull = \infty$.
	For instance, we can verify that
    \[
        \wp{\quad \ITE{x > 0}{\COMPOSE{\WEIGH{1}}{\WEIGH{1}}}{\BRANCH{\WEIGH{2}}{\WEIGH{3}}}\quad  }{0} \qeq 2 ~.
    \]
	For that, consider the program annotations in \cref{fig:ex:wp-calc-a}, which express that $2 = \iverson{x>0} \ivop (1 + 1) \mmadd \iverson{x = 0} \ivop (2 \min 3) = \wp{C}{0}$ where $C = \texttt{if}\: (x > 0)\: \{{\ldots}\}$.
    This reflects that if $x > 0$ initially, then the only possible terminating path has weight $1+1 = 2$.
	Otherwise, i.e.\ if $x = 0$, then the minimum weight of the two possible paths is also $2$.\qedtriangle
\end{example}%
%
%
%
%
\begin{figure}[t]
	\centering
    \begin{subfigure}[t]{0.4\linewidth}
		\abovedisplayskip=-0em%
		\belowdisplayskip=0pt%
		\begin{align*}
			&\bowtieannotate{g'} \tag{meaning $g \bowtie g'$}\\
			&\wpannotate{g} \tag{meaning $g = \wp{C}{f}$}\\
			&C \\
			&\annotate{f} \tag{postweighting is $f$}
		\end{align*}%
		\normalsize%
        \subcaption{Style for $\wpsymbol$ annotations. ${\bowtie} \in \{{\preceq},{=},{\succeq}\}$.}
        \label{fig:wp-annotations}
    \end{subfigure}
    \hfill
    \begin{subfigure}[t]{0.55\linewidth}
		\abovedisplayskip=0em%
		\belowdisplayskip=0pt%
		\begin{align*}
			& \eqannotate{2 \quad \gray{\eeq g'}} \\
			& \wpannotate{\iverson{x>0}\ivop (1 + 1) \mmadd  \iverson{x = 0} \ivop (2 \min 3) \quad \gray{\eeq g}} \\
			& \ITE{x > 0}{\COMPOSE{\WEIGH{1}}{\WEIGH{1}}}{\BRANCH{\WEIGH{2}}{\WEIGH{3}}}\\
			& \annotate{0 \quad \gray{\eeq f}}
		\end{align*}%
		\normalsize%
		\caption{Annotations for \cref{ex:wp-tropical}.}
        \label{fig:ex:wp-calc-a}
    \end{subfigure}
    \caption{Annotations for weakest preweightings. It is more intuitive to read these from the bottom to top.}
    \label{fig:annotationzz}
\end{figure}%
\begin{example}
    \label{ex:wp-lang}
    Let $\Sring = \finWordLangSringDef{\ab}$ be the semiring of formal languages over $\ab$.
	Similarly to \cref{ex:semimodules}, we now choose the monoid $\Monoid = \ab^*$ of words and view $\finWordLangSring{\ab}$ as a $\ab^*$-module, i.e.\ the weighting-statements are of the form $\WEIGH{w}$ for some single \emph{word} $w \in \finWordMonoid{\ab}$.
	The weightings~$\Weightings$, however, associate an entire \emph{language} to each state.
    For all initial states $\state \in \States$, we have
    \[
        \wp{C}{\set{\epsilon}}(\state) \qeq \text{ \enquote{language of all terminating computation paths starting in $\state$}}~,
    \]
    where each terminating path contributes the single word obtained from concatenating all symbols occurring in the weight-statements along this path (this may also yield the empty word $\epsilon = \monone$).\qedtriangle
\end{example}%

\subsection{Weakest Liberal Preweightings}
\label{sec:wlp}

\newcommand{\mayNotTerminate}[2]{\sem{#1}_{#2}^\Uparrow}

The weakest preweighting calculus developed in the previous section assigns a weight to each initial state $\state$ based on the \emph{terminating} computation paths starting in $\state$ and the postweighting~$f$.
In particular, $\wpsymbol$ ignores (more precisely: assigns weight $\mnull$ to) nonterminating behavior, i.e.\ the preweighting $\wp{C}{f}$ is independend of the \emph{infinite} computation paths of $C$.
For instance, in the formal languages semiring
$\finWordLangSring{\set{a,b}}$ with monoid $\Monoid = \set{a,b}^*$, we have for all $f \in \Weightings$ that
\[
    \wp{\WHILEDO{\true}{\WEIGH{a}}}{f} \qeq \emptyset \qeq \wp{\WHILEDO{\true}{\WEIGH{b}}}{f} ~,
\]
even though the computation trees of the two programs are clearly distinguishable.

In this section, we define weakest \emph{liberal} preweightings ($\wlpsymbol$) as a means to reason about such infinite, i.e.\ nonterminating, program behaviors, thus generalizing Dijkstra's classical weakest liberal preconditions.
Unlike Dijkstra's weakest liberal preconditions who just assign $\true$ (instead of $\false$) to \emph{any} nonterminating behavior, our weakest liberal preweightings can inspect nonterminating behavior more nuancedly.
As a teaser:
our weakest liberal preweightings \emph{can} distinguish between $\WHILEDO{\true}{\WEIGH{a}}$ and $\WHILEDO{\true}{\WEIGH{b}}$ as we will demonstrate below.

Reconsidering the \emph{map perspective} on weakest preconditions explained in \cref{sec:wp}, the weakest liberal precondition of program $C$ with respect to postcondition $\psi$ maps an initial state $\state$ to $\true$ iff
(i) $C$ started on $\state$ \emph{can} terminate in a state $\tau$ satisfying $\psi$, or (ii) it is \emph{possible} that $C$ does not terminate at all, or both.
In more symbolic terms,
\[
    \wlp{C}{\psi}(\state)
    \qeq
    \bigvee_{\mathclap{\tau \in \sem{C}_{\state}}} \psi(\tau)
    ~\lor~
    \mayNotTerminate{C}{\state}
    ~,
\]
where $\mayNotTerminate{C}{\state}$ holds iff the nondeterministic program $C$ \emph{may not terminate} on $\state$.\footnote{Recall that we consider angelic nondeterminism.} 
We have
\begin{align}
    \label{eq:classical-wlp}
    \wlp{C}{\psi}(\state) \eeq \wp{C}{\psi}(\state) \morespace{\lor} \mayNotTerminate{C}{\state}
    \qqand
    \wlp{C}{\false}(\state) \eeq \mayNotTerminate{C}{\state}
    ~,
\end{align}
implying that $\wlp{C}{\false}$ captures precisely the nonterminating behavior of $C$, and hence $\wlp{C}{\false}$ characterizes precisely the difference between $\wp{C}{f}$ and $\wlp{C}{f}$.

In the realm of monoids and modules, the predicate $\false$ is the zero $\snull$ of the Boolean semiring.
We now define a weakest liberal pre\emph{weighting} calculus generalizing \eqref{eq:classical-wlp} by satisfying
\[
    \forall\, f \in \Weightings \colonq\quad \wlp{C}{f} \qeq \wp{C}{f} \mmadd \wlp{C}{\mnull} ~.
\]
%
Intuitively, $\wlp{C}{\mnull}$ captures the weights of the nonterminating paths in $C$:
For the two example programs from the beginning of this subsection considered over the $\ab^*$-module $\mixedWordLangSmodule{\set{a,b}}$ of $\omega$-potent formal languages, we get for example
\[
   \wlp{\WHILEDO{\true}{\WEIGH{a}}}{\snull} \eeq\set{a^\omega} \qand \wlp{\WHILEDO{\true}{\WEIGH{b}}}{\snull} \eeq \set{b^\omega}~.
\]
%
%
%
\begin{definition}[Weakest Liberal Preweighting Transformer]
    The transformer $\wlpsymbol\colon \wgcl \to (\Weightings \to \Weightings)$ is inductively defined according to \cref{tab:wp-semantics} with $\lfpop$ replaced by $\gfpop$ and with every occurrence of $\wpsymbol$ replaced by $\wlpsymbol$.
    In particular, $\wlp{\WHILEDO{\guard}{C'}}{f}$ is defined as the \emph{greatest} fixed point of the \emph{characteristic function}
    \[
        \charwlp{\guard}{C}{f} \colon \Weightings \to \Weightings, \quad X \quad \mapsto \quad \iverson{\neg \guard} \ivop f \mmadd \iverson{\guard} \ivop \wlp{C'}{X} ~. \tag*{\qedtriangle}
    \]
\end{definition}%
\noindent{}%
%
%
%
%
We obtain a well-definedness result analogous to \cref{thm:wp-continuous}:%
\begin{theorem}[Well-Definedness of $\wlpsymbol$]
    \label{thm:wlp-continuous}
	Let $\Smodule$ be an $\omega$-\emph{co}continuous $\Monoid$-module.
    For all $\Monoid$-$\wgcl$ programs $C$, the transformer $\wlpC{C}$ is a well-defined $\omega$-\emph{co}continuous endofunction on the module of weightings over $\Smodule$.
    In particular, if \:$\charwlp{}{}{f}$ is the $\wlpsymbol$-characteristic function of $\WHILEDO{\guard}{C}$ with respect to postweighting $f$, then
    \[
        \wlp{\WHILEDO{\guard}{C}}{f} \qeq \bigmeet_{i \in \Nats} \charwlp{}{}{f}^i(\top) ~.
    \]
\end{theorem}%
%
\noindent{}%
As stated above, we furthermore get the following fundamental property:%
%
\begin{theorem}[Decomposition of $\wlpsymbol$]
    \label{thm:wlp-decomp}
	Let $\Smodule$ be an $\omega$-\emph{bi}continuous $\Monoid$-module.
    Then for all programs~$C$ and postweightings $f$, 
    \[
        \wlp{C}{f} \qeq \wp{C}{f} \mmadd \wlp{C}{\mnull}~.
     \]
\end{theorem}%
%
%
\noindent{}%
Moreover, we get a statement relating (infinite) computation paths and $\wlp{C}{\mnull}$:%
%
\begin{theorem}[Soundness of $\wlpsymbol$]
    \label{thm:wlp-operational-semantics}
	Let the monoid module $\Smodule$ over $\Monoid$ be $\omega$-\emph{bi}continuous\footnote{In the statement, we assert $\omega$-\emph{bi}continuity: our proof makes heavy use of \cref{thm:wlp-decomp} and thus we need $\wpsymbol$ to be well-defined. However, it might be possible to prove a link between operational semantics and $\wlpsymbol$ assuming only $\omega$-\emph{co}continuity. But a proof seems much more convoluted than our current one.}.
    Then for all programs $C$ and initial states $\state$, 
    \begin{align}
		\label{eq:wlp-operational-semantics}
		\wlp{C}{\mnull}(\state) 
		\qeq \bigmeet_{n \in \Nats} \,
			\sbigadd_{\compPath \in \pathsOfLengthStartingIn{n}{\sosStateAbbr{C}{\state}} }
				\pathWeight(\compPath) \smop \sTop
		~.
    \end{align}
\end{theorem}%
\noindent{}%
Note that \Cref{thm:wlp-operational-semantics} is phrased in terms of the \emph{finite} computation paths $\pathsOfLengthStartingIn{n}{\sosStateAbbr{C}{\state}}$.
This is because it is somewhat difficult to define a general \emph{infinite product} in $\Monoid$ that is compliant with the way our $\wlpsymbol$ assigns weights to infinite computation paths.
Nevertheless, the right-hand side of \eqref{eq:wlp-operational-semantics} depends \emph{only} on the infinite paths which can be seen intuitively as follows:
For arbitrary $n \in \Nats$ consider the \emph{finite} (not necessarily terminating) computation paths up to length $n$.
Let $v(n)$ denote the sum their weights, where the weight of each path is additionally multiplied by $\sTop$, the top element of $\Smodule$.
Then $(v(n))_{n \in \Nats}$ is a \emph{decreasing} chain in the module $\Smodule$.
In the limit (i.e. infimum), all \emph{terminating} computation paths will be ruled out as each of them has some finite length.
The limit/infimum of the $v(n)$ exists by our theory and is independent of the program's terminating paths.
In fact, for programs that do not exhibit infinite paths, we can show that the limit is $\snull$ using Kőnig's classic infinity lemma.
We discuss the implications of this in \Cref{sec:uniquefp}.
%
%

Note that \cref{thm:wlp-operational-semantics} indeed implies that $\wlpsymbol$ is backward compatible to classical weakest liberal preconditions:
In the Boolean semiring, the right-hand side of \eqref{eq:wlp-operational-semantics} equals $\true$ iff there exists an infinite computation path starting in $\state$, and thus $\wlp{C}{\mnull}(\state) \morespace{\equiv} \mayNotTerminate{C}{\state}$ holds as expected.%
%
%
\begin{example}
    \label{ex:wlp-tropical}
    Reconsider the tropical semiring $\tropringdef$ with $\sTop = 0$.
	The infimum~$\bigmeet$ in the natural order is the supremum in the standard order on $\NatsInf$, and multiplication with the top element $\sTop = 0$ is effectless as $a \smul \sTop = a + 0 = a$.
    It follows from \cref{thm:wlp-operational-semantics} that 
    \[
        \wlp{C}{\mnull}(\state) \qeq \text{\enquote{minimum weight of all \emph{infinite} computation paths starting in $\state$}} ~,
    \]
    or $\infty$ --- the tropical $\mnull$ --- if no infinite path exists.
    Hence $\wlp{C}{0}(\state)$ --- where the natural number $0$ is the tropical $\monone$ --- is the minimum path weight among all finite \emph{and infinite} computation paths starting in $\state$.
    For example, for the program $C$ given by
	\[
		\WHILEDO{x = 2}{\quad {\BRANCH{\COMPOSE{\ASSIGN{x}{3}}{\WEIGH{5}}}{\SKIP}} \quad}
	\]
    and initial state $\sigma$ with $\sigma(x) = 2$, we have $\wp{C}{0}(\sigma) = 5$ 
    %
	but $\wlp{C}{0}(\sigma) = 0$ because there exists an infinite path (only performing $\SKIP$) with weight $0 < 5$.\qedtriangle
\end{example}%
\begin{example}
    \label{ex:wlp-lang}
    Let $\Smodule = \mixedWordLangSmodule{\ab}$ be the module of $\omega$-potent formal languages over the monoid of words $\Monoid = \ab^*$ (cf.~\cref{ex:wp-lang}).
    Thus, weightings $f \in \Weightings$ associate states with languages that contain \emph{both} finite \emph{and $\omega$-words}.
    Let $C$ be a $\ab^*$-$\wgcl$ program.
    It follows from \cref{thm:wlp-operational-semantics} that 
    \[
        \wlp{C}{\snull}(\state) \qeq \parbox{0.55\textwidth}{\enquote{language of all (finite and $\omega$-)words that have some\linebreak \phantom{"}\emph{infinite} computation path starting in $\state$ as prefix}} ~,
    \]
    where we have identified computation paths with the words they are labelled with.
    In particular, if all infinite paths of $C$ are weighted with an $\omega$-word, then $\wlp{C}{\snull}(\state)$ is precisely the language consisting of all these words.
    For example, let $\ab = \{a,b\}$ and consider the following program $C$:
    \begin{align*}
        \WHILE{x = 1} 
        \quad {\BRANCH{\COMPOSE{\ASSIGN{x}{0}}{\WEIGH{a}}}{\WEIGH{b}}} \quad
        \}
    \end{align*}
    If initially $\sigma(x)=1$, then $\wp{C}{\monone}(\sigma) = \set{a, ba, bba, \ldots }$, where $\monone = \set{\epsilon}$, but $\wlp{C}{\snull}(\state) = \{b^\omega\}$ and hence
    \[
        \wlp{C}{\monone}(\state) \eeq \wp{C}{\set{\epsilon}}(\sigma) \morespace{\cup} \wlp{C}{\snull}(\state) \eeq  \set{b^\omega, a, ba, bba, \ldots } ~.\tag*{$\triangle$}
    \]
    %
\end{example}%
\begin{remark}[Probabilistic Weakest Liberal Preexpectations]
\newcommand{\pwlpsymbol}{\sfsymbol{pwlp}}
\newcommand{\pwlp}{\wt{\pwlpsymbol}}
\newcommand{\wlpparamsymbol}[1]{\sfsymbol{wlp}^{\natord #1}}
\newcommand{\wlpparam}[1]{\wt{\wlpparamsymbol{#1}}}

\citet{DBLP:series/mcs/McIverM05} and \citet{DBLP:journals/jcss/Kozen85} define a probabilistic $\wlpsymbol$-semantics where $\wlp{C}{0}(\state)$ yields the probability that $C$ diverges on input $\state$.
Technically, probabilistic programs are $\wgcl$-programs over the real semiring $(\PosRealsInf, +, \cdot, 0, 1)$ where branching and weighting is restricted to statements of the form $\WCHOICE{\ldots}{p}{q}{\ldots}$, where $p + q = 1$.
However, by \cref{thm:wlp-operational-semantics}, our $\wlpsymbol$ over $\PosRealsInf$ yields for \emph{all loops} and \emph{all states} $\wlp{\mathit{loop}}{0}(\state) \in \{\snull = 0, \sTop = \infty\}$ and is thus trivial.
We cannot simply fix this by choosing probabilities in $[0,1]$ as our module $\Smodule$ since $[0,1]$ is not closed under addition.
Nonetheless, we can recover the $\wlpsymbol$ of \citet{DBLP:series/mcs/McIverM05,DBLP:journals/jcss/Kozen85} by considering the \emph{greatest fixed point below or equal to $1$} instead of the true gfp in $\PosRealsInf$, i.e. we would consider a modified transformer $\wlpparamsymbol{1}$.
It is easy to show that this is well-defined and still satisfies \cref{thm:wlp-decomp} and \cref{thm:wlp-operational-semantics} (with the multiplication $\smul \sTop$ on the right hand side of \eqref{eq:wlp-operational-semantics} omitted).
\qedtriangle
\end{remark}


\section{Verification of Loops}
\label{sec:loop-verification}

For loop-free programs, weakest (liberal) preweightings can be obtained essentially by means of syntactic reasoning.
For loops, however, this is not the case since we need to reason about fixed points. 
This section introduces easy-to-apply proof rules for bounding weakest (liberal) preweightings of loops, generalizing rules from the probabilistic setting~\cite{DBLP:series/mcs/McIverM05}.

\subsection{Invariant-Based Verification of Loops}



Let us fix throughout the rest of the section an ambient monoid $\Monoid$ of program weights and an ambient $\omega$-bicontinuous $\Monoid$-module $\Smodule$.
Since the $\wpsymbol$- and $\wlpsymbol$-characteristic functions of loops are $\omega$-(co)continuous (see \cref{thm:wp-continuous} and \ref{thm:wlp-continuous}), we obtain proof rules for loops by Park induction \cref{app:thm:kleene-park}:

\begin{theorem}[Induction Rules for Loops]
    \label{thm:park-wp}
    Let \,$\charwp{\guard}{C}{f}$ and \,$\charwlp{\guard}{C}{f}$ be the $\wpsymbol$- and $\wlpsymbol$-characteristic functionals of the loop $\WHILEDO{\guard}{C}$ with respect to postweighting $f$.
    Then for all $I \in \Weightings$,
    \begin{align*}
       \charwp{\guard}{C}{f}(I) \nnatord I &\qqimplies \wp{\WHILEDO{\guard}{C}}{f} \nnatord I~, \qqand \\
       I \nnatord \charwlp{\guard}{C}{f}(I) &\qqimplies I \nnatord \wlp{\WHILEDO{\guard}{C}}{f}  ~.
    \end{align*}
\end{theorem}%
\noindent{}%
%
%
The weightings $I$ are called \emph{$\wpsymbol$-superinvariants} and \emph{$\wlpsymbol$-subinvariants}, respectively (or just \emph{invariants} if clear from context). 
In many cases --- in particular for loop-free loop bodies --- the above proof rules are easy to apply as they only require to apply the respective characteristic functional \emph{once}. \cref{ex:termination-arctic} demonstrates invariant-based reasoning and our annotation-style for loops.

What about the converse directions, i.e.\ lower bounds for $\wpsymbol$ and upper bounds for $\wlpsymbol$? For that, the analogous formulations of the above proof rules do not hold in general~\cite{DBLP:phd/dnb/Kaminski19}.
In the next subsection we show that in the case of \emph{terminating} programs, these formulations \emph{do} hold.

\subsection{Terminating Programs and Unique Fixed Points}
\label{sec:uniquefp}

The notion of \emph{universal certain termination} is central to the results of this section:%
\begin{definition}[Universal Certain Termination]
    \label{def:uct}
    A $\wgcl$-program $C$ \emph{terminates certainly} on initial state $\state \in \States$ if there does not exist an infinite computation path starting in $\sosStateAbbr{C}{\state}$. 
    Moreover, $C$ is \emph{universally certainly terminating} (UCT) if it terminates certainly on all $\state \in \States$. \qedtriangle
\end{definition}%
\noindent{}%
Certain termination of a program $C \in \wgcl$ is also known as \emph{demonic termination} of the program obtained from $C$ by ignoring all weight-statements and interpreting branching as demonic non-determinism.
Note that all loop-free programs are trivially UCT.
A well-established method for proving certain termination is by use of \emph{ranking functions}~\cite{DBLP:journals/cacm/Dijkstra75}.
An important consequence of UCT is that $\wpsymbol$ and $\wlpsymbol$ coincide. This is intuitively clear because if $C$ is UCT then $\wlpC{C}$ has no additional nonterminating behavior to account for compared to $\wpC{C}$.
Formally:
%
%
%
%
\begin{theorem}[Unique Fixed Points by Universal Certain Termination]
	\label{thm:strong_term_for_state_unique_fp}
	Let $\WHILEDO{\guard}{C}$ have a UCT loop body $C$ and let $\charwp{\guard}{C}{f}$ and $\charwlp{\guard}{C}{f}$ be its $\wpsymbol$- and $\wlpsymbol$-characteristic functionals with respect to an arbitrary postweighting~$f$. 
	Then $\charwp{\guard}{C}{f} = \charwlp{\guard}{C}{f}$.
	
	Furthermore, let $I, J \in \Weightings$ be fixed points of $\charwp{\guard}{C}{f}$.
	Then
%
    \[
        \WHILEDO{\guard}{C} \textnormal{ terminates certainly on } \state \qqimplies I(\state) \eeq J(\state)~.
    \]
    Moreover, if $\WHILEDO{\guard}{C}$ is UCT, then $\charwp{\guard}{C'}{f}$ has a unique fixed point and%
    \begin{align*}
    	\wp{\WHILEDO{\guard}{C}}{f} \eeq \wlp{\WHILEDO{\guard}{C}}{f}~.
    \end{align*}%
\end{theorem}%
\noindent{}%
%
%
%
Hence, the converse directions of the rules in \cref{thm:park-wp} \emph{do} hold for UCT loops with UCT loop-body. In particular, we can reason about \emph{exact} weakest (liberal) preweightings of such loops.%
%
%
\begin{corollary}
    \label{cor:strong_term_unique_fp}
    If both $\WHILEDO{\guard}{C}$ and $C$ are UCT, then 
    for all $f \in \Weightings$ and all $I \in \Weightings$,
    \begin{align*}
    &I \nnatord \charwp{\guard}{C}{f}(I)  \qqimplies I \nnatord \wp{\WHILEDO{\guard}{C}}{f} ~,~\text{and} \\
    & \charwlp{\guard}{C}{f}(I) \nnatord I \qqimplies \wlp{\WHILEDO{\guard}{C}}{f} \nnatord I ~.
    \end{align*}
\end{corollary}%
\noindent{}%
%
%
Let us now look at reasoning about loops in action.
For this, we extend our annotation scheme to loops as shown in \cref{fig:loop-annotations}.
Again, read the annotations from bottom to top as follows and consider ${\bowtie}$ as ${\natord}$ for simplicity:%
\begin{figure}[t]
    \begin{subfigure}[t]{.45\linewidth}
            \begin{minipage}{1\linewidth}
                \begin{align*}
					&\bowtieannotate{I} \tag{meaning $g \bowtie I$}\\
					&\phiannotate{g} \tag{$g = \iverson{\neg \guard} \ivop f \sadd \iverson{\guard} \ivop I''$}\\
					&\WHILE{\guard} \\
					&\qquad \bowtieannotate{I''} \tag{meaning $I' \bowtie I''$}\\
					&\qquad \wpannotate{I'} \tag{meaning $g = \wp{C}{f}$}\\
					&\qquad C \\
					&\qquad \starannotate{I} \tag{we employ invariant $I$}\\
					&\} \\
					&\annotate{f} \tag{postweighting is $f$}
                \end{align*}
            \end{minipage}
		\subcaption{Annotation style for loops using invariants.}
        \label{fig:loop-annotations}
    \end{subfigure}
    \quad\hfill
    \begin{subfigure}[t]{.5\linewidth}
            \begin{minipage}{1\linewidth}
                \begin{align*}
                &\eqannotate{\scriptstyle \iverson{x =0  \vee y = 0} \ivop 0 \mmadd \iverson{x = 0  \vee y = 0} \ivop (2x + y) \gray{\eeq I}} \\
                &\phiannotate{\scriptstyle \iverson{x = 0  \vee y = 0} \ivop 0 \mmadd \iverson{x > 0  \wedge y > 0} \ivop I'' \gray{\eeq g}} \\
                & \WHILE{x > 0  \wedge y > 0} \\
                & \qquad \eqannotate{\scriptstyle \iverson{(x >1) \vee (x > 0 \wedge y > 1)} \ivop (2x+y) \mmadd 1 \smop \iverson{x \leq 1 \vee y \leq 1} \ivop 0 \gray{\eeq I''}} \\
                & \qquad \wpannotate{\scriptstyle \textit{(expression ommited)} \gray{\eeq I'}}  \\
                & \qquad \BRANCH{\COMPOSE{\ASSIGN{x}{x{-}1}}{\ASSIGN{y}{y{+}1}}}{\ASSIGN{y}{y{-}1}} \fatsemi \WEIGH{1} \\
                & \qquad \starannotate{\scriptstyle \iverson{x =0  \vee y = 0} \ivop 0 \mmadd \iverson{x = 0  \vee y = 0} \ivop  (2x + y) \gray{\eeq I}} \\
				& \} \\
				& \annotate{0 \gray{\eeq f}}
                \end{align*}
            \end{minipage}
        \subcaption{$\wpsymbol$ loop annotations for \cref{ex:termination-arctic}.}
        \label{fig:ex:wp-upper-bounds}
    \end{subfigure}
    \caption[]{
        Inside the loop, we push an invariant $I$ (provided externally, denoted by $\annocolor{{\talloblong}\!{\talloblong}\:I}$) through the loop body, thus obtaining $I''$ which is (possibly an over- or underapproximation of) $I' = \wp{\mathit{C'}}{I}$.
        Above the loop head, we then annotate $g = \iverson{\neg\guard} \ivop f \sadd \iverson{\guard} \ivop I''$.
        In the first line, we establish $g \bowtie I$, for~\mbox{${\bowtie} \in \{{\preceq},\, {=},\, {\succeq}\}$}. 
    }
    \label{fig:loop-annotationzzz}
\end{figure}
\begin{enumerate}[align=left]
\item[(1)\hspace{0.975em}$\annotate{f}$\hspace{.35em}]
	We start our reasoning from postweighting $f \in \Weightings$.
\item[(2)\hspace{0.95em}$\starannotate{I}$\hspace{.55em}] We choose (\emph{creatively}) an invariant $I$ which we are going to push through the loop body.
\item[(3)\hspace{1em}$\wpannotate{I'}$\hspace{.25em}]
	This annotation is obtained (\emph{uncreatively}) from applying $\wp{\cloze{C}}{\cloze{f}}$, just as in \cref{fig:annotationzz}.
	The program passed into $\wpsymbol$ is the loop body $C$ and the continuation is the invariant $I$.
	Hence, this annotation states $I' = \wp{C}{I}$ and by that we have pushed $I$ through \mbox{the loop body}.
\item[(4)\hspace{0.95em}$\preceqannotate{I''}$] 
	This annotation states that $I' \natord I''$, i.e.\ $I''$ overapproximates $I'$, just as in \cref{fig:annotationzz}.
\item[(5)\hspace{1em}$\phiannotate{g}$\hspace{.475em}] 
	This annotation $g$ is obtained from $I''$ --- the result of pushing the invariant $I$ through the loop body (and possibly overapproximating the result) --- by constructing $g = \iverson{\neg \guard} \ivop f \sadd \iverson{\guard} \ivop I''$. 
	This annotation states that $g \succeq \charwp{}{}{f}(I)$.
\item[(6)\hspace{0.975em}$\preceqannotate{I}$\hspace{.55em}] 
	This annotation states that $g \natord I$, just as in \cref{fig:annotationzz}.
	Since $\charwp{}{}{f}(I) \natord g \natord I$, this final annotation states by \Cref{thm:park-wp} that $\wp{\mathit{loop}}{f} \natord I$ and we could continue reasoning with $I$.
\end{enumerate}%
%
%
%
%
%
\begin{example}
    \label{ex:termination-arctic}
	Consider the \emph{arctic} semiring $\arcring = ({\NatsInf \cup \set{-\infty}},\, {\max},\, {+},\, {-\infty},\, {0})$ and the program
    \[
        C
        \qeq
        \WHILEDO{x > 0  \wedge y > 0}{ \quad
            \COMPOSE{\BRANCH{\COMPOSE{\ASSIGN{x}{x-1}}{\ASSIGN{y}{y+1}}}{\ASSIGN{y}{y-1}}}{\quad \WEIGH{1} \quad }} ~.
    \]
    $C$ is UCT, witnessed by the ranking function $r = 3x + 2y$:
	Both branches of the loop body strictly decrease the value of $r$.
	We verify that $I = \iverson{\neg \guard} \ivop 0 \sadd \iverson{\guard} \ivop (2(x - 1) + y)$, where $\guard = ( x > 0 \land y > 0)$, is a fixed point of $\charwp{}{}{0}$ in \cref{fig:ex:wp-upper-bounds}.
	Hence, by \cref{cor:strong_term_unique_fp}, we get $\wp{C}{0} = I$.
    By \cref{thm:wp-operational-semantics}, $\wp{C}{0}$ is the maximum weight among all terminating computations paths.
    In $C$, the weight of a path is the number of times it passes through the loop body.
	We thus conclude that the number of $C$'s loop iterations is bounded by $2(x - 1) + y$ if initially $x > 0  \wedge y > 0$ holds.
	This bound is sharp.
\end{example}%

\section{Case Studies}
\label{sec:app}


\subsection{Competitive Analysis of Online Algorithms by Weighted Programming}
\label{sec:competitive-analysis}

%

\begin{center}
    \vspace{-1ex}
    \begin{adjustbox}{max width=1\linewidth}
        \fbox{%
            \parbox{1.2\textwidth}{%
                \smallskip%
                \hspace*{.5em}
                \begin{minipage}{.33\linewidth}
                    \begin{tabular}{l@{\quad}l}
                        \textbf{Field:}& Competitive Analysis\\
                        \textbf{Problem:}& Ski Rental Problem\\
                    \end{tabular}%
                \end{minipage}%
                \hfill
                \begin{minipage}{.37\linewidth}
                    \begin{tabular}{l@{\quad}l}
                        \textbf{Model:}& Optimization Problem \\
                        \textbf{Semiring:}& Tropical Semiring
                    \end{tabular}%
                \end{minipage}%
                \hfill
                \begin{minipage}{.25\linewidth}\begin{tabular}{l@{\quad}l}
                        \textbf{Techniques:}& $\wpsymbol$
                    \end{tabular}%
                \end{minipage}%
            }%
        }%
    \end{adjustbox}\medskip
\end{center}

\noindent
We now demonstrate \emph{how to model optimization problems by means of weighted programming} and how to reason about \emph{competitive ratios of online algorithms}~\cite{DBLP:books/daglib/0097013,DBLP:conf/dagstuhl/1996oa}  \emph{on source code level} by means of our $\wpsymbol$ calculus with the aid of invariants. 
In particular, we model both the \emph{optimal solution to the Ski Rental Problem itself} as well as \emph{the optimal deterministic online algorithm for the problem} as weighted programs.
We argue that weighted programming provides a natural formalism for reasoning about the competitive ratio of online algorithms since weighted programs enable the succinct integration of \emph{cost models}.

\subsubsection{Online Algorithms and Competitive Analysis}

Online algorithms perform their computation without knowing the entire input a priori. 
Rather, parts of the input are revealed to the online algorithm during the course of the computation. 
We 
consider here the well-known \emph{Ski Rental Problem} ~\cite{DBLP:series/txtcs/Komm16}: 
Suppose we go an a ski trip for an \emph{a priori unknown} number of $\varvaclen \geq 1$ days and we do \emph{not} own a pair of skis.
At the beginning of each day, we must choose between either renting skis for exactly one day~(cost: $1$ Euro) or to buy a pair of skis~(cost: $\varbuycost$ euros).

The optimization goal is to minimize the total cost for the whole trip. 
If we knew the duration $\varvaclen$ of the trip a priori, the optimal solution would be rather obvious: 
If $\varvaclen \geq \varbuycost$, we \emph{buy} the skis. 
Otherwise, we are cheaper off renting every day. 
This situation would correspond to an \emph{offline} setting, with both $\varvaclen$ and $\varbuycost$ at hand, allowing for an optimal solution. 
Conversely, if the trip duration $\varvaclen$ is \emph{unkown} and only the cost $\varbuycost$ of the skis is known, we are in an online setting of the Ski Rental Problem.

Lacking knowledge about the entire input a priori often comes at the cost of \emph{non-optimality}: 
An online algorithm typically performs worse than the optimal offline algorithm. 
\emph{Competitive analysis}~\cite{DBLP:books/daglib/0097013} is a technique for measuring the degree of optimality of an online algorithm. The central notion is the \emph{competitive ratio} of an online algorithm. Given a problem instance $\rho$, denote by $\algonl(\rho)$ and $\algopt(\rho)$ the cost of an online algorithm $\algonl$ and the cost of its optimal offline counterpart $\algopt$ on $\rho$, respectively. The competitive ratio of $\algonl$ is defined as%
%
\begin{align*}
     \sup_\rho \frac{\algonl(\rho)}{\algopt(\rho)}  ~,
\end{align*}%
\normalsize%
i.e.\ the smallest constant upper-bounding the ratio between the cost of $\algonl$ and $\algopt$ for all problem instances $\rho$. We determine such competitive ratios by $\wpsymbol$-reasoning on weighted programs.


\begin{figure}
	\centering
	\begin{subfigure}[t]{0.45\textwidth}
			\begin{align*}
				%
				%
				%
				%
				&\WHILE{\varvaclen > 0} \\
				%
				%
				%
				& \qquad \ASSIGN{\varvaclen}{\varvaclen - 1}\fatsemi \\
				%
				%
				%
				%
				& \qquad \{ ~ \texttt{\textcolor{gray}{(* rent *)}}  \\
				%
				%
				&\qquad \qquad \WEIGH{1} \\
				%
				%
				&\qquad \} \BranchSymbol \{  ~ \texttt{\textcolor{gray}{(* buy *)}} \\
				%
				%
				%
				& \qquad \qquad \WEIGH{\varbuycost} \fatsemi \\
				%
				%
				%
				%
				&\qquad \qquad \ASSIGN{\varvaclen}{0}  ~ \texttt{\textcolor{gray}{(* terminate *)}}\\
				%
				%
				&\qquad \} ~ \}
				%
				%
			\end{align*}
		\subcaption{The program $\crentnondet$.}
	\end{subfigure}
	\hfill
	\begin{subfigure}[t]{0.45\textwidth}
			\begin{align*}
			   %
				&\ASSIGN{\varvaccount}{0}\fatsemi \\
				%
				%
				%
                 %
				%
				&\WHILE{\varvaclen > 0} \\
				%
				%
				%
				%
				%
				& \qquad \ASSIGN{\varvaclen}{\varvaclen - 1}\fatsemi
				%
				%
				\ASSIGN{\varvaccount}{\varvaccount + 1}\fatsemi \\
				%
				%
				& \qquad \IF{\varvaccount < \varbuycost} \\
				%
				%
				& \qquad \qquad \WEIGH{1} \\
				%
				%
				& \qquad \ELSE \\
				%
				%
				%
				& \qquad \qquad \WEIGH{\varbuycost}\fatsemi 
				%
				%
				%
				%
				\ASSIGN{\varvaclen}{0} \\
				%
				%
				& \qquad \} ~ \}
				%
				%
			\end{align*}
			%
		\subcaption{The program $\crentonline$.}
	\end{subfigure}
	%
	\caption{The \emph{optimal solution} to the Ski Rental Problem is modeled by $\crentnondet$. The program $\crentonline$ implements the optimal \emph{deterministic online algorithm}.}
	\label{fig:ski_rental_programs_verification}
\end{figure}

\subsubsection{Modeling Infinite-State Online Algorithms as Weighted Programs}

Together with $\wpsymbol$-reasoning, weighted programs over the tropical semiring $\tropring$ provide an appealing formalism for the competitive analysis of \emph{infinite-state} online algorithms since (1) (nondeterministic) programs naturally describe algorithmic problems- and solutions, and (2) reasoning on \emph{source code level} enables reasoning about \emph{infinite-state} models. 
Modeling online algorithms as weighted programs is inspired by \cite{DBLP:conf/soda/AminofKL09,DBLP:journals/talg/AminofKL10}, who employ finite-state weighted automata for the automated competitive analysis of \emph{finite-state} online algorithms. We drop the restriction to finite-state algorithms which comes, however, at the cost of full automation of their verification.

Consider the nondeterministic weighted program $\crentnondet$ on the left-hand side of Figure~\ref{fig:ski_rental_programs_verification} (let us ignore the annotations for the moment). An initial program state $\sigma \in \States$ fixes an instance of the Ski Rental Problem given by the duration $\sigma(\varvaclen)$ of the trip and the cost $\sigma(\varbuycost)$ of the skis. Every execution of $\crentnondet$ on $\sigma$ corresponds to one  possible solution: Each iteration of the loop corresponds to one day of the ski trip. As long as the trip did not end ($\varvaclen >0$), we can either rent the skis (first branch) or buy the skis (second branch). If we buy the skis, there is no further choice to be taken, so the loop terminates. The cost of each choice is modeled by weighing the respective branches appropriately. 

Now recall that in the tropical semiring $\tropring$ we have $\sadd = \min$, $\smul = +$, $\snull = \infty$, and $\sone=0$. Thus, the weight of a terminating computation path $\compPath$ is the sum of the weights along $\compPath$, i.e.\ the cost of the solution given by $\compPath$.
This enables determining the \emph{optimal cost} for every initial program state $\sigma$, i.e.\ every instance of the Ski Rental Problem, by $\wpsymbol$-reasoning (cf.\ \cref{ex:wp-lang}) since 
\[
\wp{\crentnondet}{\sone}(\state) \eeq \text{ \enquote{minimum weight of all terminating computation paths starting in $\state$}}~.
\]%
%
%
%
 Program $\crentonline$ on the right-hand side in Figure~\ref{fig:ski_rental_programs_verification} implements the optimal solution for the \emph{online} version of the Ski Rental Problem. The decisions made by $\crentonline$ must therefore not depend on $\varvaclen$. Let us compare the programs $\crentnondet$ and $\crentonline$. Program $\crentonline$ is obtained from $\crentnondet$ by introducing a counter $\varvaccount$ keeping track of the elapsed time and by replacing the nondeterministic choice in $\crentnondet$ by a \emph{deterministic} one. As long as the current duration of the trip is smaller than the cost of the skis, we rent the skis. As soon as this duration is at least the cost of the skis, we buy the skis. 
Since $\crentonline$ is deterministic, the \emph{cost} of $\crentonline$ on $\sigma$ is given by $\wp{\crentonline}{\sone}(\sigma)$.
%
%
%
%
%
%
\subsubsection{Determining Competitive Ratios by $\wpsymbol$-Reasoning}
Due to the above reasoning,
\[
     \sup_{\sigma \in \States} \,
     %
     \frac{
     \wp{\crentonline}{\sone}(\sigma)}
      {
       \wp{\crentnondet}{\sone}(\sigma)}~.
     %
\]
is the competitive ratio of $\crentonline$.
Hence, we obtain the competitive ratio of $\crentonline$ by determining $\wp{\crentnondet}{\sone}$ and $\wp{\crentonline}{\sone}$. This can be done in an \emph{invariant-based} manner:
\begin{theorem}
	\label{thm:ski_rental_wp}
	We have 
\begin{align*}
	\wp{\crentnondet}{\sone} = \varvaclen \sadd \varbuycost~\text{, and}~
	\wp{\crentonline}{\sone} =  \iverson{\varvaclen=0} \ivop 0    \sadd \iverson{0< \varbuycost}\ivop
	\big( (2\varbuycost  -1) \sadd \iverson{\varvaclen \leq \varbuycost - 1}\ivop\varvaclen \big)~.
\end{align*}
\end{theorem}
\begin{proof}
	Since both $\crentnondet$ and $\crentonline$ are UCT (witnessed by the ranking function $\varvaclen$), it suffices to show that the above weightings are fixed points of the respective characteristic functional by \cref{cor:strong_term_unique_fp}. We proceed by annotating the programs. See \cref{sec:app_ski_rental} for details.
\end{proof}%
\noindent{}%
The fact that $\wp{\crentnondet}{\sone} = \varvaclen \sadd \varbuycost = \lambda \sigma.~ \sigma(\varvaclen) \min \sigma(\varbuycost)$ corresponds to our informal description from the beginning of this section: Depending on whether the duration of the trip $\varvaclen$ exceeds the cost $\varbuycost$ of the skis, it is optimal to either immediately buy the skis or to keep renting them every day. The cost $\wp{\crentonline}{\sone}$ of $\crentonline$ is more involved. If $\varvaclen = 0$, the cost of $\crentonline$ is $0$. Otherwise, i.e.\ if the trip lasts for at least one day, there are two cases. If $\varvaclen$ is strictly smaller than $\varbuycost$, then the cost of $\crentonline$ is the minimum of $2\varbuycost - 1$ and $\varvaclen$. Otherwise, i.e.\ if $\varvaclen$ is at least $\varbuycost$, the cost of $\crentonline$ is $2\varbuycost - 1$.

We can now determine the competitive ratio of $\crentonline$. 
Let, for simplicity, both $\varvaclen > 0$ and $\varbuycost >0$ so that $\varvaclen \madd \varbuycost >0$. 
This assumption is reasonable since the problem becomes trivial if the trip ends immediately or the skis are gratis. 
Given two weightings $f,g$ with $g > 0$, we define $\frac{f}{g} = \lambda \sigma. \frac{f(\sigma)}{g(\sigma)}$. We conclude that the competitive ratio of $\crentonline$ is $2$ since $2$ is the smallest constant upper bounding
\belowdisplayskip=0pt%
\begin{align*}
	  \frac{\wp{\crentonline}{\sone}}{\wp{\crentnondet}{\sone}} 
	 \eeq  \frac
	      {  (2\varbuycost  -1) \sadd \iverson{\varvaclen \leq \varbuycost - 1}\ivop \varvaclen}
	      {\varvaclen \sadd \varbuycost } 
	 %
	 %
	 %
	  \eeq \iverson{n \geq y} \ivop \big( 2- \frac
	 { 1}{\varbuycost } \big)
	 \sadd 
	 \iverson{\varvaclen < \varbuycost }\ivop 1~.
\end{align*}%
\normalsize%
%

\subsection{Mutual Exclusion}

\begin{center}
    \vspace{-1ex}
    \begin{adjustbox}{max width=1\linewidth}
        \fbox{%
            \parbox{1.2\textwidth}{%
                \smallskip%
                \hspace*{.5em}
                \begin{minipage}{.33\linewidth}
                    \begin{tabular}{l@{\quad}l}
                        \textbf{Field:}& Formal Verification\\
                        \textbf{Problem:}& Mutual Exclusion
                    \end{tabular}%
                \end{minipage}%
                \hfill
                \begin{minipage}{.37\linewidth}
                    \begin{tabular}{l@{\quad}l}
                        \textbf{Model:}& Computation Traces \\
                        \textbf{Module:}& $\omega$-potent Formal Languages
                    \end{tabular}%
                \end{minipage}%
                \hfill
                \begin{minipage}{.25\linewidth}\begin{tabular}{l@{\quad}l}
                        \textbf{Techniques:}& $\wlpsymbol$
                    \end{tabular}%
                \end{minipage}%
            }%
        }%
    \end{adjustbox}\medskip
\end{center}


\newcommand{\letw}{\ensuremath{W}\xspace}
\newcommand{\letr}{\ensuremath{R}\xspace}
\newcommand{\letc}{\ensuremath{C}\xspace}
\newcommand{\cmut}{\ensuremath{C_\text{mut}}\xspace}
\newcommand{\varprocid}{\ensuremath{i}\xspace}
\newcommand{\varprocssymbol}{\ensuremath{\ell}\xspace}
\newcommand{\varprocs}[1]{\ensuremath{\varprocssymbol[#1]}\xspace}
\newcommand{\varprocsid}{\ensuremath{\ell[\varprocid]}\xspace}
\newcommand{\varsem}{\ensuremath{y}\xspace}
\newcommand{\varnumprocs}{\ensuremath{N}\xspace}
\newcommand{\varaux}{\ensuremath{k}\xspace}

\newcommand{\valw}{\ensuremath{w}\xspace}
\newcommand{\valn}{\ensuremath{n}\xspace}
\newcommand{\valc}{\ensuremath{c}\xspace}
\newcommand{\valchoose}{\ensuremath{j}\xspace}

\smallskip
\noindent{}%
In this case study, we instantiate weighted programming with the module of $\omega$-potent formal languages to reason about \emph{infinite} behaviors of a semaphore-based mutual exclusion algorithm.
This is done in an invariant-based manner enabled by $\wlpsymbol$-reasoning.

\subsubsection{A Mutual Exclusion Protocol}

Consider the program $\cmut$ shown in~\cref{fig:mutual_exclusion_program_annot} and disregard the weightings for the moment.
Program $\cmut$ models $\varnumprocs$ processes participating in a semaphore-based mutual exclusion protocol.
On each iteration of the non-terminating while loop, a scheduler selects one of the $\varnumprocs$ processes.
The status $\varprocsid$ of the selected process $\varprocid$ is either idle ($\valn$), waiting ($\valw$), or critical ($\valc$).
If the process $\varprocid$ idles, it enters the waiting state.
If the process $\varprocid$ is waiting, it checks whether the binary semaphore (modeled by the shared variable $\varsem$) allows to enter the critical section ($\varsem > 0$) and, if so, enters the critical section.
Otherwise, i.e.\ if $\varsem = 0$, the process must continue waiting.
Finally, if the process $\varprocid$ is in the critical section, it releases the critical section and updates the semaphore appropriately.

It can be shown by standard means that the protocol modeled by $\cmut$ indeed ensures mutual exclusion.
That is, whenever we start in a state where at most one process is in the critical section and $\varsem = 0$, it will never be the case that more than one process is in the critical section.
However, the protocol exhibits unfair behavior.
Suppose the semaphore forbids some waiting process $\varaux$ to enter the critical section, i.e.\ $\varsem=0$ and $\varprocs{\varaux} = \valw$.
It is then possible that the scheduler behaves in an adversarial manner such that process $\varaux$ is going to \emph{starve}, i.e.\ wait \emph{forever}.

\subsubsection{Reasoning about Infinite Behavior by $\wlpsymbol$-Reasoning}
We prove that the protocol exhibits unfair behavior by weighted programming and $\wlpsymbol$-reasoning.
To that end, we instantiate our framework with the $\ab^*$-module\footnote{This is the only example in this section where we do not pursue the \emph{default} method of specifying both the monoid $\Monoid$ and the module $\Smodule$ at once by means of a single semiring $\Sring$.} $\mixedWordLangSmodule{\ab}$ of $\omega$-potent formal languages over $\ab$ (cf.\ \cref{ex:wlp-lang}), where
\[
	\ab \eeq \bigcup_{j \in \Nats\setminus\{0\}}\:  \{ \letr_j,\letw_j,\letc_j  \} ~.
\]
Recall that our module addition $\madd$ is union $\cup$, the monoid and scalar-multiplications are concatenations and the zero element is $\mnull = \emptyset$.
Intuitively, (finite or infinite) behaviors of $\cmut$ correspond to (finite or infinite) words over $\ab$.
For instance, the $\omega$-word $\letc_1\letw_2^\omega$ indicates that process $1$ enters the critical section and that subsequently process $2$ waits forever.
This is realized by weighing the branches of the loop body in $\cmut$ appropriately: If the process $\varprocid$ enters the critical section, waits, or releases the critical section, we weight the corresponding branch by $\letc_\varprocid$, $\letw_\varprocid$, or $\letr_\varprocid$, respectively.
This is similar to labeling the states of a transition system by atomic propositions to express properties of the system in, e.g.\ LTL\ \cite{DBLP:books/daglib/0020348}.
Notice, however, that the transition system underlying $\cmut$ is \emph{infinite} so that standard finite-state model checking techniques do not apply.
Now recall from \cref{ex:wlp-lang} that the \emph{language of $\omega$-words} produced by the loop in $\cmut$ on initial state $\sigma$ is $\wlp{\cmut}{\snull}(\sigma)$.
Since the natural ordering $\natord$ on $\mixedWordLangSmodule{\ab}$ is $\subseteq$, verifying that $\cmut$ indeed exhibits the described unfair behavior boils down to proving that 
\[
     \iverson{1 \leq \varaux \leq \varnumprocs \wedge \varprocs{\varaux} = \valw \wedge \varsem=0} \ivop \letw_\varaux^\omega ~{}\natord{}~ \wlp{\cmut}{\snull}~,
\]
i.e.\ we are obliged to prove a \emph{lower bound} on the weakest liberal preweighting of $\cmut$ w.r.t.\ (irrelevant) postweighting $\snull$, which is done in an invariant-based manner \cref{sec:app_mutual_exclusion}.
The above property indeed states that $\cmut$ exhibits unfair behavior:
Whenever some process $\varaux$ is waiting and the semaphore forbids entering the critical section ($\varsem=0$), the behavior $\letw_\varaux^\omega$ is possible, i.e.\ process $\varaux$ might wait forever.

\begin{figure}[t]
	\centering
    \begin{subfigure}[t]{.53\linewidth}
                \begin{minipage}{1\linewidth}
                    \begin{align*}
						&\WHILE{\true} \\
						& \qquad \BIGBRANCH{\valchoose=1}{\varnumprocs}{\ASSIGN{\varprocid}{\valchoose}} \fatsemi \\
						& \qquad \IF{\varprocsid = \valn} \\
						& \qquad\qquad \ASSIGN{\varprocsid}{\valw} \\
						& \qquad \ELSEIF{\varprocsid = \valw} \\
						& \qquad\qquad \IF{\varsem > 0} \,
							\WEIGH{\letc_\varprocid } \fatsemi
								\ASSIGN{\varsem}{\varsem - 1} \fatsemi
								\ASSIGN{\varprocsid}{\valc} \,\} \\
						& \qquad\qquad \ELSESYMBOL \, \{ \WEIGH{ \letw_\varprocid  }\,\} \\
						& \qquad \ELSEIF{\varprocsid = \valc} \\
						& \qquad\qquad \WEIGH{ \letr_\varprocid } \fatsemi
							\ASSIGN{\varsem}{\varsem + 1} \fatsemi
							\ASSIGN{\varprocsid}{\valn} \\
						& \qquad \} \\
						& \}
                    \end{align*}
                \end{minipage}
		\subcaption{The program $\cmut$.}
		\label{fig:mutual_exclusion_program_annot}
    \end{subfigure}
    \hfill
    \begin{subfigure}[t]{.43\linewidth}
                \begin{minipage}{1\linewidth}
                    \begin{align*}
						& \COMPOSE{\ASSIGN{m}{0}}{\ASSIGN{c}{0}} \,\fatsemi \COMMENT{$\ASSIGN{res}{\texttt{[]}} \,\fatsemi$} \\
						& \WHILE{n > 0} \\
						& \qquad \ASSIGN{n}{n-1} \,\fatsemi \, \{ \\
						& \qquad\qquad \ASSIGN{c}{0}\, \COMMENT{$\APPEND{res}{0} \,\fatsemi$} \\
						& \qquad \} \BranchSymbol \{ \\
						& \qquad\qquad \ASSIGN{c}{c+1} \,\fatsemi \COMMENT{$\APPEND{res}{1} \,\fatsemi$} \\
						& \qquad\qquad \ASSIGN{m}{\max(m,c)} \\
						& \qquad \} \\
						& \}
                    \end{align*}
                \end{minipage}
		\subcaption{The program $\Cfib$.}
		\label{fig:fib}
    \end{subfigure}
	%
	\caption{The program $\cmut$ is a mutual exclusion protocol adapted from~\cite{DBLP:books/daglib/0020348}. The program $\Cfib$ generates $n$-bit strings and stores the maximum number of consecutive 1's in $m$.}
	\label[figure]{fig:bb}
\end{figure}

\subsection{Proving a Combinatorial Identity by Program Analysis}

\begin{center}
    \vspace{-1ex}
    \begin{adjustbox}{max width=1\linewidth}
        \fbox{%
            \parbox{1.2\textwidth}{%
                \smallskip%
                \hspace*{.5em}
                \begin{minipage}{.33\linewidth}
                    \begin{tabular}{l@{\quad}l}
                        \textbf{Field:}& Combinatorics\\
                        \textbf{Problem:}& Counting bit patterns
                    \end{tabular}%
                \end{minipage}%
                \hfill
                \begin{minipage}{.37\linewidth}
                    \begin{tabular}{l@{\quad}l}
                        \textbf{Model:}& Combinatorial class \\
                        \textbf{Semiring:}& Natural numbers
                    \end{tabular}%
                \end{minipage}%
                \hfill
                \begin{minipage}{.25\linewidth}\begin{tabular}{l@{\quad}l}
                        \textbf{Techniques:}& $\wpsymbol$
                    \end{tabular}%
                \end{minipage}%
            }%
        }%
    \end{adjustbox}\medskip
\end{center}


\smallskip%
\noindent{}%
We instantiate our framework with the semiring $(\NatsInf, +, \cdot, 0, 1)$ to count the number of computation paths in our programs.
If a program $C$ does not contain weight-statements, it follows from \cref{thm:wp-operational-semantics} that the number of terminating computation paths starting in $\state$ is given by $\wp{C}{1}(\state)$.
More generally, given a predicate $\guard$ over the program variables, the number of paths terminating in a state satisfying $\guard$ on initial state $\state$ is given by $\wp{C}{\iverson{\guard} \ivop 1}(\sigma)$.
Thus, counting computation paths reduces to weakest preweighting-reasoning as illustrated in the following example.

Suppose we were to count the number of bit strings of length $n$ that avoid the pattern \enquote{11}.
Program $\Cfib$ in \cref{fig:fib} non-deterministically \enquote{constructs} bit strings of length equal to the (input) variable $n$ and simultaneously keeps track of the maximum amount of consecutive 1's that have occurred in variable $m$. Since we are interested in counting strings not containing \enquote{11}, we have to determine $\wp{\Cfib}{\iverson{m \leq 1} \ivop 1}$.
To handle the loop in $\Cfib$, we employ the loop invariant
\[
    I \qcoloneqq \iverson{m \leq 1} \iivop (\iverson{c=0}\ivop \fib(n+2) \ssadd \iverson{c>0}\ivop \fib(n+1))
\]
and verify that $I$ is indeed a fixed point of the $\wpsymbol$-characteristic functional of the loop \cref{app:path_counting}.
Here, $\fib(n)$ is the $n$-th \emph{Fibonacci number} defined recursively via $\fib(0) \coloneq 0, \fib(1) \coloneq 1$, and for all $n \geq 2$, $\fib(n) \coloneq \fib(n-1) + \fib(n-2)$.
Since $\Cfib$ is obviously certainly terminating and $I$ is a fixed point of the $\wpsymbol$-characteristic function of the loop w.r.t.\ postweighting $\iverson{m \leq 1} \ivop 1$, we have
\[
    \wp{\Cfib}{\iverson{m \leq 1} \ivop 1} \qeq I\subst{c}{0}\subst{m}{0} \qeq \fib(n+2) ~,
\]
by \cref{thm:strong_term_for_state_unique_fp}, i.e. the number of \enquote{11}-avoiding bit strings of length $n$ is equal to $\fib(n+2)$.



\section{Related Work}
\label{sec:related}

We organize related works in three categories:
(1)~Other generalized predicate transformers, (2)~semiring programming paradigms, (3)~other approaches to modelling optimization problems.

\subsection{Generalized Predicate Transformers and Hoare Logics}

A well-known concrete instance of generalized, quantitative predicates are \emph{potential functions} $\Phi \colon \States \to \PosReals$.
Such functions are used in amortized complexity analysis~\cite{tarjan1985amortized} can be regarded an instance of the weightings introduced in this paper.
\citet{DBLP:conf/pldi/Carbonneaux0S15,carbonneaux2018modular} present a resource bound verification system for a subset of C programs based on potential functions.
A non-trivial subset of their verification rules can be recovered by instantiating our framework with the tropical semiring, and interpreting their resource consumption statement $\code{tick(n)}$ as our weight primitive $\WEIGH{n}$.
More specifically, \citet{DBLP:conf/pldi/Carbonneaux0S15} define a quantitative Hoare triple $\set{\Phi} P \set{\Phi'}$, where $\Phi$, $\Phi'$ are potential functions, and $P$ is a (deterministic) program.
Such a triple is valid iff for all initial states $\sigma \in \States$ such that $P$ terminates in a final state $\sigma'$ it holds that $\Phi(\sigma) \geq n + \Phi'(\sigma')$, where $n$ is the resource consumption of $P$ started on $\sigma$.
It follows that $\set{\wp{P}{\Phi'}} P \set{\Phi'}$ is always a valid triple;
furthermore, $\wp{P}{\Phi'}$ is the \emph{least} potential $X$ that validates the triple $\set{X} P \set{\Phi'}$.
While the programming language from~\cite{DBLP:conf/pldi/Carbonneaux0S15} has advanced features such as procedures and recursion, it lacks a non-deterministic choice as present in $\wgcl$.
A promising direction for future work is to investigate whether the automatic inference algorithm of~\citet{DBLP:conf/pldi/Carbonneaux0S15} can be extended to non-deterministic programs.

Very recent works have studied predicate transformers and Hoare-style logics from an abstract categorical perspective.
A generic approach to define predicate transformers, like our $\wpsymbol$ and $\wlpsymbol$, is given by \citet{DBLP:conf/mfps/AguirreK20}, but \emph{only for loop-free programs}.
On the loop-free fragment of $\wgcl$, our weakest preweighting transformer $\wpsymbol$ is an instance of their framework.
They capture the computational side effects, like our weightings, in a \emph{monad}.
More precisely, our transformer is obtained from the composed monad \code{MSet (Writer w -)} of a multiset monad \code{MSet -} that distributes over a writer monad \code{Writer w -}.
This specific instance, however, is not discussed explicitly by \citet{DBLP:conf/mfps/AguirreK20}.
The writer monad corresponds to our weighting monoid $\Monoid$, whereas the multiset monad captures the branching construct $\BRANCH{C_1}{C_2}$ that we treat via $\Monoid$-modules.
In contrast to their work, our $\wpsymbol$ is defined for loops.
Moreover, we introduce \emph{two} transformers, $\wpsymbol$ for finite computations and $\wlpsymbol$ that additionally accounts for \emph{infinite computations}.
Finally, the correspondence to an operational semantics is not established in \cite{DBLP:conf/mfps/AguirreK20}.
An interesting direction for future work is to explicitly construct a \emph{strongest postcondition} transformer for weighted programming, which \citet{DBLP:conf/mfps/AguirreK20} define non-constructively as an adjoint to $\wpsymbol$.
Problems with defining strongest postexpectations for probabilistic programs, see \cite{claire90}, demonstrate that giving a \emph{concrete} strongest post semantics is far less easy, even if it can be defined abstractly as an adjoint.

In a similar spirit, \citet{DBLP:conf/esop/GaboardiKOS21} introduce the notion of \emph{graded categories} to unify \emph{graded monadic} and \emph{graded comonadic} effects.
The gradings are over partially ordered monoids (pomonoids) and can, for example, model probabilities or resources like our weightings.
In the setting of \emph{imperative} languages, they consider it \enquote{natural to have just the multiplicative structure of the semiring as a pomonoid}~\cite[Sec.\ 6]{DBLP:conf/esop/GaboardiKOS21} because their programs only have one input and output.
The additive structure of semiring gradings has been used to join multiple \emph{inputs} for resource consumption in the $\lambda$-calculus with \emph{comonadic contexts} \cite{DBLP:conf/esop/BrunelGMZ14,DBLP:conf/esop/GhicaS14,DBLP:conf/icfp/PetricekOM14}.
In contrast, we use addition to join multiple \emph{outputs} in \emph{monadic computations}, \eg \emph{branching} in our examples (\cref{fig:ski_rental_programs_verification} and \cref{fig:bb}).
Hence, it might be interesting future work to extend their categorical semantics with branching.
They go on to construct a \emph{Graded Hoare Logic (GHL)} with \emph{judgments} $\vdash_{w} \{\phi\}\,C\,\{\psi\}$ corresponding to $\iverson{\phi} \ivop w \natord \wp{C}{\iverson{\psi}\ivop\sone}$ given \emph{(Boolean)} pre- and postconditions $\phi, \psi$, program $C$, and a weight $w$ from a semiring.
Although \emph{unbounded} loops have been studied in \emph{concrete instances}, they restrict to \emph{bounded} loops in the general setting:
\enquote{%
    This allows us to focus on the grading structures for total functions, leaving the study of the interaction between grading and partiality to future work.%
}~\cite[Sec.\ 2]{DBLP:conf/esop/GaboardiKOS21}.
Our work does not impose such restrictions.
Both of our verification calculi $\wpsymbol$ and $\wlpsymbol$ deal with possibly unbounded loops.

\citet{DBLP:journals/pacmpl/SwierstraB19} handle effects by monads, focussing on \emph{functional} rather than \emph{imperative} programming. 
They show how to \emph{synthesize} programs from specifications using general results on predicate transformers.
Combining these synthesis techniques with the above monad instance of \cite{DBLP:conf/mfps/AguirreK20} for the synthesis of weighted programs is an interesting direction for future work.

\subsection{Computing with Semirings}
There exist a number of computation and programming paradigms in the literature that --- similarly to our approach --- are parameterized by a semiring.
\citet{OConner11} and \citet{DBLP:conf/icfp/Dolan13} show that computational problems such as shortest paths, deriving the regular expression of a finite automaton, dataflow analysis, and others can be reduced to linear algebra over a suitable semiring.
They also provide concise Haskell implementations solving the resulting matrix problems in a unified way.
The heart of these techniques is to compute the so-called \emph{star} or \emph{closure} $x^* = 1 + x + x^2 + \ldots$ where $x$ is a matrix over the semiring.
The same $x^*$ is also the least solution of the equation $x^* = 1 + x \cdot x^*$.
This fixed point equation is closely related to the $\lfpop$ occurring in our $\wpsymbol$.
In fact, it can be interpreted as an automata-theoretic explicit-state analog to our $\wpsymbol$.
Our framework, however, extends this to infinite state spaces and allows reasoning in a symbolic fashion.
The above techniques, on the other hand, would require an infinite transition matrix $x$, and are therefore limited to \emph{finite-state} problems, \eg shortest paths in \emph{finite} graphs.

Functional and declarative approaches for programming with semirings have also been explored.
\citet{DBLP:conf/lics/LairdMMP13} and \citet{DBLP:conf/esop/BrunelGMZ14} consider functional languages parameterized by a semiring and provide a categorical semantics.
Their languages feature weighting computation steps similar to our language.
Additionally, \citet{DBLP:conf/esop/BrunelGMZ14} provide static analysis techniques to obtain upper bounds on the weight of a computation.
Indeed, with an appropriate semiring, the semantics defined in these works also allows reasoning about \eg best/worst-case resource consumption, reachability probabilities, or expected values.
In contrast to \cite{DBLP:conf/lics/LairdMMP13,DBLP:conf/esop/BrunelGMZ14}, our programming language is imperative and our semantics generalizes weakest preconditions. Moreoever, while \citet{DBLP:conf/lics/LairdMMP13} exemplify how their framework can be used to detect infinite reduction sequences, it does not provide a general way to assign a weight to diverging computation paths as our $\wlpsymbol$ does.
\citet{DBLP:conf/esop/BrunelGMZ14} do not deal with infinite computations.

\citet{DBLP:journals/ijar/BelleR20} pursue a declarative approach by computing the \emph{weighted model count} of logical formulae in some theory where the literals are weighted in a semiring.
Applications include matrix factorization, computing polyhedral volumes, or probabilistic inference.
Furthermore, \citet{DBLP:conf/iclp/CohenSS08,DBLP:conf/iwpt/BalkirGC20} study \emph{weighted logic programs} with a focus on parsers.
This declarative paradigm is, however, rather different from our weighted programs which allow specifying models in an algorithmic, imperative manner.

Kleene Algebras with Tests (KAT) \cite{DBLP:conf/lics/Kozen99,DBLP:journals/tocl/Kozen00} can model imperative programs in an abstract fashion by identifying them with the objects from a \emph{Kleene algebra} with an embedded Boolean subalgebra.
An important application of KAT is \emph{equational reasoning}; and hence to \eg derive the rules of Hoare logic by applying algebraic manipulations.
Note that a Kleene algebra is itself an idempotent semiring whose purpose, however, is not to model weights of any kind but the programs themselves.
Nonetheless, to reason about weighted computations similar to us, KAT was recently generalized to Graded KAT \cite{DBLP:journals/cuza/GomesMB19} by replacing the Boolean subalgebra with a more general object that can be viewed as a semiring with additional operations and axioms.
The elements of this semiring constitute the graded (or weighted) outcomes of the tests.
However, (Graded) KAT are no concrete programming languages; their main purpose is to prove general results about imperative languages with loops and conditionals in an abstract fashion.
Indeed, investigating which of our $\wpsymbol$ ($\wlpsymbol$) and invariant-based proof rules can be derived in Graded KAT is an appealing direction for future work.

\subsection{Optimization}
There exists a large amount of work on modelling and solving optimization problems.
A prominent example is constrained optimization (\eg linear programming \cite{DBLP:books/daglib/0090562,DBLP:journals/networks/Horen85}) for which standardized- and domain-specific languages exist \cite{DBLP:conf/cp/NethercoteSBBDT07,Lofberg2004YALMIPA}.
Modelling and solving optimization problems with weighted programming differs mainly in two aspects from these techniques.
(1) The way \emph{how} optimization problems are modelled and (2) \emph{what} is modelled and for what purpose.
Regarding aspect (1), techniques like integer linear programming or languages like MiniZinc model optimization problems in a constraint-based manner.
With weighted programs, we describe these problems instead in an algorithmic fashion.
As an intuition, constrained optimization \vs weighted programming could be considered analogous to logic programming \vs imperative programming.

Regarding aspect (2), constraint-based techniques often model \emph{one particular} problem instance for which an optimal solution is \emph{computed}.
Weighted programs, on the other hand, provide a means to model and reason about every (out of possibly infinitely many) problem instance at once.
This comes, however, at the cost of computability.
The case study on the ski rental problem exemplifies this: We verify the competitive ratio of the optimal online algorithm for \emph{every} trip duration of the ski rental problem.
Automating this verification process is an appealing direction for future work. 

More closely related is the work by \citet{DBLP:conf/ijcai/BistarelliMR97}, who generalize Constraint Logic Programming (CLP) by parameterizing CLP with a semiring $S$.
Elements and operations of $S$ take over the role Boolean constants and connectives.
This allows to, \eg solve optimization problems by finding atom instantiations of minimal cost.


\section{Conclusion}
\label{sec:conclusion}
We have studied weighted programming as a programming paradigm for specifying mathematical models. 
We developed a weakest (liberal) precondition-style verification framework for reasoning about both finite and infinite computations of weighted programs and demonstrated the efficacy of our framework on several case studies. 
Future work includes automated reasoning about weighted programs using, e.g., generalizations of $k$-induction \cite{DBLP:conf/cav/BatzCKKMS20,DBLP:conf/fmcad/SheeranSS00} and weighted program synthesis \cite{DBLP:series/natosec/AlurBDF0JKMMRSSSSTU15,DBLP:journals/toplas/MannaW80}.
Further directions are weighted separation logics \cite{DBLP:conf/popl/IshtiaqO01,DBLP:conf/lics/Reynolds02,DBLP:journals/pacmpl/BatzKKMN19} as well as to investigate \enquote{sampling} algorithms for weighted programs.
For instance, what would be an analogon to MCMC sampling in a weighted setting?

\begin{acks}
	This work was supported by the \grantsponsor{erc}{ERC}{http://dx.doi.org/10.13039/501100000781} \grantnum[https://cordis.europa.eu/project/id/787914]{erc}{AdG Frappant (787914)} and \grantnum[https://gepris.dfg.de/gepris/projekt/282652900]{dfg}{RTG 2236 UnRAVeL} funded by the \grantsponsor{dfg}{German Research Foundation}{https://doi.org/10.13039/501100001659}.
	Part of this work was carried out at Schloss Dagstuhl -- Leibniz Center for Informatics. We thank Lena Verscht and Linpeng Zhang for the fruitful discussions at Schloss Dagstuhl.
\end{acks}

\bibliography{references}
\appendix

\newpage
\section{Background on Semirings, Semimodules, and Fixed Point Theory}
\label{app:preliminaries}

\subsection{Fixed Points}
\label{app:fixed-point-theory}
We apply fixed point iteration and fixed point induction to our $\wpsymbol$ and $\wlpsymbol$ calculi.
Hence, we recall the required material from Domain Theory here.
For a thorough introduction, we refer to \cite[Ch.\ Domain Theory]{DBLP:books/lib/Abramsky94}.

A \emph{reflexive}, \emph{transitive}, and \emph{antisymmetric} binary relation $\ord$ on a set $\poset$ is a \emph{partial order} and we call $\Poset$ a \emph{partially ordered set (poset)}.
Let $B \subseteq A$.
We say that $b \in B$ is a \emph{least element} in $B$ if $b \ord b'$ for all $b' \in B$.
Note that there exists at most one least element.
The least element of $B = A$ is written $\bot$ whenever it exists.
Further, if the set $\Set{c \in A}{\forall b \in B \colon b \ord c}$ has a least element, then we call it the \emph{least upper bound} or \emph{supremum} of $B$ and denote it with $\bigjoin B$.
An infinite sequence $(a_i)_{i \in \Nats}$ of elements from $A$ is called an \emph{ascending $\omega$-chain} if
\[
    a_0 \oord a_1 \oord a_2 \oord \ldots \quad,
\]
i.e. for all $i \in \Nats$ we have $a_i \ord a_{i+1}$.

It is easy to verify that the structure $(\poset, \rord)$ that results from reverting the order $\ord$ is also a poset.
\emph{Greatest elements}, \emph{greatest lower bounds} or \emph{infima}, and \emph{descending $\omega$-chains} in $\Poset$ are defined as least elements, suprema and ascending $\omega$-chains in $(\poset, \rord)$, respectively.
The greatest element in $\poset$ is denoted $\top$ if it exists and we adapt the notation $\bigmeet B$ for infima.
\begin{definition}[$\omega$-cpo]
	\label{app:def:omega-cpo}
	The poset $\Poset$ is a (pointed) \emph{$\omega$-complete partial order} ($\omega$-cpo) if there exists\footnote{Some authors define $\omega$-cpo without requiring the existence of least elements, and speak of $\omega$-cpo \emph{with bottom} or \emph{pointed} $\omega$-cpo.} a least element $\bot$ and every ascending $\omega$-chain $(a_i)_{i \in \Nats}$ has a supremum $\bigjoin_{i \in \Nats} a_i$.
	Dually, we call $\Poset$ a (pointed) \emph{$\omega$-cocomplete partial order} ($\omega$-cocpo) if $(\poset,\rord)$ is an $\omega$-cpo, i.e. if there exists a greatest element $\top$ and every descending $\omega$-chain $(a_i)_{i \in \Nats}$ has an infimum $\bigmeet_{i \in \Nats} a_i$.
    If $\Poset$ is both an $\omega$-cpo and an $\omega$-cocpo, then we call it an \emph{$\omega$-bicpo}.
	\qedtriangle
\end{definition}%
\noindent{}%
Consider a poset $\Poset$ with a function $f \colon \poset \to \poset$.
$f$ is called \emph{monotone} if
\[
    \forall a_1, a_2 \in \poset \colonq a_1 \oord a_2 \qimplies f(a_1) \oord f(a_2) ~.
\]
Note that $\omega$-chains are preserved under monotone functions:
If $(a_i)_{i \in\Nats}$ is ascending (descending, respectively), then the same holds for $(f(a_i))_{i \in\Nats}$.

\begin{definition}[$\omega$-continuous functions]
	\label{app:def:omega-continuity}
	Let $\Poset$ be a poset and $f \colon \poset \to \poset$ a function.
	$f$ is called \emph{$\omega$-continuous} if it preserves suprema, i.e. for all ascending $\omega$-chains $(a_i)_{i \in \Nats}$ in $\poset$ we have
    \[
        f  \pr{\bigjoin_{i \in \Nats} a_i}  \qeq \bigjoin_{i \in \Nats} f(a_i) ~.
    \]
    Dually, $f$ is called \emph{$\omega$-cocontinuous} if it preserves infima of descending chains.
    If $f$ is both $\omega$-continuous and -cocontinuous, then $f$ is called \emph{$\omega$-bicontinuous}.
	\hfill$\triangle$
\end{definition}
\noindent{}%
It is easy to see that $\omega$-(co)continuity implies monotonicity, but the converse is false in general.
\begin{lemma}
	\label{app:thm:compositional-continuity}
	If $f\colon \Poset \to \Poset$ and $g\colon \Poset \to \Poset$ are $\omega$-(co)continuous functions, then their composition $g \circ f\colon \Poset \to \Poset$ is also $\omega$-(co)continuous.
\end{lemma}
\noindent{}%
The $n$-fold composition of a function $f\colon A \to A$ is recursively defined as $f^0 = \idFun$ and $f^n = f \circ f^{n-1}$ for all $n > 0$.

Let $\Poset$ be a poset.
A \emph{fixed point} of $f \colon \poset\to\poset$ is an element $a \in \poset$ such that $f(a) = a$.
A \emph{least} (\emph{greatest}) fixed point of $f$ is a least (greatest, respectively) element in the set of fixed points of $f$.


\begin{theorem}[Kleene Iteration \& Park Induction]
    \label{app:thm:kleene-park}
    Let $\Poset$ be a poset and $f \colon \poset \to \poset$.
    \begin{enumerate}
        \item \label{it:lfp}
            If $\Poset$ is an $\omega$-cpo and $f$ is $\omega$-continuous, then $f$ has a least fixed point $\lfpop f$ satisfying
            \[
                \lfpop f \eeq \bigjoin_{n \in \Nats} f^n(\bot)
                \qqand
                \forall a \in A \colonq f(a) \oord a \qimplies \lfpop f \oord a
                ~.
            \]
        \item \label{it:gfp}
           If $\Poset$ is an $\omega$-cocpo and $f$ is $\omega$-cocontinuous, then $f$ has a greatest fixed point $\gfpop f$ satisfying
            \[
                \gfpop f \eeq \bigmeet_{n \in \Nats} f^n(\top)
                \qqand
                \forall a \in A \colonq a \oord f(a) \qimplies a \oord \gfpop f
                ~.
            \]
    \end{enumerate}
\end{theorem}

\subsection{Semirings}
\label{app:sec:srings}
In the literature the term semiring is given different meanings; to prevent any confusion we restate the definition we use.
As usual, multiplication $\smul$ associates stronger than addition $\sadd$ and we drop parentheses accordingly.
For an in-depth introduction, we refer to \cite[Chapter 1, 2]{Droste2009}.
\begin{definition}[Monoids]
	\label{app:def:monoid}
	A \emph{monoid} $\Monoid = (\monoid,\, \monop,\, \monone)$ consists of a \emph{carrier set} $\monoid$, an \emph{operation} $\monop \colon \monoid \times \monoid \to \monoid$, and an \emph{identity} $\monone \in \monoid$, such that for all $a,b,c \in \monoid$,%
	\begin{enumerate}
		\item 
			the operation $\monop$ is associative, i.e.\ ,%
			\begin{align*}
				a \mmonop ( b \monop c) \eeq ( a \monop b ) \mmonop c~, \qqand
			\end{align*}%
		\item
			$\monone$ is an identity with respect to $\monop$, i.e.\ ,%
			\begin{align*}
				a \monop \monone \eeq \monone \monop a \eeq a~.
			\end{align*}%
	\end{enumerate}%
	The monoid $\Monoid$ is \emph{commutative}, if%
	\begin{enumerate}
		\setcounter{enumi}{2}
		\item 
			the operation $\monop$ is commutative: \quad $a \monop b \eeq b \monop a$~. \qedtriangle
	\end{enumerate}%
\end{definition}
\noindent{}%
\begin{definition}[Semirings]
	\label{app:def:sring}
	A \emph{semiring} $\sringdef$ consists of a \emph{carrier set} $\sring$, an \emph{addition}~${\sadd}\colon \sring \times \sring \to \sring$, a \emph{multiplication}~${\smul}\colon \sring \times \sring \to \sring$, a \emph{zero} $\snull \in \sring$, and a \emph{one} $\sone \in \sring$, 
	such that%
	\begin{enumerate}
		\item 
			$(\sring,\, {\sadd},\, \snull)$ forms a \emph{commutative monoid}, \qand
		\item
			$(\sring,\, {\smul},\, \sone)$ forms a (possibly not-commutative) \emph{monoid},
	\end{enumerate}%
	and for all $a,b,c \in \sring$,%
	\begin{enumerate}
		\setcounter{enumi}{2}
		\item
			multiplication \emph{distributes} over addition, i.e.\ ,%
			\begin{align*}
				a \ssmul (b \sadd c) \eeq a \smul b \ssadd a \smul c
				\qqand
				(a \sadd b) \ssmul c \eeq a \smul c \ssadd b \smul c ,
			\end{align*}%

		\item
			and multiplication by zero \emph{annihilates} $\sring$, i.e.\ ,%
			\begin{align*}
				\snull \smul a \eeq a \smul \snull \eeq \snull .
				\tag*{$\triangle$}
			\end{align*}%
	\end{enumerate}%
\end{definition}%
\noindent{}%

Given a semiring $\Sring$, we can construct a semiring of functions via point-wise lifting of the operations.
\begin{lemma}[Semirings of Semiring-valued Functions]
	\label{app:thm:function-sring}
	Let $\Sring = (\sring,\, {\sadd_{\Sring}},\, {\smul_{\Sring}},\, \snull_{\Sring}, \sone_{\Sring})$ be a semiring and $X$ be a non-empty set.
	Then $\Sring^X \coloneqq (\sring^X,\, {\sadd},\, {\smul},\, \snull,\, \sone)$, where $\sring^X$ is the set of functions of type~\mbox{$X \to \sring$} and for all $f, g \in \sring^X$,%
	\begin{align*}
		f \sadd g &\qqcoloneqq \qlam{x} f(x) ~{}\sadd_{\Sring}{}~ g(x) ,\\
		f \smul g &\qqcoloneqq \qlam{x} f(x) ~{}\smul_{\Sring}{}~ g(x) ,\\
		\snull &\qqcoloneqq \qlam{x} \snull_{\Sring} ,\\
		\sone &\qqcoloneqq \qlam{x} \sone_{\Sring} ~,
	\end{align*}%
	also forms a semiring which we call the \emph{lifting of $\Sring$ with respect to $X$}.\qedtriangle
\end{lemma}%
\noindent{}%
There is a canonical embedding $\sring \to \sring^X, a \mapsto (\lam{x} a)$ mapping semiring elements in $\sring$ to constant functions in $\sring^X$.
For better readability, we overload notation and identify elements $a$ with their corresponding constant functions $\lam{x} a$, e.g.\ writing $a \smul f$ instead of $(\lam{x} a) \smul f$ where $f \in \sring^X$.

For our purpose of developing a weakest-precondition-style calculus for weighted programs, we need to impose additional structure on our semirings; most essentially: an order, in particular one which is compatible with the algebraic structure of the semiring.
Let us, for the remainder of this section, fix an ambient semiring $\sringdef$.%
\begin{definition}[Natural Order]
	\label{app:def:sring-natord}
	The relation $\natord$ is defined for all $a, b \in \sring$ by
	\[
		a \nnatord b \qqiff \exists\, c \in \sring \colonq a \sadd c \eeq b .
	\]
	If $\natord$ is a partial order, then we call $\Sring$ \emph{naturally ordered} and $\natord$ the \emph{natural order} on $\sring$.
	\hfill$\triangle$
\end{definition}
\noindent{}%
Note that a presence of additive inverses (other than the self-inverse $\snull$) prohibits the existence of a natural order as antisymmetry is violated.
Indeed, the relation $\natord$ degenerates to $\sring \times \sring$ for rings.%

\begin{lemma}[{Least Elements and Monotonicity of Algebraic Operations~\textnormal{\cite[Ch.\ 9, Thm.\ 2.1]{handbookFormalLang}}}]
	\label{app:thm:sring-natordprops}
	Let $\Sring$ be naturally ordered.
	Then,%
	\begin{enumerate}
		\item
			$\snull$ is the \emph{unique} least element, and
			
		\item
			$\sadd$ and $\smul$ are monotone, i.e. for all $a,b,c \in \sring$,%
			\[
				a \nnatord b
				\qqimplies
				a \sadd c \nnatord b \sadd c
				\qand
				a \smul c \nnatord b \smul c
				\qand
				c \smul a \nnatord c \smul b~.
			\]%
	\end{enumerate}%
\end{lemma}%
\noindent{}%
If the natural order $\natord$ moreover has a greatest element, then this is unique and denoted $\sTop_\Sring$, where we drop the subscript $\Sring$ whenever it is clear from the context.

The partial order can be lifted pointwise to functions in the spirit of \cref{app:thm:function-sring}, i.e.\ a partial order ${\ord} \subseteq \sring^X \times \sring^X$ on $\Sring^X$ is given by 
\[
	f \oord g \qqiff \forall x \in X \colonq f(x) \nnatord g(x) ~.
\]
Moreover, if $\natord$ is the natural order on $\Sring$, then $\natord$ as defined above is the natural order on $\Sring^X$.
In case of $\omega$-(co)cpos, joins $\join$ (meets $\meet$) are hence given by pointwise joins (meets).

%
Second, in order to apply fixed point theory to semirings, we require some continuity constraints.
\begin{definition}[$\omega$-continuous semirings {\cite[Def.\ 14]{10.1007/978-3-540-85780-8_1}}]
    \label{app:def:continuous-sring}
	A semiring $\sringdef$ is \emph{$\omega$-continuous}, if $(\sring,\natord)$ is an $\omega$-cpo and addition and multiplication by constants are $\omega$-continuous functions, i.e. for all $a \in \sring$ and all ascending chains $(b_i)_{i \in \Nats}$ in $\sring$, we require
    \[
        a \ssadd \bigjoin_{i \in \Nats} b_i
        \qeq
        \bigjoin_{i \in \Nats} \left(a \sadd b_i \right)
    \]
    as well as
    \[
        a \ssmul \bigjoin_{i \in \Nats} b_i
        \qeq
        \bigjoin_{i \in \Nats} \left(a \smul b_i\right)
        \qqand
        \bigjoin_{i \in \Nats} b_i \ssmul a 
        \qeq
        \bigjoin_{i \in \Nats} \left(b_i \smul a \right)
        ~.
    \]
    Dually, $\Sring$ is \emph{$\omega$-cocontinuous} if $(\sring,\natord)$ is an $\omega$-cocpo and addition/multiplication with constants are $\omega$-cocontinuous functions.
    $\Sring$ is \emph{$\omega$-bicontinuous} if it is both $\omega$-continuous and -cocontinuous.
	\hfill$\triangle$
\end{definition}
\noindent{}%
Moreoever, it is clear that if $\Sring$ is $\omega$-(co)continuous, then the semiring $\Sring^X$ of functions from a set $X$ to $\Sring$ (cf.\ \cref{app:thm:function-sring}) is $\omega$-(co)continuous as well.

The $\omega$-continuity also allows to define \emph{countably infinite sums}.
Semirings that admit an infinite sum operation are called \emph{complete}.
We only consider the case of \emph{$\omega$-finitary} semirings \cite[Sec.\ 5]{Karner1992}, with respect to the \emph{natural order} $\natord$, where such infinite sums are defined as follows \cite[Thm.\ 19]{10.1007/978-3-540-85780-8_1}:
Given a family $(a_i)_{i \in I}$ in $\sring$ over a \emph{countable} index set $I$,
\begin{align}
    \label{eq:sring-infinite-sums}
	\sbigadd_{i \in I} a_i \qcoloneqq \bigjoin_{\substack{F \subseteq I,\\F \text{ finite} }} \sbigadd_{i \in F} a_i ~.
\end{align}
This definition enjoys two important properties:
(\romannumeral 1) If $I$ is \emph{finite}, the value coincides with the usual sum;
(\romannumeral 2) the summation order is irrelevant by definition.
In fact, if $I$ is infinite it can be shown~\cite[Ch.\ 9, Thm.\ 2.3]{handbookFormalLang}, (\cref{app:proof-finitary-sums-coincide}) that the right-hand side of \eqref{eq:sring-infinite-sums} is equal to the supremum of the partial sums associated with any \emph{arbitrary summation order}, i.e. for all $\Nats$-indexed families $(b_i)_{i \in \Nats}$ such that there exists a bijection $\tau \colon \Nats \to I$ with $b_i = a_{\tau(i)}$ for all $i \in \Nats$ we have
\begin{align}
    \label{eq:sring-partial-sums}
    \sbigadd_{i \in I} a_i
    \qeq
    \bigjoin_{n \in \Nats} \sbigadd_{i \leq n} b_i
    ~.
\end{align}
The latter formulation is used as definition of $\omega$-continuous semirings in \cite{DBLP:journals/tcs/Kuich91}.
It shows that being an $\omega$-cpo already suffices to define infinite sums since $\left(\sbigadd_{i \leq n} b_i\right)_{n \in \Nats}$ is clearly an $\omega$-chain.\footnote{The converse also is true, e.g.\ $\omega$-finitary semirings with respect to the natural order $\natord$ are an $\omega$-cpo with respect to $\natord$ \cite[Thm.\ 2.3]{DBLP:journals/tcs/Kuich91}. Their notion of $\omega$-continuous semiring is hence also equivalent to \cref{app:def:continuous-sring}.}
Moreover, it follows by $\omega$-continuity that the \emph{extended distributive laws} are satisfied:
For all $c \in \sring$,
\[
	c \smul \pr{ \sbigadd_{i \in I} a_i } \eeq \sbigadd_{i \in I} c \smul a_i
	\qand
	\pr{ \sbigadd_{i \in I} a_i } \smul c \eeq \sbigadd_{i \in I} a_i \smul c ~.
\]

For an in-depth discussion of \emph{complete}, \emph{finitary}, and \emph{continuous} semirings we refer to \cite{Karner1992,DBLP:journals/tcs/Kuich91,10.1007/978-3-540-85780-8_1,DBLP:journals/tcs/Karner04,handbookFormalLang,Goldstern2002,Golan2003SemiringsAA}.

\subsection{Modules over Monoids}
\label{app:sec:modules}
\emph{Like vector spaces over a field, modules over rings, or semimodules over semirings we define modules over monoids}.
Semimodules over semirings in the setting of weighted automata are studied in \cite{Droste2009}.
The modules represent what our programs act on -- they are a required generalization to study formal languages as \cref{ex:mixed-formal-languages-cocontinuity-counterexample} shows.
We present everything in parallel to \cref{app:sec:srings}.
\begin{definition}[Module over a Monoid]
	\label{app:def:mon-module}
	%
	\def\smop{\otimes}
	\def\msop{\otimes}
	Let $\monoiddef$ be a monoid.
	A \emph{(left) $\Monoid$-module} $\smoduledef$ is a \emph{commutative monoid} $(\smodule,\, {\madd},\, {\mnull})$ equipped with a (left) action called \emph{scalar multiplication} $\smop\colon \monoid \times \smodule \to \smodule$, such that
	\begin{enumerate}
		\item
			the scalar multiplication $\smop$ is \emph{associative}, i.e.\ for all $a, b \in \monoid$ and $v \in \smodule$,
			\[
				(a \smul b) \smop v \eeq a \smop (b \smop v) ~,
			\]

		\item
			the scalar multiplication $\smop$ is \emph{distributive}, i.e.\ for all $a \in \monoid$ and $v, w \in \smodule$,
			\[
				a \smop (v \madd w) \eeq (a \smop v) \madd (a \smop w)
				~,
			\]
		\item
			the monoid's one $\sone$ is \emph{neutral} and the module's zero $\snull$ \emph{annihilates}, i.e.\ for all $a \in \monoid$ and $v \in \smodule$,
			\[
				\monone \smop v \eeq v
				\qqand
				a \smop \mnull \eeq \mnull
				~.
				\tag*{$\triangle$}
			\]
	\end{enumerate}
\end{definition}
\noindent{}%
To simplify language we speak of modules (and forget about the \enquote{over a monoid} part); this should not be confused with a module over a ring.
We emphasize that all the results developed in this paper apply to the important special case where the monoid and the module together form a \emph{semiring}:
The multiplication $\smul$ of a semiring $\sringdef$ is then the left-action $\otimes\colon S \times S \to S$ of the multiplicative monoid of $(\sring,\, {\smul},\, \sone)$ to the additive monoid $(\sring,\, {\sadd},\, \snull)$.
We also do not differentiate multiplication $\smul$ and left-action $\otimes$ and write $\smop$ instead of $\otimes$ from now on -- both are associative and the operation should be clear from the rightmost multiplicant's type.

Analogous to \cref{app:thm:function-sring}, we also can construct a module of functions via point-wise lifting of the operations.
\begin{lemma}[Module of Module-valued Functions]
	\label{app:thm:function-smodule}
	Let $\Smodule = (\smodule,\, {\madd_{\Smodule}},\, {\mnull_{\Smodule}},\, {\smop_{\Smodule}})$ be an $\Monoid$-module and $X$ be a non-empty set.
	Then $\Smodule^X \coloneqq (\smodule^X,\, {\madd},\, {\mnull},\, {\smop})$, where $\smodule^X$ is the set of functions of type $X \to \smodule$ and for all $a \in \monoid$, $u, v \in \smodule^X$,%
	\begin{align*}
		u \madd v &\qqcoloneqq \qlam{x} u(x) ~{}\madd_{\Smodule}{}~ v(x) ~,\\
		a \smop v &\qqcoloneqq \qlam{x} a ~{}\smop_{\Smodule}{}~ v(x) ~,\\
		\mnull &\qqcoloneqq \qlam{x} \mnull_{\Smodule} ~,
	\end{align*}%
	also forms a $\Monoid$-module which we call the \emph{lifting of $\Smodule$ with respect to $X$}.\qedtriangle
\end{lemma}%
\noindent{}%
%

Analogous to \cref{app:def:sring-natord}, we speak of naturally ordered modules $\Smodule$.
\begin{definition}[Natural Order]
	\label{app:def:smodule-natord}
	The relation $\natord$ is defined for all $a, b \in \smodule$ by
	\[
		a \nnatord b \qqiff \exists\, c \in \smodule \colonq a \sadd c \eeq b .
	\]
	If $\natord$ is a partial order, then we call $\Smodule$ \emph{naturally ordered} and $\natord$ the \emph{natural order} on $\Smodule$.
	\hfill$\triangle$
\end{definition}
\noindent{}%
Similar to \cref{app:thm:sring-natordprops}, we have for modules:
\begin{lemma}[Least Elements and Monotonicity of Algebraic Operations]
	\label{app:thm:smodule-natordprops}
	Let $\Smodule$ be naturally ordered.
	Then,%
	\begin{enumerate}
		\item
			$\mnull$ is the \emph{unique} least element, and
			
		\item
			$\madd$ and $\smop$ are monotone, i.e. for all $a \in \sring$ and $u,v,w \in \smodule$,%
			\[
				v \nnatord w
				\qqimplies
				u \madd v \nnatord u \madd w
				\qand
				a \smop v \nnatord a \smop w
				~.
			\]
	\end{enumerate}%
\end{lemma}%
\begin{proof}
	The first two statements directly follow from the natural order.
	For the last one, distributivity is additionally required.
\end{proof}
\noindent{}%
If the natural order $\natord$ moreover has a greatest element, then this is unique and denoted $\mTop_\Smodule$, where we drop the subscript $\Smodule$ whenever it is clear from the context.

%
As for semirings, we also want to apply fixed point theory to modules.
Inspired by \cref{app:def:continuous-sring}, we define the following notion of $\omega$-continuous modules.
%
\begin{definition}[$\omega$-continuous module]
    \label{app:def:continuous-smodule}
	A $\Monoid$-module $\Smodule$ is \emph{$\omega$-continuous}, if $(\smodule,\natord)$ is an $\omega$-cpo and addition and scalar multiplication with constants are $\omega$-continuous functions, i.e. for all $a \in \monoid$, $u \in \smodule$ and all ascending chains $(v_i)_{i \in \Nats}$ in $\smodule$, we require
    \[
        u \mmadd \bigjoin_{i \in \Nats} v_i \qeq \bigjoin_{i \in \Nats} \pr{ a \madd v_i }
		\qqand
        a \ssmop \bigjoin_{i \in \Nats} v_i \qeq \bigjoin_{i \in \Nats} \pr{ a \smop v_i }
        ~.
    \]
    Dually, $\Smodule$ is \emph{$\omega$-cocontinuous} if $(\smodule,\natord)$ is an $\omega$-cocpo and addition/scalar multiplication with constants are $\omega$-cocontinuous functions.
    $\Smodule$ is \emph{$\omega$-bicontinuous} if it is both $\omega$-continuous and -cocontinuous.
	\hfill$\triangle$
\end{definition}
\noindent{}%
Moreover, if $\Smodule$ is $\omega$-(co)continuous, then the module $\Smodule^X$ of functions from a set $X$ to $\Smodule$ (cf.\ \cref{app:thm:function-smodule}) is $\omega$-(co)continuous as well.

As for semirings, $\omega$-continuity allows to define \emph{countably infinite sums}.
We call a module \emph{$\omega$-finitary} if the infinite sum is defined as follows:
Given a family $(v_i)_{i \in I}$ in $\smodule$ over a \emph{countable} index set $I$,
\begin{align}
    \label{eq:smodule-infinite-sums}
	\mbigadd_{i \in I} v_i
	\qcoloneqq \bigjoin_{\substack{F \subseteq I,\\F \text{ finite} }} \sbigadd_{i \in F} v_i
	\qcoloneqq \bigjoin_{n \in \Nats} \mbigadd_{i \leq n} w_i
	~,
\end{align}
where $(b_i)_{i \in \Nats}$ is any $\Nats$-indexed family such that there exists a bijection $\tau \colon \Nats \to I$ with $b_i = a_{\tau(i)}$ for all $i \in \Nats$.
Again, from $\omega$-continuity it follows that an \emph{extended distributive law} is satisfied:
For all $a \in \monoid$,
\[
	a \smop \pr{ \sbigadd_{i \in I} v_i } \eeq \sbigadd_{i \in I} a \smop v_i
\]


\newpage
\section{Proofs of Preliminaries \ref{app:preliminaries}}
\label{app:proofs2}

\subsection{Proof of Theorem \ref{app:thm:kleene-park}}

We only show \cref{it:lfp} which implies \cref{it:gfp} by reversing the order.
The first part of \cref{it:lfp} is simply an instance of the classic Kleene Fixed Point Theorem.
The second part---the Park Induction principle---can be seen as follows.
First we show by induction that $f^n(\bot) \ord a$ for all $n \in \Nats$:
For $n = 0$ we have $f^0(\bot) = \bot \ord a$.
For $n \geq 0$, we have by the I.H.\ and monotonicity of $f$ that
\[
    f^{n+1}(\bot)
    \qeq  f(f^n(\bot))
    \quad \ord \quad f(a)
    \quad \ord \quad a
    ~.
\]
By Kleene Fixpoint Theorem and the definition of suprema it follows that
\[
    \lfpop f
    \qeq
    \bigjoin_{n \in \Nats} f^n(\bot)
    \quad \ord \quad 
    a
    ~.\tag*{\qed}
\]

\subsection{Proof that Eq.\ \ref{eq:sring-infinite-sums} and Eq.\ \ref{eq:sring-partial-sums} (notions of $\boldsymbol{\omega}$-finitary) coincide}
\label{app:proof-finitary-sums-coincide}

\begin{lemma}
	\label{thm:proof-finitary-sums-coincide}
	Let $\sringdef$ be an $\omega$-continuous semiring, $(a_i)_{i \in I}$ a family in $\sring$ over a \emph{countably infinite} index set $I$, and $(b_j)_{j \in \Nats}$ a family in $\sring$ such that there exists a bijection $\tau \colon \Nats \to I$ with $b_j = a_{\tau(j)}$ for all $j \in \Nats$.
	Then,
	\begin{align*}
		\sbigadd_{i \in I} a_i
		\qeq \bigjoin_{\substack{F \subseteq I,\\F \text{ finite} }} \sbigadd_{i \in F} a_i
		\qeq \bigjoin_{n \in \Nats} \sbigadd_{j \leq n} b_j
		~.
	\end{align*}
	In particular, $\sbigadd_{i \in I} a_i$ is well-defined and \cref{eq:sring-infinite-sums} is compatible to \cref{eq:sring-partial-sums}.
\end{lemma}
\begin{proof}
	First, as $(\sbigadd_{j \leq n} b_j)_{n \in \Nats}$ is an $\omega$-chain under the natural order $\natord$.
	Thus the supremum is well defined in the $\omega$-continuous semiring.

	Next, let $F \subseteq I$ be a finite subset.
	We denote the image of $F$ under $\tau$ as $\tau(F) \coloneqq \Set{ \tau(i) }{ i \in F }$.
	Next, define the maximal index $n_F \coloneqq \max \tau(F) \in \Nats$ corresponding to $F$ (with respect to $\tau$).
	Then, we get the partition
	\[
		\Set{ j \in \Nats }{ 0 \leq j \leq n_F } \qeq \tau(F) ~\uplus~ \Set{ j \in \Nats }{ j \leq n_F, j \notin \tau(F) } ~.
	\] 
	By definition of the natural order $\natord$ it follows
	\[
		\sbigadd_{i \in F} a_i
		\eeq \sbigadd_{i \in F} b_{\tau(i)}
		\eeq \sbigadd_{j \in \tau(F)} b_j
		\qnatord \sbigadd_{j \in \tau(F)} b_j \ssadd \sbigadd_{\substack{j \leq n_F \\ j \notin \tau(F)}} b_{j}
		\eeq \sbigadd_{j \leq n_F} b_j
		\qnatord \bigjoin_{n \in \Nats} \sbigadd_{j \leq n} b_j
		~,
	\]
	and hence (if the left supremum exists)
	\[
		\bigjoin_{\substack{F \subseteq I,\\F \text{ finite} }} \sbigadd_{i \in F} a_i
		\qnatord \bigjoin_{n \in \Nats} \sbigadd_{j \leq n} b_j
		~.
	\]

	On the other hand, we denote for $n \in \Nats$ the preimage of $\bar{n} \coloneqq \Set{ j \in \Nats }{ 0 \leq j \leq n }$ under $\tau$ as $F_n \coloneqq \Set{ i \in F }{ \tau(i) \leq n }$.
	As this always is a finite subset of $I$,
	\[
		\bigjoin_{n \in \Nats} \sbigadd_{j \leq n} b_j
		\eeq \bigjoin_{n \in \Nats} \sbigadd_{j \leq n} a_{\tau^{-1}(j)}
		\eeq \bigjoin_{n \in \Nats} \sbigadd_{i \in F_n} a_i
		\qnatord \bigjoin_{\substack{F \subseteq I,\\F \text{ finite} }} \sbigadd_{i \in F} a_i
		~.
	\]
	Combining both inequalities, we obtain exactly the equality claimed above.
\end{proof}

\subsection{Problems with the Semiring of infinite words}
\label{app:proofs2:mixed-formal-languages-cocontinuity-counterexample}
In order to extend $\finWordLangSring{\ab}$ with $\omega$-words (i.e. words of countably infinite length) to obtain a semiring of \emph{mixed} languages, i.e. subsets of $\ab^\infty = \ab^* \cup \ab^\omega$, one might be tempted to define the concatenation of languages $L_1, L_2 \subseteq \ab^* \cup \ab^\omega$ as follows:
Partition $L_1 = K_1 \cup M_1$ where $K_1 \subseteq \ab^*$ are the finite and $M_1 \subseteq \ab^\omega$ are the $\omega$-words of $L_1$.
We set
\[
	L_1 \langConcat L_2 \qcoloneqq
	\begin{cases}
		K_1 \langConcat L_2  \cup M_1 &\text{if } L_2 \neq \emptyset ~,\\
		\emptyset &\text{if } L_2 = \emptyset ~.
	\end{cases}
\]
Intuitively, concatenating $\omega$-words from the left is absorptive.
This way, $(2^{\ab^* \cup \ab^\omega},\, {\cup},\, {\cdot},\, {\emptyset},\, {\set{\epsilon}})$ indeed is a $\omega$-continuous semiring.
However, the multiplication $\langConcat$ is not $\omega$-cocontinuos as the following example shows.
\begin{example}[Counterexample]
	\label{ex:mixed-formal-languages-cocontinuity-counterexample}
	Consider the singular alphabet $\ab = \set{ a }$.
	Define the descending $\omega$-chain $(L_n)_{n \in \Nats}$ where $L_n \coloneqq \Set{ a^i }{ i \geq n }$.
	Then,
	\begin{align*}
		\pr{ \bigmeet_{n \in \Nats} L_n } \langConcat \set{ a^\omega }
		\qeq \pr{ \bigcap_{n \in \Nats} L_n } \langConcat \set{ a^\omega }
		\qeq \emptyset \langConcat \set{ a^\omega }
		\qeq \emptyset ~.
	\end{align*}
	On the other hand,
	\begin{align*}
		\bigmeet_{n \in \Nats} ( L_n \langConcat \set{ a^\omega } )
		\qeq \bigcap_{n \in \Nats} ( L_n \langConcat \set{ a^\omega } )
		\qeq \bigcap_{n \in \Nats} \set{ a^\omega }
		\qeq \set{ a^\omega } ~.
	\end{align*}
	Hence, $\langConcat$ is not $\omega$-cocontinuos.
\end{example}
Even as a semimodule multiplication with a semiring element would not be $\omega$-cocontinuos.
An approach to resolve this problem is the notion of \emph{star semiring - omega semimodule pairs} or \emph{quemirings} \cite[Ch.\ 3]{Droste2009}.
Using semirings has the benefit that it allows to study matricial theories \cite[Ch.\ 1 \& 2]{Droste2009}.

Alternatively, we have the following approach.
Consider the alphabet $\ab$.
We obtain the \emph{word monoid} $\finWordMonoidDef{\ab}$ where $\langConcat$ is the usual concatenation.
Next, we define a module $\mixedWordLangSmodule{\ab} \coloneqq ()$ over this monoid.

\subsection{Proof that $\mixedWordLangSmodule{\ab}$ is a $\omega$-bicontinuous module}
\label{app:proofs2:mixed-formal-languages-module}

Let $\ab$ be a non-empty alphabet.
The \emph{word monoid} is $\finWordMonoidDef{\ab}$, where $\langConcat$ is the usual \emph{concatenation} and $\epsilon$ is the \emph{empty word}.
As a shorthand denote $\ab^\infty \coloneqq \ab^* \cup \ab^\omega$.
For a word $v \in \ab^*$ (from the monoid) and a formal language $L \subseteq \ab^\infty$ (from the module) we define their \emph{concatenation} $v \langConcat L \coloneqq \Set{vw}{w \in L}$ as usual.
Clearly, the \emph{module of mixed languages} $\mixedWordLangSmodule{\ab} \coloneqq (2^{\ab^\infty},\, {\cup},\, {\emptyset},\, {\langConcat})$ is a module over this monoid.

\begin{claim}
	This module is $\omega$-bicontinuous.
\end{claim}
\begin{proof}
	The natural order $\natord$ simply is set inclusion $\subseteq$ of languages.
	The infimum $\meet$ of a descending $\omega$-chain simply is the intersection $\cap$, and the supremum $\join$ of an ascending $\omega$-chain simply is the union $\cup$, i.e.\ subsets of $\ab^\infty$ by definition.
	Let $(A_i)_{i \in \Nats}$ be an ascending $\omega$-chain ($A_i \subseteq A_{i+1}$), $(D_i)_{i \in \Nats}$ be an descending $\omega$-chain ($D_i \supseteq D_{i+1}$), $L \subseteq \ab^\infty$, and $v \in \ab^*$ then
	\begin{itemize}
	\item Clearly, addition $\cup$ is $\omega$-continuous:
		\[
			L \cup \bigcup_{i \in \Nats} A_i \qeq \bigcup_{i \in \Nats} (L \cup A_i) ~.
		\]

	\item Also, addition $\cup$ is $\omega$-cocontinuous:
		\[
			L \cup \bigcap_{i \in \Nats} D_i \qeq \bigcap_{i \in \Nats} (L \cup D_i) ~.
		\]

	\item More interestingly, scalar multiplication $\langConcat$ also is $\omega$-continuous:
		\[
			v \langConcat \bigcup_{i \in \Nats} A_i \qeq \bigcup_{i \in \Nats} v \langConcat A_i ~.
		\]

		Let $w \in v \langConcat \bigcup_{i \in \Nats} A_i$, then there is an $u \in \bigcup_{i \in \Nats} A_i$ such that $w = vu$.
		But then there is an $i \in \Nats$ with $u \in A_i$ and hence $w = vu \in v \langConcat A_i$.

		On the other hand, let $w \in \bigcup_{i \in \Nats} v \langConcat A_i$, then there is an $j \in \Nats$ such that $w = vu$ for some $u \in A_j$.
		But this implies $u \in \bigcup_{i \in \Nats} A_i$ and hence $w = vu \in v \langConcat \bigcup_{i \in \Nats} A_i$.

	\item Finally, scalar multiplication $\langConcat$ is $\omega$-cocontinuous:\footnote{This is in general false if one were to allow languages (as opposed to single words) on the left hand side, i.e.\ $L_1 \langConcat L_2$ for $L_1 \subset \ab^*$ and $L_2 \subset \ab^\infty$}
		\[
			v \langConcat \bigcap_{i \in \Nats} D_i \qeq \bigcap_{i \in \Nats} v \langConcat D_i ~.
		\]

		Let $w \in v \langConcat \bigcap_{i \in \Nats} D_i$, then there is an $u \in \bigcap_{i \in \Nats} D_i$ such that $w = vu$.
		But then $u \in D_i$ and hence $w = vu \in v \langConcat D_i$ for all $i \in \Nats$.

		On the other hand, let $w \in \bigcap_{i \in \Nats} v \langConcat D_i$.
		Then, there is an $u_i \in L_i$ with $w = v u_i$ for all $i \in \Nats$.
		But as $w$ is fixed, this implies $u_i = u_j$ for all $i, j \in \Nats$.
		Hence, there is an $u = u_1 \in \bigcap_{i \in \Nats} D_i$ with $w = vu$, i.e.\ $v \langConcat \bigcap_{i \in \Nats} D_i$.
	\end{itemize}
\end{proof}

\newpage
\section{Proofs of Section \ref{sec:weighted-semantics}}
\label{app:proofs4}

First of all, note that $\iverson{\guard} \ivop$ distributes of $\madd$:
Let $f, g \in \Weightings$ and $\state \in \States$.
\begin{itemize}
\item Case $\state \models \guard$.
	\begin{align*}
		(\iverson{\guard} \ivop (f \madd g))(\state)
		\eeq f(\state) \madd g(\state)
		\eeq (\iverson{\guard} \ivop f)(\state) \mmadd (\iverson{\guard} \ivop f)(\state)
		\eeq (\iverson{\guard} \ivop f \madd \iverson{\guard} \ivop g))(\state)
	\end{align*}

\item Case $\state \not\models \guard$.
	\begin{align*}
		(\iverson{\guard} \ivop (f \madd g))(\state)
		\eeq \mnull
		\eeq \mnull(\state) \mmadd \mnull(\state)
		\eeq (\iverson{\guard} \ivop f)(\state) \mmadd (\iverson{\guard} \ivop f)(\state)
		\eeq (\iverson{\guard} \ivop f \madd \iverson{\guard} \ivop g)(\state)
	\end{align*}
\end{itemize}

\subsection{Proof of Theorem \ref{thm:wp-continuous}}
\label{app:proof:wp-continuous}

We will use the following:
\begin{lemma}
    \label{thm:madd-cont-both-args}
    Let $(a_i)_{i \in \Nats}$ and $(b_i)_{i \in \Nats}$ be ascending $\omega$-chains in an $\omega$-continuous module $\Smodule$.
    Then
    \[
        \bigjoin_{i \in \Nats} (a_i \madd b_i) \eeq \bigjoin_{i \in \Nats} a_i \mmadd \bigjoin_{i \in \Nats} b_i ~.
        \tag*{$\triangle$}
    \]
\end{lemma}
\begin{proof}
	Follows because addition with constants is $\omega$-continuous (see \cref{app:def:continuous-smodule}) and applying e.g.~\cite[Ch.\ Domain Theory, Lem.\ 3.2.6]{DBLP:books/lib/Abramsky94}.
\end{proof}

\begin{theorem}[\cref{thm:wp-continuous}]
	Let the monoid module $\Smodule$ over $\Monoid$ be $\omega$-continuous.
    For all $\Monoid$-$\wgcl$ programs $C$ and $\omega$-continuous $\Monoid$-modules $\Smodule$ the weighting transformer $\wpC{C}\colon \Weightings \to \Weightings$ is a well-defined $\omega$-continuous function.
	If $C$ is of the form $\WHILEDO{\guard}{C'}$, the least fixed point is
    \[
		\wp{C}{f} \qeq \bigjoin_{i \in \Nats} \charwp{\guard}{C'}{f}^i(\mnull) ~.
		\tag*{$\triangle$}
    \]
\end{theorem}
\begin{proof}
    We employ induction on the structure of $C$.
    Let $(f_i)_{i \in \Nats}$ be an ascending $\omega$-chain in $\Weightings$.
    Because $\Smodule$ and $\Weightings$ are $\omega$-continuous mondules, all the following joins exist.
    
    The following forms the induction base.
    \begin{itemize}
	\item
        The program $C$ is of the form $\ASSIGN{x}{E}$.
        \begin{align*}
			& \wp{C}{\bigjoin_{i \in \Nats} f_i} \\
			\cmmnt{Def.\ of $\wpsymbol$} \\
			\eeq& \Big( \bigjoin_{i \in \Nats} f_i \Big) \subst{x}{E} \\
			\cmmnt{Def.\ of $\subst{x}{E}$} \\
			\eeq& \lam{\sigma} \Big( \bigjoin_{i \in \Nats} f_i \Big) \bigl( \sigma\statesubst{x}{\eval{E}{\sigma}} \bigr) \\
			\cmmnt{Def.\ of $\bigjoin$} \\
			\eeq& \lam{\sigma} \bigjoin_{i \in \Nats} \Bigl( f_i \bigl( \sigma\statesubst{x}{\eval{E}{\sigma}} \bigr) \Bigr) \\
			\cmmnt{Def.\ of $\bigjoin$} \\
			\eeq& \bigjoin_{i \in \Nats} \Bigl( \lam{\sigma} f_i \bigl( \sigma\statesubst{x}{\eval{E}{\sigma}} \bigr) \Bigr) \\
			\cmmnt{Def.\ of $\subst{x}{E}$} \\
			\eeq& \bigjoin_{i \in \Nats} (f_i\subst{x}{E}) \\
			\cmmnt{Def.\ of $\wpsymbol$} \\
			\eeq& \bigjoin_{i \in \Nats} \wp{C}{f_i}
        \end{align*}
        
	\item
        The program $C$ is of the form $\WEIGH{a}$.
        This is an immediate consequence of \cref{app:thm:compositional-continuity} and the $\omega$-continuity of $\smop$.
        \begin{align*}
			& \wp{C}{\bigjoin_{i \in \Nats} f_i} \\
			\cmmnt{Def.\ of $\wpsymbol$} \\
			\eeq& a \smop \Big( \bigjoin_{i \in \Nats} f_i \Big) \\
			\cmmnt{$\smop$ is $\omega$-continuous in the second argument} \\
			\eeq& \bigjoin_{i \in \Nats} (a \smop f_i) \\
			\cmmnt{Def.\ of $\wpsymbol$} \\
			\eeq& \bigjoin_{i \in \Nats} \wp{C}{f_i}
        \end{align*}
    \end{itemize}
    
    The following forms the induction step.
    Hence, for each deconstruction of $C$, we assume both $\wpC{C_1}$ and $\wpC{C_2}$ to be $\omega$-continuous as our induction's hypothesis.
    Notice that the subprograms $C_1$, $C_2$, $C'$ are always shorter than $C$.
    \begin{itemize}
	\item
        The program $C$ is of the form $\COMPOSE{C_1}{C_2}$.
        This is an immediate consequence of \cref{app:thm:compositional-continuity}.
        \begin{align*}
			& \wp{C}{\bigjoin_{i \in \Nats} f_i} \\
			\cmmnt{Def.\ of $\wpsymbol$} \\
			\qeq& \wp{C_1}{ \wp{C_2}{ \bigjoin_{i \in \Nats} f_i } } \\
			\cmmnt{Induction Hypothesis} \\
			\qeq& \wp{C_1}{ \bigjoin_{i \in \Nats} \wp{C_2}{f_i} } \\
			\cmmnt{Induction Hypothesis} \\
			\qeq& \bigjoin_{i \in \Nats} \wp{C_1}{ \wp{C_2}{f_i} } \\
			\cmmnt{Def.\ of $\wpsymbol$} \\
			\qeq& \bigjoin_{i \in \Nats} \wp{C}{f_i}
        \end{align*}
        
	\item
        The program $C$ is of the form $\ITE{\guard}{C_1}{C_2}$.
        This is an immediate consequence of \cref{app:thm:compositional-continuity} and the module's $\omega$-continuity.
        \begin{align*}
			& \wp{C}{\bigjoin_{i \in \Nats} f_i} \\
			\cmmnt{Def.\ of $\wpsymbol$} \\
			\eeq& \iverson{\guard} \ivop \wp{C_1}{\bigjoin_{i \in \Nats} f_i} \mmadd \iverson{\neg\guard} \ivop \wp{C_2}{\bigjoin_{i \in \Nats} f_i} \\
			\cmmnt{Induction Hypothesis} \\
			\eeq& \iverson{\guard} \ivop \Big( \bigjoin_{i \in \Nats} \wp{C_1}{f_i} \Big) \mmadd \iverson{\neg\guard} \ivop \Big( \bigjoin_{i \in \Nats} \wp{C_2}{f_i} \Big) \\
			\cmmnt{$\mnull \ivop$, $\mTop \ivop$ are $\omega$-continuous} \\
			\eeq& \Big( \bigjoin_{i \in \Nats} \iverson{\guard} \ivop \wp{C_1}{f_i} \Big) \mmadd \Big( \bigjoin_{i \in \Nats} \iverson{\neg\guard} \ivop \wp{C_2}{f_i} \Big) \\
			\cmmnt{$\madd$ is $\omega$-continuous, \cref{thm:madd-cont-both-args}} \\
			\eeq& \bigjoin_{i \in \Nats} \big( \iverson{\guard} \ivop \wp{C_1}{f_i} \mmadd \iverson{\neg\guard} \ivop \wp{C_2}{f_i} \big) \\
			\cmmnt{Def.\ of $\wpsymbol$} \\
			\eeq& \bigjoin_{i \in \Nats} \wp{C}{f_i}
        \end{align*}
        
	\item
        The program $C$ is of the form $\BRANCH{C_1}{C_2}$.
        This is an immediate consequence of \cref{app:thm:compositional-continuity} and the $\omega$-continuity of $\madd$.
        \begin{align*}
			& \wp{C}{\bigjoin_{i \in \Nats} f_i} \\
			\cmmnt{Def.\ of $\wpsymbol$} \\
			\eeq& \wp{C_1}{ \bigjoin_{i \in \Nats} f_i } \mmadd \wp{C_2}{ \bigjoin_{i \in \Nats} f_i } \\
			\cmmnt{Induction Hypothesis} \\
			\eeq& \Big( \bigjoin_{i \in \Nats} \wp{C_1}{f_i} \Big) \mmadd \Big( \bigjoin_{i \in \Nats} \wp{C_2}{f_i} \Big) \\
			\cmmnt{$\madd$ is $\omega$-continuous, \cref{thm:madd-cont-both-args}} \\
			\eeq& \bigjoin_{i \in \Nats} \big( \wp{C_1}{f_i} \mmadd \wp{C_2}{f_i} \big) \\
			\cmmnt{Def.\ of $\wpsymbol$} \\
			\eeq& \bigjoin_{i \in \Nats} \wp{C}{f_i}
        \end{align*}

	\item
        The program $C$ is of the form $\WHILEDO{\guard}{C'}$.
		First we show the loop-characteristic function
		\[
			\charwp{\guard}{C'}{f}\colon \quad \Weightings \to \Weightings, \quad g ~{}\mapsto{}~ \iverson{\neg\guard} \ivop f \mmadd \iverson{\guard} \ivop \wp{C'}{g}
		\]
		to be $\omega$-continuous in both $f$ and its argument $g$.
		\begin{align*}
			& \charwp{\guard}{C'}{\bigjoin_{i \in \Nats} f_i} \\
			\cmmnt{Def.\ of $\charwp{\guard}{C'}{f}$} \\
			\eeq& \lam{g} \iverson{\neg\guard} \ivop \Big( \bigjoin_{i \in \Nats} f_i \Big) \mmadd \iverson{\guard} \ivop \wp{C'}{g} \\
			\cmmnt{$\mnull \ivop$, $\mTop \ivop$ are $\omega$-continuous} \\
			\eeq& \lam{g} \Big( \bigjoin_{i \in \Nats} \iverson{\neg\guard} \ivop f_i \Big) \mmadd \iverson{\guard} \ivop \wp{C'}{g} \\
			\cmmnt{$\madd$ is $\omega$-continuous} \\
			\eeq& \lam{g} \bigjoin_{i \in \Nats} \Big( \iverson{\neg\guard} \ivop f_i \mmadd \iverson{\guard} \ivop \wp{C'}{g} \Big) \\
			\cmmnt{Def.\ of $\join$ on $\Sring^{\Weightings}$} \\
			\eeq& \bigjoin_{i \in \Nats} \Big( \lam{g} \iverson{\neg\guard} \ivop f_i \mmadd \iverson{\guard} \ivop \wp{C'}{g} \Big) \\
			\cmmnt{Def.\ of $\charwp{\guard}{C'}{f}$} \\
			\eeq& \bigjoin_{i \in \Nats} \charwp{\guard}{C'}{f_i}
		\end{align*}
		Let $(g_i)_{i \in \Nats}$ be an (ascending) $\omega$-chain in $\Weightings$.
		\begin{align*}
			& \charwp{\guard}{C'}{f}\Big( \bigjoin_{i \in \Nats} g_i \Big) \\
			\cmmnt{Def.\ of $\charwp{\guard}{C'}{f}$} \\
			\eeq& \iverson{\neg\guard} \ivop f \mmadd \iverson{\guard} \ivop \wp{C'}{\bigjoin_{i \in \Nats} g_i} \\
			\cmmnt{Induction Hypothesis} \\
			\eeq& \iverson{\neg\guard} \ivop f \mmadd \iverson{\guard} \ivop \Big( \bigjoin_{i \in \Nats} \wp{C'}{g_i} \Big) \\
			\cmmnt{$\mnull \ivop$, $\mTop \ivop$ are $\omega$-continuous} \\
			\eeq& \iverson{\neg\guard} \ivop f \mmadd \bigjoin_{i \in \Nats} \big( \iverson{\guard} \ivop \wp{C'}{g_i} \big) \\
			\cmmnt{$\madd$ is $\omega$-continuous} \\
			\eeq& \bigjoin_{i \in \Nats} \big( \iverson{\neg\guard} \ivop f \mmadd \iverson{\guard} \ivop \wp{C'}{g_i} \big) \\
			\cmmnt{Def.\ of $\charwp{\guard}{C'}{f}$} \\
			\eeq& \bigjoin_{i \in \Nats} \charwp{\guard}{C'}{f}(g_i)
		\end{align*}
		By \cref{app:thm:kleene-park} the fixed point $\wp{\WHILEDO{\guard}{C'}}{f}$ is well defined and equals the claimed expression.
		It is $\omega$-continuous in $f$ by $\omega$-continuity of $f \mapsto \charwp{\guard}{C'}{f}$ and the fixed-point operator $\lfpop$, e.g.~\cite[Ch.\ Domain Theory, Thm.\ 2.1.19]{DBLP:books/lib/Abramsky94}.
    \end{itemize}
\end{proof}

\subsection{Proof of Theorem \ref{thm:wp-operational-semantics} (Soundness of $\wpsymbol$ w.r.t.\ operational semantics)}
\label{app:proof:wp-soundness}

The following proof is based on \cite[Appendix B]{DBLP:journals/corr/abs-1802-10467}.

\begin{theorem}[\cref{thm:wp-operational-semantics}]
	Let the monoid module $\Smodule$ over $\Monoid$ be $\omega$-continuous.
    For any $\wgcl$ program $C$, initial state $\state \in \States$, and post-weighting $f \in \Weightings$,
    \begin{align*}
		\wp{C}{f}(\state) \qeq \mbigadd_{\compPath \in \termPathsStartingIn{\sosState{C}{\state}{n}{v}}} \pathWeight(\compPath) \smop f(\lastState(\compPath)) ~.
    \end{align*}
\end{theorem}
\begin{proof}
    We use a few auxiliary definitions and lemmas that can be found below.
    By \cref{thm:least-wgcl-functional} $\ext{\opsymbol}$ is the least $\wgcl$-functional.
    As $\ext{\wpsymbol}$ also is a $\wgcl$-functional due to \cref{thm:wp-functional}, we get $\ext{\opsymbol} \ewtord \ext{\wpsymbol}$.
    On the other hand, $\ext{\wpsymbol} \ewtord \ext{\opsymbol}$ by \cref{thm:wp-less-than-op}.
    Both imply $\ext{\wpsymbol} = \ext{\opsymbol}$.
    Now, $\wpsymbol = \opsymbol$ due to \cref{def:weighting-transformer}.
    This is the above claim by \cref{def:operational-functional}.
\end{proof}

\begin{definition}
    \label{def:weighting-transformer}
    A map $\Phi\colon \wgcl \to (\Weightings \to \Weightings)$ is called \emph{weighting transformer}.
    The corresponding \emph{extended weighting transformer} $\ext{\Phi}\colon (\wgcl \cup \set{\termState}) \to (\Weightings \to \Weightings)$ is defined via
    \[
		\ewt{\Phi}{C}{f} \qcoloneqq \begin{cases}
			f & \text{ if } C = \termState \\
			\wt{\Phi}{C}{f} & \text{ otherwise }
		\end{cases}
    \]
    for all $C \in \wgcl \cup \set{\termState}$ and $f \in \Weightings$.
    Define the partial order $\ewtord$ for weighting transformers $\Phi, \Psi$ via
    \[
		\ext{\Phi} \ewtord \ext{\Psi}
		\qqiff
		\wt{\Phi}{C}{f} \natord \wt{\Psi}{C}{f} \text{ for all } C \in \wgcl, f \in \Weightings ~.
		\tag*{$\triangle$}
    \]
\end{definition}

\begin{definition}
    \label{def:wgcl-functional}
    An extended weighting transformer $\ext{\Phi}$ is called \emph{$\wgcl$-functional} if for all $\wgcl$ programs $C$, postweightings $f \in \Weightings$, states $\state \in \States$, $n \in \Nats$, and $\beta \in \set{L,R}^{\ast}$,
    \[
		\ewt{\Phi}{C}{f}(\state)
		\qeq \mbigadd_{\sosTrans{\sosState{C}{\state}{n}{\beta}}{a}{\sosState{C'}{\state'}{n+1}{\beta'}}}
			a \smop \ewt{\Phi}{C'}{f}(\state') ~.
		\tag*{$\triangle$}
    \]
\end{definition}

\begin{definition}
    \label{def:operational-functional}
    The map $\opsymbol\colon \wgcl \to (\Weightings \to \Weightings)$ is defined for any $\wgcl$ program $C$, state $\state \in \States$, and post-weighting $f \in \Weightings$ via
    \[
		\op{C}{f}(\state)
		\qcoloneqq \mbigadd_{\compPath \in \termPathsStartingIn{\sosState{C}{\state}{0}{\epsilon}}}
			\pathWeight(\compPath) \smop f(\lastState(\compPath)) ~.
		\tag*{$\triangle$}
    \]
\end{definition}

\begin{lemma}
    \label{thm:least-wgcl-functional}
    The map $\ext{\opsymbol}$ is the least $\wgcl$-functional.
\end{lemma}
\begin{proof}
    First, we show that $\ext{\opsymbol}$ is a $\wgcl$-functional.
    Given a configuration $\conf_0 = \sosState{C}{\state}{n}{\beta}$, where $C$ is a $\wgcl$ program and $\state \in \States$ a state, we have
    \begin{align*}
		& \eop{C}{f}(\state) \\
		\cmmnt{Def.\ of $\opsymbol$} \\
		\eeq& \mbigadd_{\compPath \in \termPathsStartingIn{\conf_0}}
			\pathWeight(\compPath) \smop f(\lastState(\compPath)) \\
		\cmmnt{Def.\ of $\pathWeight{\compPath}$} \\
		\eeq& \mbigadd_{\conf_0 \conf_1 \ldots \conf_k \in \termPathsStartingIn{\conf_0}}
			(\weight{\conf_0}{\conf_1} \monop \pathWeight(\conf_1 \ldots \conf_k)) \smop f(\lastState(\conf_1 \ldots \conf_k)) \\
		\cmmnt{split path} \\
		\eeq& \mbigadd_{\sosTrans{\conf_0}{a}{\conf_1}} \mbigadd_{\compPath \in \termPathsStartingIn{\conf_1}}
			a \smop \pathWeight(\compPath) \smop f(\lastState(\compPath)) \\
		\cmmnt{Distributivity} \\
		\eeq& \mbigadd_{\sosTrans{\conf_0}{a}{\conf_1}}
			a \smop \mbigadd_{\compPath \in \termPathsStartingIn{\conf_1}}
				\pathWeight(\compPath) \smop f(\lastState(\compPath)) \\
		\cmmnt{Induction on \cref{fig:sosrules}} \\
		\eeq& \mbigadd_{\sosTrans{\sosState{C}{\state}{n}{\beta}}{a}{\sosState{C'}{\state'}{n+1}{\beta'}}}
			a \smop \mbigadd_{\compPath \in \termPathsStartingIn{\sosState{C'}{\state'}{0}{\epsilon}}}
				\pathWeight(\compPath) \smop f(\lastState(\compPath)) \\
		\cmmnt{Def.\ of $\opsymbol$} \\
		\eeq& \mbigadd_{\sosTrans{\sosState{C}{\state}{n}{\beta}}{a}{\sosState{C'}{\state'}{n+1}{\beta'}}}
			a \smop \eop{C'}{f}(\state')
		~.
    \end{align*}
    
    Now let $\ext{\Phi}$ be another $\wgcl$-functional.
    We want to show that $\ext{\opsymbol} \ewtord \ext{\Phi}$.
    Thus, we perform induction over the maximum length $\ell \in \Nats$ of terminating paths from a configuration $\conf_0 = \sosState{C}{\state}{n}{\beta}$.
    To that end, denote
    \[
		\termPathsStartingIn{\conf_0}^{\leq \ell}
		\qcoloneqq
		\Set{ \conf_0 \ldots \conf_k \in \termPathsStartingIn{\conf_0} }{ k \leq \ell }
    \]
    and
    \[
		\eoplen{\ell}{C}{f}(\state)
		\qcoloneqq
		\mbigadd_{\compPath \in \termPathsStartingIn{\conf_0}^{\leq \ell}} \pathWeight(\compPath) \smop f(\lastState(\compPath))
		~.
    \]
    We prove $\eoplen{\ell}{C}{f} \natord \ewt{\Phi}{C}{f}$ for all $\ell \in \Nats$, then the claim follows because the module $\Smodule$ is $\omega$-finitary and hence
    \[
		\eop{C}{f}
		\eeq \bigjoin_{\ell \in \Nats} \eoplen{\ell}{C}{f}
		\nnatord \ewt{\Phi}{C}{f}
		~.
    \]
    Moreover, we may assume $C \neq \termState$ as otherwise $\eop{C}{f} = f = \ewt{C}{\termState}{f}$.
    
    The following forms the induction base.
    Let $\ell = 0$.
    Then, $\termPathsStartingIn{\conf_0}^{\leq 0} = \emptyset$ and
    \[
		\eoplen{0}{C}{f} \eeq \mnull \nnatord \ewt{\Phi}{C}{f} ~.
    \]
    
    The following forms the induction step.
    Let $\ell \in \Nats$ such that $\eoplen{\ell}{C}{f} \natord \ewt{\Phi}{C}{f}$.
    \begin{align*}
		& \eoplen{\ell + 1}{C}{f}(\state) \\
		\cmmnt{$\wgcl$-functional} \\
		\eeq& \mbigadd_{\sosTrans{\sosState{C}{\state}{n}{\beta}}{a}{\sosState{C'}{\state'}{n+1}{\beta'}}}
			a \smop \eoplen{\ell}{C'}{f}(\state') \\
		\cmmnt{Induction Hypothesis} \\
		\nnatord& \mbigadd_{\sosTrans{\sosState{C}{\state}{n}{\beta}}{a}{\sosState{C'}{\state'}{n+1}{\beta'}}}
			a \smop \ewt{\Phi}{C'}{f}(\state') \\
		\cmmnt{$\wgcl$-functional} \\
		\eeq& \ewt{\Phi}{C}{f}(\state)
		~.
    \end{align*}
\end{proof}

\begin{lemma}
    \label{thm:wp-functional}
    The map $\ext{\wpsymbol}$ is a $\wgcl$-functional.
\end{lemma}
\begin{proof}
    We employ structural induction on the rules from \cref{fig:sosrules} grouped by the structure of $C$.
    
    The following forms the induction base.
    \begin{itemize}
	\item
        The program $C$ is of the form $\ASSIGN{x}{E}$.
        
        \begin{align*}
			& \ewp{C}{f}(\state) \\
			\cmmnt{Def.\ of $\wpsymbol$} \\
			\eeq& f\subst{x}{E}(\state) \\
			\cmmnt{Def.\ of $f\subst{x}{E}$} \\
			\eeq& f(\state\statesubst{x}{\eval{E}{\state}}) \\
			\cmmnt{Def.\ of $\ext{\wpsymbol}$} \\
			\eeq& \ewp{\termState}{f}(\state\statesubst{x}{\eval{E}{\state}}) \\
			\cmmnt{\cref{fig:sosrules} (assign)} \\
			\eeq& \mbigadd_{\sosTrans{\sosState{C}{\state}{n}{\beta}}{a}{\sosState{C'}{\state'}{n+1}{\beta'}}}
				a \smop \ewp{C'}{f}(\state')
        \end{align*}
        
	\item
        The program $C$ is of the form $\WEIGH{a}$.
        
        \begin{align*}
			& \ewp{C}{f}(\state) \\
			\cmmnt{Def.\ of $\wpsymbol$} \\
			\eeq& (a \smop f)(\state) \\
			\cmmnt{Def.\ of $\ext{\wpsymbol}$} \\
			\eeq& a \smop \ewp{\termState}{f}(\state) \\
			\cmmnt{\cref{fig:sosrules} (weight)} \\
			\eeq& \mbigadd_{\sosTrans{\sosState{C}{\state}{n}{\beta}}{a}{\sosState{C'}{\state'}{n+1}{\beta'}}}
				a \smop \ewp{C'}{f}(\state')
        \end{align*}
        
    \end{itemize}
    
    The following forms the induction step.
    \begin{itemize}
	\item
        The program $C$ is of the form $\COMPOSE{C_1}{C_2}$.
        
        There are the following two exclusive cases by the rules from \cref{fig:sosrules}:
        \begin{enumerate}
		\item
            Case $\sosTrans{ \sosState{C_1}{\state}{n}{\beta} }{a}{ \sosState{\termState}{\state'}{n+1}{\beta'} }$.
            
            \begin{align*}
				& \ewp{C}{f}(\state) \\
				\cmmnt{Def.\ of $\wpsymbol$} \\
				\eeq& \ewp{C_1}{\ewp{C_2}{f}}(\state) \\
				\cmmnt{Induction Hypothesis} \\
				\eeq& \mbigadd_{\sosTrans{\sosState{C_1}{\state}{n}{\beta}}{a}{\sosState{\termState}{\state'}{n+1}{\beta'}}}
					a \smop \ewp{\termState}{\ewp{C_2}{f}}(\state') \\
				\cmmnt{Def.\ of $\ext{\wpsymbol}$} \\
				\eeq& \mbigadd_{\sosTrans{\sosState{C_1}{\state}{n}{\beta}}{a}{\sosState{\termState}{\state'}{n+1}{\beta'}}}
					a \smop \ewp{C_2}{f}(\state') \\
				\cmmnt{\cref{fig:sosrules} (seq.\ 1)} \\
				\eeq& \mbigadd_{\sosTrans{\sosState{C}{\state}{n}{\beta}}{a}{\sosState{C'}{\state'}{n+1}{\beta'}}}
					a \smop \ewp{C'}{f}(\state')
            \end{align*}
            
		\item
            Case $\sosTrans{ \sosState{C_1}{\state}{n}{\beta} }{a}{ \sosState{C_1'}{\state'}{n+1}{\beta'} }$.
            
            \begin{align*}
				& \ewp{C}{f}(\state) \\
				\cmmnt{Def.\ of $\wpsymbol$} \\
				\eeq& \ewp{C_1}{\ewp{C_2}{f}}(\state) \\
				\cmmnt{Induction Hypothesis} \\
				\eeq& \mbigadd_{\sosTrans{\sosState{C_1}{\state}{n}{\beta}}{a}{\sosState{C_1'}{\state'}{n+1}{\beta'}}}
					a \smop \ewp{C_1'}{\ewp{C_2}{f}}(\state') \\
				\cmmnt{Def.\ of $\wpsymbol$} \\
				\eeq& \mbigadd_{\sosTrans{\sosState{C_1}{\state}{n}{\beta}}{a}{\sosState{C_1'}{\state'}{n+1}{\beta'}}}
					a \smop \ewp{\COMPOSE{C_1'}{C_2}}{f}(\state') \\
				\cmmnt{\cref{fig:sosrules} (seq.\ 1)} \\
				\eeq& \mbigadd_{\sosTrans{\sosState{C}{\state}{n}{\beta}}{a}{\sosState{C'}{\state'}{n+1}{\beta'}}}
					a \smop \ewp{C'}{f}(\state')
            \end{align*}
        \end{enumerate}
        
	\item
        The program $C$ is of the form $\ITE{\guard}{C_1}{C_2}$.
        
        There are the following two exclusive cases:
        \begin{enumerate}
		\item
            Case $\state \models \guard$.
            
            \begin{align*}
				& \ewp{C}{f}(\state) \\
				\cmmnt{Def.\ of $\wpsymbol$} \\
				\eeq& (\iverson{\guard} \ivop \ewp{C_1}{f} \mmadd \iverson{\neg\guard} \ivop \ewp{C_2}{f})(\state) \\
				\cmmnt{Case $\state \models \guard$} \\
				\eeq& \ewp{C_1}{f}(\state) \\
				\cmmnt{\cref{fig:sosrules} (if)} \\
				\eeq& \mbigadd_{\sosTrans{\sosState{C}{\state}{n}{\beta}}{a}{\sosState{C'}{\state'}{n+1}{\beta'}}}
					a \smop \ewp{C'}{f}(\state')
            \end{align*}
            
		\item
            Case $\state \not\models \guard$.
            
            \begin{align*}
				& \ewp{C}{f}(\state) \\
				\cmmnt{Def.\ of $\wpsymbol$} \\
				\eeq& (\iverson{\guard} \ivop \ewp{C_1}{f} \mmadd \iverson{\neg\guard} \ivop \ewp{C_2}{f})(\state) \\
				\cmmnt{Case $\state \not\models \guard$} \\
				\eeq& \ewp{C_2}{f}(\state) \\
				\cmmnt{\cref{fig:sosrules} (else)} \\
				\eeq& \mbigadd_{\sosTrans{\sosState{C}{\state}{n}{\beta}}{a}{\sosState{C'}{\state'}{n+1}{\beta'}}}
					a \smop \ewp{C'}{f}(\state')
            \end{align*}
            
        \end{enumerate}
        
	\item
        The program $C$ is of the form $\BRANCH{C_1}{C_2}$.
        
        \begin{align*}
			& \ewp{C}{f}(\state) \\
			\cmmnt{Def.\ of $\wpsymbol$} \\
			\eeq& \ewp{C_1}{f}(\state) \mmadd \ewp{C_2}{f}(\state) \\
			\cmmnt{\cref{fig:sosrules} (l.\ branch), (r.\ branch)} \\
			\eeq& \mbigadd_{\sosTrans{\sosState{C}{\state}{n}{\beta}}{a}{\sosState{C'}{\state'}{n+1}{\beta'}}}
				a \smop \ewp{C'}{f}(\state')
        \end{align*}
        
	\item
        The program $C$ is of the form $\WHILEDO{\guard}{C_1}$.
        
        There are the following two exclusive cases:
        \begin{enumerate}
		\item
            Case $\state \models \guard$.
            
            \begin{align*}
				& \ewp{C}{f}(\state) \\
				\cmmnt{Def.\ of $\wpsymbol$} \\
				\eeq& (\iverson{\neg\guard} \ivop f \mmadd \iverson{\guard} \ivop \wp{C_1}{\wp{C}{f}} )(\state) \\
				\cmmnt{Case $\state \models \guard$} \\
				\eeq& \wp{C_1}{\wp{C}{f}}(\state) \\
				\cmmnt{Def.\ of $\wpsymbol$} \\
				\eeq& \wp{\COMPOSE{C_1}{C}}{f}(\state) \\
				\cmmnt{\cref{fig:sosrules} (while)} \\
				\eeq& \mbigadd_{\sosTrans{\sosState{C}{\state}{n}{\beta}}{a}{\sosState{C'}{\state'}{n+1}{\beta'}}}
					a \smop \ewp{C'}{f}(\state')
            \end{align*}
            
		\item
            Case $\state \not\models \guard$.
            
            \begin{align*}
				& \ewp{C}{f}(\state) \\
				\cmmnt{Def.\ of $\wpsymbol$} \\
				\eeq& (\iverson{\neg\guard} \ivop f \mmadd \iverson{\guard} \ivop \wp{C_1}{\wp{C}{f}} )(\state) \\
				\cmmnt{Case $\state \not\models \guard$} \\
				\eeq& f(\state) \\
				\cmmnt{Def.\ of $\ext{\wpsymbol}$} \\
				\eeq& \ewp{\termState}{f}(\state) \\
				\cmmnt{\cref{fig:sosrules} (break)} \\
				\eeq& \mbigadd_{\sosTrans{\sosState{C}{\state}{n}{\beta}}{a}{\sosState{C'}{\state'}{n+1}{\beta'}}}
					a \smop \ewp{C'}{f}(\state')
            \end{align*}
        \end{enumerate}
    \end{itemize}
\end{proof}

\begin{lemma}
    \label{thm:operational-composition}
    Let $C_1, C_2$ be $\wgcl$ programs and $f \in \Weightings$, then
    \[
		\op{\COMPOSE{C_1}{C_2}}{f} \eeq \op{C_1}{\op{C_2}{f}} ~.
		\tag*{$\triangle$}
    \]
\end{lemma}
\begin{proof}
    Consider any state $\state \in \States$.
    We employ structural induction on the rules from \cref{fig:sosrules}.
    
    The following forms the induction base.
    The case $\sosTrans{ \sosState{C_1}{\state}{n}{\beta} }{a}{ \sosState{\termState}{\state'}{n+1}{\beta'} }$.
    \begin{align*}
		& \eop{\COMPOSE{C_1}{C_2}}{f}(\state) \\
		\cmmnt{$\wgcl$-functional} \\
		\eeq& \mbigadd_{\sosTrans{\sosState{\COMPOSE{C_1}{C_2}}{\state}{n}{\beta}}{a}{\sosState{C'}{\state'}{n+1}{\beta'}}}
			a \smop \eop{C'}{f}(\state') \\
		\cmmnt{\cref{fig:sosrules} (seq.\ 1)} \\
		\eeq& \mbigadd_{\sosTrans{\sosState{C_1}{\state}{n}{\beta}}{a}{\sosState{\termState}{\state'}{n+1}{\beta'}}}
			a \smop \eop{C_2}{f}(\state') \\
		\cmmnt{$\wgcl$-functional} \\
		\eeq& \eop{C_1}{\eop{C_2}{f}}(\state)
    \end{align*}
    
    The following forms the induction step.
    The case $\sosTrans{ \sosState{C_1}{\state}{n}{\beta} }{a}{ \sosState{C_1'}{\state'}{n+1}{\beta'} }$.
    \begin{align*}
		& \eop{\COMPOSE{C_1}{C_2}}{f}(\state) \\
		\cmmnt{$\wgcl$-functional} \\
		\eeq& \mbigadd_{\sosTrans{\sosState{\COMPOSE{C_1}{C_2}}{\state}{n}{\beta}}{a}{\sosState{C'}{\state'}{n+1}{\beta'}}}
			a \smop \eop{C'}{f}(\state') \\
		\cmmnt{\cref{fig:sosrules} (seq.\ 1)} \\
		\eeq& \mbigadd_{\sosTrans{\sosState{C_1}{\state}{n}{\beta}}{a}{\sosState{C_1'}{\state'}{n+1}{\beta'}}}
			a \smop \eop{\COMPOSE{C_1'}{C_2}}{f}(\state') \\
		\cmmnt{Induction Hypothesis} \\
		\eeq& \mbigadd_{\sosTrans{\sosState{C_1}{\state}{n}{\beta}}{a}{\sosState{C_1'}{\state'}{n+1}{\beta'}}}
			a \smop \eop{C_1'}{\eop{C_2}{f}}(\state') \\
		\cmmnt{$\wgcl$-functional} \\
		\eeq& \eop{C_1}{\eop{C_2}{f}}(\state)
    \end{align*}
\end{proof}

\begin{lemma}
    \label{thm:operational-looping}
    Let $C'$ be a $\wgcl$ program and $f \in \Weightings$, then
    \[
		\op{\WHILEDO{\guard}{C'}}{f} \eeq \iverson{\guard} \ivop f \mmadd \iverson{\neg\guard} \ivop \op{C'}{\op{\WHILEDO{\guard}{C'}}{f}} ~.
		\tag*{$\triangle$}
    \]
\end{lemma}
\begin{proof}
    Let $\state \in \States$ be a state.
    We distinguish two cases:
    \begin{itemize}
	\item
        $\state \models \guard$.
        
        \begin{align*}
			& \op{\WHILEDO{\guard}{C'}}{f}(\state) \\
			\cmmnt{$\wgcl$-functional} \\
			\eeq& \mbigadd_{\sosTrans{\sosState{C}{\state}{0}{\epsilon}}{a}{\sosState{C'}{\state'}{1}{\beta'}}}
				a \smop \eop{C'}{f}(\state') \\
			\cmmnt{\cref{fig:sosrules} (while)} \\
			\eeq& \op{\COMPOSE{C'}{\WHILEDO{\guard}{C'}}}{f}(\state) \\
			\cmmnt{\cref{thm:operational-composition}} \\
			\eeq& \op{C'}{\op{\WHILEDO{\guard}{C'}}{f}}(\state) \\
			\cmmnt{Case $\state \models \guard$} \\
			\eeq& (\iverson{\neg\guard} \ivop f \mmadd \iverson{\guard} \ivop \op{C'}{\op{\WHILEDO{\guard}{C'}}{f}})(\state) 
			~.
        \end{align*}
        
	\item
		$\state \not\models \guard$.
		
		\begin{align*}
			& \op{\WHILEDO{\guard}{C'}}{f}(\state) \\
			\cmmnt{$\wgcl$-functional} \\
			\eeq& \mbigadd_{\sosTrans{\sosState{C}{\state}{0}{\epsilon}}{a}{\sosState{C'}{\state'}{1}{\beta'}}}
				a \smop \eop{C'}{f}(\state') \\
			\cmmnt{\cref{fig:sosrules} (break)} \\
			\eeq& \eop{\termState}{f}(\state) \\
			\cmmnt{\cref{def:weighting-transformer}} \\
			\eeq& f(\state) \\
			\cmmnt{Case $\state \not\models \guard$} \\
			\eeq& (\iverson{\neg\guard} \ivop f \mmadd \iverson{\guard} \ivop \op{C'}{\op{\WHILEDO{\guard}{C'}}{f}})(\state)
			~.
		\end{align*}
    \end{itemize}
\end{proof}

\begin{lemma}
    \label{thm:wp-less-than-op}
    We have $\ext{\wpsymbol} \ewtord \ext{\opsymbol}$.
    \hfill$\triangle$
\end{lemma}
\begin{proof}
    We perform induction on the structure of $\wgcl$ programs.
    
    The following forms the induction base.
    \begin{itemize}
	\item
        The program $C$ is of the form $\ASSIGN{x}{E}$.
        
        We have
        \[
			\termPathsStartingIn{\sosState{C}{\state}{0}{\epsilon}}
			\eeq \set{ \sosTrans{ \sosState{C}{\state}{0}{\epsilon} }{\monone}{ \sosState{\termState}{\state\statesubst{x}{\eval{E}{\state}}}{1}{\epsilon} } }
			~.
        \]
        Hence,
        \begin{align*}
			& \eop{C}{f}(\state) \\
			\cmmnt{$\wgcl$-functional} \\
			\eeq& \eop{\termState}{f}(\state\statesubst{x}{\eval{E}{\state}}) \\
			\cmmnt{Def.\ of $\ext{\opsymbol}$} \\
			\eeq& f(\state\statesubst{x}{\eval{E}{\state}}) \\
			\cmmnt{Def.\ of $f\subst{x}{E}$} \\
			\eeq& f\subst{x}{E}(\state) \\
			\cmmnt{Def.\ of $\wpsymbol$} \\
			\eeq& \ewp{C}{f}(\state)
			~.
        \end{align*}
        
	\item
        The program $C$ is of the form $\WEIGH{a}$.
        
        We have
        \[
			\termPathsStartingIn{\sosState{C}{\state}{0}{\epsilon}}
			\eeq \set{ \sosTrans{ \sosState{C}{\state}{0}{\epsilon} }{a}{ \sosState{\termState}{\state}{1}{\epsilon} } }
			~.
		\]
        Hence,
        \begin{align*}
			& \eop{C}{f}(\state) \\
			\cmmnt{$\wgcl$-functional} \\
			\eeq& a \smop \eop{\termState}{f}(\state) \\
			\cmmnt{Def.\ of $\ext{\opsymbol}$} \\
			\eeq& a \smop f(\state) \\
			%
			%
			\cmmnt{Def.\ of $\wpsymbol$} \\
			\eeq& \ewp{C}{f}(\state)
			~.
        \end{align*}
        
    \end{itemize}
    
    The following forms the induction step.
    \begin{itemize}
	\item
        The program $C$ is of the form $\COMPOSE{C_1}{C_2}$.
        
        \begin{align*}
			& \ewp{C}{f}(\state) \\
			\cmmnt{Def.\ of $\wpsymbol$} \\
			\eeq& \ewp{C_1}{\ewp{C_2}{f}}(\state) \\
			\cmmnt{Monotonicity; Induction Hypothesis on $C_2$} \\
			\nnatord& \ewp{C_1}{\eop{C_2}{f}}(\state) \\
			\cmmnt{Induction Hypothesis on $C_1$} \\
			\nnatord& \eop{C_1}{\eop{C_2}{f}}(\state) \\
			\cmmnt{\cref{thm:operational-composition}} \\
			\eeq& \eop{C}{f}(\state)
			~.
        \end{align*}
        
	\item
        The program $C$ is of the form $\BRANCH{C_1}{C_2}$ or $\ITE{\guard}{C_1}{C_2}$.
        
        \begin{align*}
			& \ewp{C}{f}(\state) \\
			\cmmnt{$\wgcl$-functional} \\
			\eeq& \mbigadd_{\sosTrans{\sosState{C}{\state}{0}{\epsilon}}{a}{\sosState{C'}{\state'}{1}{\beta'}}}
				a \smop \ewp{C'}{f}(\state') \\
			\cmmnt{Induction Hypothesis on $C' \in \set{ C_1, C_2 }$} \\
			\nnatord& \mbigadd_{\sosTrans{\sosState{C}{\state}{0}{\epsilon}}{a}{\sosState{C'}{\state'}{1}{\beta'}}}
				a \smop \eop{C'}{f}(\state') \\
			\cmmnt{$\wgcl$-functional} \\
			\eeq& \eop{C}{f}(\state)
        \end{align*}
        
	\item
        The program $C$ is of the form $\WHILEDO{\guard}{C'}$.
        
        Let $\charwp{\guard}{C'}{f}$ be the corresponding characteristic function
        \[
			\charwp{\guard}{C'}{f}(X) \qcoloneqq \iverson{\neg\guard} \ivop f \mmadd \iverson{\guard} \ivop \wp{C'}{X} ~,
        \]
        i.\,e.\ $\wp{C}{f} = \lfpop \charwp{\guard}{C'}{f}$.
        The map $\eop{C}{f}$ is a prefixed point of $\charwp{\guard}{C'}{f}$:
        \begin{align*}
			& \charwp{\guard}{C'}{f}(\eop{C}{f}) \\
			\cmmnt{Def.\ of $\charwp{\guard}{C'}{f}$} \\
			\eeq& \iverson{\neg\guard} \ivop f \mmadd \iverson{\guard} \ivop \wp{C'}{\eop{C}{f}} \\
			\cmmnt{Induction Hypothesis on $C'$} \\
			\nnatord& \iverson{\neg\guard} \ivop f \mmadd \iverson{\guard} \ivop \eop{C'}{\eop{C}{f}} \\
			\cmmnt{\cref{thm:operational-looping}} \\
			\eeq& \eop{C}{f}
			~.
        \end{align*}
        With \cref{app:thm:kleene-park} it follows $\ewp{C}{f} \natord \eop{C}{f}$.
    \end{itemize}
\end{proof}

\subsection{Proof of Theorem \ref{thm:wp-properties}}
\label{app:proof:wp-properties}

\begin{theorem}[\cref{thm:wp-properties}]
	Let the monoid module $\Smodule$ over $\Monoid$ be $\omega$-continuous.
    For all $\Monoid$-$\wgcl$ programs $C$, the $\wpsymbol$ transformer is 
    \begin{itemize}
		\item \emph{monotone}, i.e.\ for all $f, g \in \Weightings$ with $f \natord g$,  \quad $\wp{C}{f} \natord \wp{C}{g}$ ~;
		\item \emph{strict}, i.e.\ \quad $\wp{C}{\snull} \eeq \snull$ ~;
        \item \emph{additive}, i.e.\ for all $f, g \in \Weightings$, \quad $\wp{C}{f \mmadd g} \eeq \wp{C}{f} \mmadd \wp{C}{g}$ ~;
        \item and moreover, if the monoid $\Monoid$ is commutative, then $\wpsymbol$ is \emph{linear}, i.e.\ for all $a \in \monoid$,
        \[
            \quad \wp{C}{a \smop f} \eeq a \smop \wp{C}{f} ~.
            \tag*{$\triangle$}
        \]
    \end{itemize}
\end{theorem}
\begin{proof}
    \begin{itemize}
		\item \emph{Monotonicity}.
			Follows directly from $\omega$-continuity of $\wpC{C}$, \cref{thm:wp-continuous}.

		\item \emph{Strictness}.
			Follows from \cref{thm:wp-operational-semantics} and annihilation of $\mnull$:
			Let $\state \in \States$.
			\begin{align*}
				& \wp{C}{\mnull}(\state) \\
				\qeq& \mbigadd_{\compPath \in \termPathsStartingIn{\sosState{C}{\state}{n}{v}}} \pathWeight(\compPath) \smop \mnull(\lastState(\compPath)) \\
				\qeq& \mbigadd_{\compPath \in \termPathsStartingIn{\sosState{C}{\state}{n}{v}}} \mnull \\
				\qeq& \mnull  
			\end{align*}

		\item \emph{Additivity}.
			Follows from \cref{thm:wp-operational-semantics}:
			Let $\state \in \States$ and $f, g \in \Weightings$.
			\begin{align*}
				&\wp{C}{f \madd g}(\state) \\
				\qeq& \mbigadd_{\compPath \in \termPathsStartingIn{\sosState{C}{\state}{n}{v}}} \pathWeight(\compPath) \smop (f \madd g)(\lastState(\compPath)) \\
				\qeq& \mbigadd_{\compPath \in \termPathsStartingIn{\sosState{C}{\state}{n}{v}}} \pathWeight(\compPath) \smop (f(\lastState(\compPath) \madd g(\lastState(\compPath)) \\
				\qeq& \mbigadd_{\compPath \in \termPathsStartingIn{\sosState{C}{\state}{n}{v}}} \pathWeight(\compPath) \smop f(\lastState(\compPath) \madd \pathWeight(\compPath) \smop g(\lastState(\compPath)) \\
				\qeq& \mbigadd_{\compPath \in \termPathsStartingIn{\sosState{C}{\state}{n}{v}}} \pathWeight(\compPath) \smop f(\lastState(\compPath) \\
				\mmadd& \mbigadd_{\compPath \in \termPathsStartingIn{\sosState{C}{\state}{n}{v}}} \pathWeight(\compPath) \smop g(\lastState(\compPath)) \\
				\qeq& \wp{C}{f}(\state) \madd \wp{C}{g}(\state)
			\end{align*}

		\item \emph{Linearity}.
			Let the monoid $\Monoid$ be commutative; we apply \cref{thm:wp-operational-semantics}:
			As additivity always holds, we only have to show homogenity.
			Let $\state \in \States$, $a \in \monoid$ and $f \in \Weightings$.
			\begin{align*}
				& \wp{C}{a \smop f}(\state) \\
				\qeq& \mbigadd_{\compPath \in \termPathsStartingIn{\sosState{C}{\state}{n}{v}}} \pathWeight(\compPath) \smop (a \smop f)(\lastState(\compPath)) \\
				\qeq& \mbigadd_{\compPath \in \termPathsStartingIn{\sosState{C}{\state}{n}{v}}} \pathWeight(\compPath) \smop (a \smop f(\lastState(\compPath)) \\
				\qeq& \mbigadd_{\compPath \in \termPathsStartingIn{\sosState{C}{\state}{n}{v}}} (\pathWeight(\compPath) \monop a) \smop f(\lastState(\compPath) \\
				\qeq& \mbigadd_{\compPath \in \termPathsStartingIn{\sosState{C}{\state}{n}{v}}} (a \monop \pathWeight(\compPath)) \smop f(\lastState(\compPath) \\
				\qeq& \mbigadd_{\compPath \in \termPathsStartingIn{\sosState{C}{\state}{n}{v}}} a \smop (a \monop \pathWeight(\compPath) \smop f(\lastState(\compPath)) \\
				\qeq& a \smop \mbigadd_{\compPath \in \termPathsStartingIn{\sosState{C}{\state}{n}{v}}} \pathWeight(\compPath) \smop f(\lastState(\compPath) \\
				\qeq& a \smop \wp{C}{f}(\state)
			\end{align*}
    \end{itemize}
\end{proof}

\subsection{Proof of Theorem \ref{thm:wlp-continuous}}
\begin{theorem}[\cref{thm:wlp-continuous}]
	Let $\Smodule$ be an $\omega$-\emph{co}continuous $\Monoid$-module.
    For all $\Monoid$-$\wgcl$ programs $C$, the weighting transformer $\wlpC{C}$ is a well-defined $\omega$-\emph{co}continuous endofunction on the module of weightings over $\Smodule$.
    In particular, if $C=\WHILEDO{\guard}{C'}$, we have for all $f \in \Weightings$ that
    \[
        \wlp{C}{f} \qeq \bigmeet_{i \in \Nats} \charwlp{}{}{f}^i(\top) ~.
        \tag*{$\triangle$}
    \]
\end{theorem}
\begin{proof}
	Fully analogous to \cref{thm:wp-continuous}.
\end{proof}

\subsection{Proof of Theorem \ref{thm:wlp-decomp}}
\label{app:proof:wlp-decomp}

Let $\Smodule$ be an $\omega$-\emph{bi}continuous $\Monoid$-module.
First we need an auxiliary lemma.

\begin{lemma}
    \label{thm:chains-up-and-down}
    Let $(a_i)_{i \in \Nats}$ be an ascending $\omega$-chain and $(b_i)_{i \in \Nats}$ be a descending $\omega$-chain in an $\omega$-bicontinuous module.
    Furthermore, suppose that $(c_i)_{i \in \Nats}$ is a descending $\omega$-chain such that for all $i \in \Nats$ it holds that $c_i = a_i \madd b_i$.
    Then we have
    \[
        \bigmeet_{i \in \Nats} c_i \qeq \bigjoin_{i \in \Nats} a_i \mmadd \bigmeet_{i \in \Nats} b_i ~.
    \]
\end{lemma}
\begin{proof}
	Let $a_{\omega} \coloneq \bigjoin_{i \in \Nats} a_i$.
	We have
	\[
		\bigmeet_{i \in \Nats} c_i
		\qeq \bigmeet_{i \in \Nats} \left( a_i \madd b_i \right)
		\qnatord \bigmeet_{i \in \Nats} \left( a_{\omega} \madd b_i \right)
		\qeq a_{\omega} \madd \bigmeet_{i \in \Nats} b_i ~,
	\]
    where the last equality holds by $\omega$-cocontinuity of $\madd$.

    To prove the inequality in the other direction we will make use of the following auxiliary claim:
    
    \begin{claim}
        Suppose that $(x_i)_{i \in \Nats}$ is an ascending $\omega$-chain and $(y_i)_{i \in \Nats}$ is a descending $\omega$-chain such that for all $i \in \Nats$ we have $x_i \natord y_i$.
		Then $\bigjoin_{i \in \Nats} x_i \natord \bigmeet_{i \in \Nats} y_i$.
    \end{claim}
    \begin{proof}[Proof of Claim]
        We first show the following:
        \[
            \forall i \in \Nats \quad \forall j \in \Nats \colonq x_i \natord y_j ~.
        \]
        To this end let $i, j$ be arbitrary. There are two cases to consider:
        \begin{enumerate}
            \item $i \leq j$. By assumption, $x_j \natord y_j$. Further, since the $x$'s form an ascending chain, we have $x_i \natord x_j$. 
            \item $i > j$. By assumption, $x_i \natord y_i$. Further, since the $y$'s form a descending chain, we have $y_i \natord y_j$.
        \end{enumerate}
		In both cases, $x_i \natord y_j$ holds.
		Now fix $n \in \Nats$.
		We just have shown that $x_i \natord y_n$ for all $i \in \Nats$.
		Thus by definition of $\bigjoin$ we immediately have $\bigjoin_{i \in \Nats} x_i \natord y_n$.
		Since $n$ was arbitrary, we finally obtain $\bigjoin_{i \in \Nats} x_i \natord \bigmeet_{i \in \Nats} y_i$ by definition of $\bigmeet$.
    \end{proof}

    To conclude the proof of \cref{thm:chains-up-and-down}, let $b_{\omega} \coloneq \bigmeet_{i \in \Nats} b_i$ and note that for all $i \in \Nats$,
    \[
        a_i \madd b_\omega \qnatord a_i \madd b_i \qnatord c_i
    \]
    and that $(a_i \madd b_\omega)_{i \in \Nats}$ is an ascending $\omega$-chain.
    Invoking the auxiliary claim above, we obtain
    \[
        \bigjoin_{i \in \Nats} \left( a_i \madd b_\omega \right) \qnatord \bigmeet_{i \in \Nats} c_i ~.
    \]
    Noticing that $\bigjoin_{i \in \Nats} \left( a_i \madd b_\omega \right) = \left( \bigjoin_{i \in \Nats}  a_i \right) \madd b_\omega$ by $\omega$-continuity of $\madd$ concludes the proof.
\end{proof}

\begin{theorem}[\cref{thm:wlp-decomp}]
	Let $\Smodule$ be an $\omega$-\emph{bi}continuous $\Monoid$-module.
    Let $C$ be a $\Monoid$-$\wgcl$ program and $\Smodule$ an $\omega$-bicontinuous $\Monoid$-module.
    Then for all $f \in \Weightings$, 
    \[
        \wlp{C}{f} \qeq \wp{C}{f} \mmadd \wlp{C}{\mnull} ~.
	\]
\end{theorem}
\begin{proof}
	We employ induction on the structure of $C$.
	Let $f \in \Weightings$ be an arbitrary postweighting.

	\begin{itemize}
	\item The cases where $C$ is of the form $\ASSIGN{x}{E}$ or $\WEIGH{a}$.
		Then $\wlpC{C} = \wpC{C}$ and thus $\wlp{C}{\mnull} = \wp{C}{\mnull} = \mnull$.
		
	\item $C$ is of the form $\COMPOSE{C_1}{C_2}$.

		In this case, we have
		\begin{align*}
			& \wlp{\COMPOSE{C_1}{C_2}}{f} \\
			\cmmnt{Def.\ of $\wlpsymbol$} \\
			\qeq& \wlp{C_1}{\wlp{C_2}{f}} \\
			\cmmnt{I.H.\ for $C_2$} \\
			\qeq& \wlp{C_1}{\wp{C_2}{f} \madd \wlp{C_2}{\mnull}} \\
			\cmmnt{I.H.\ for $C_1$} \\
			\qeq& \wp{C_1}{\wp{C_2}{f} \madd \wlp{C_2}{\mnull}} \mmadd \wlp{C_1}{\mnull} \\
			\cmmnt{Additivity of $\wpsymbol$} \\
			\qeq& \wp{C_1}{\wp{C_2}{f}} \mmadd \wp{C_1}{\wlp{C_2}{\mnull}} \mmadd \wlp{C_1}{\mnull} \\
			\cmmnt{I.H.\ for $C_1$} \\
			\qeq& \wp{C_1}{\wp{C_2}{f}} \mmadd \wlp{C_1}{\wlp{C_2}{\mnull}} \\
			\cmmnt{Def.\ of $\wpsymbol$ and $\wlpsymbol$} \\
			\qeq& \wp{\COMPOSE{C_1}{C_2}}{f} \mmadd \wlp{\COMPOSE{C_1}{C_2}}{\mnull}
			~.
		\end{align*}
		
		\item $C$ is of the form $\ITE{\guard}{C_1}{C_2}$.
		\begin{align*}
			& \wlp{\ITE{\guard}{C_1}{C_2}}{f} \\
			\cmmnt{Def.\ of $\wlpsymbol$} \\
			\qeq& \iverson{\guard} \ivop \wlp{C_1}{f} \mmadd \iverson{\neg\guard} \ivop \wlp{C_2}{f} \\
			\cmmnt{I.H.\ on $C_1$ and $C_2$}\\
			\qeq& \iverson{\guard} \ivop \left( \wp{C_1}{f} \sadd \wlp{C_1}{\mnull}\right) \mmadd \iverson{\neg\guard} \ivop \left(\wp{C_2}{f} \sadd \wlp{C_2}{\mnull} \right) \\
			\cmmnt{Distributivity} \\
			\qeq& \iverson{\guard} \ivop \wp{C_1}{f} \mmadd \iverson{\guard} \ivop \wlp{C_1}{\mnull} \\
			\qmadd& \iverson{\neg\guard} \ivop \wp{C_2}{f} \mmadd \iverson{\neg\guard} \ivop \wlp{C_2}{\mnull} \\
			\cmmnt{Commutativity of $\madd$} \\
			\qeq& \iverson{\guard} \ivop \wp{C_1}{f} \mmadd \iverson{\neg\guard} \ivop \wp{C_2}{f} \\
			\qmadd& \iverson{\guard} \ivop \wlp{C_1}{\mnull} \mmadd \iverson{\neg\guard} \ivop \wlp{C_2}{\mnull} \\
			\cmmnt{Def.\ of $\wpsymbol$ and $\wlpsymbol$} \\
			\qeq& \wp{\ITE{\guard}{C_1}{C_2}}{f} \mmadd \wlp{\ITE{\guard}{C_1}{C_2}}{\mnull}
			~.
		\end{align*}
		
		\item $C$ is of the form $\BRANCH{C_1}{C_2}$.
		\begin{align*}
			& \wlp{\BRANCH{C_1}{C_2}}{f} \\
			\cmmnt{Def.\ of $\wlpsymbol$} \\
			\qeq& \wlp{C_1}{f} \mmadd \wlp{C_2}{f} \\
			\cmmnt{I.H.\ on $C_1$ and $C_2$} \\
			\qeq& \wp{C_1}{f} \madd \wlp{C_1}{\mnull} \mmadd \wp{C_2}{f} \madd \wlp{C_2}{\mnull} \\
			\cmmnt{Commutativity of $\madd$} \\
			\qeq& \wp{C_1}{f} \madd \wp{C_2}{f} \mmadd \wlp{C_1}{\mnull} \madd \wlp{C_2}{\mnull} \\
			\cmmnt{Def.\ of $\wpsymbol$ and $\wlpsymbol$} \\
			\qeq& \wp{\BRANCH{C_1}{C_2}}{f} \mmadd \wlp{\BRANCH{C_1}{C_2}}{\mnull} \\
		\end{align*}
		
	\item $C$ is of the form $\WHILEDO{\guard}{C'}$.

		Let $\charwlp{\guard}{C'}{f}$ and $\charwp{\guard}{C'}{f}$ be the $\wlpsymbol$- and $\wpsymbol$-characteristic functions of the loop, respectively.
		We claim that for all $n \in \Nats$ it holds that
		\begin{align}
            \label{eq:inner-induction-wlp-decomp}
			\charwlpn{\guard}{C'}{f}{n}(\mTop) \qeq \charwpn{\guard}{C'}{f}{n}(\mnull) \mmadd \charwlpn{\guard}{C'}{\mnull}{n}(\mTop) ~.
		\end{align}
		This claim is proved by induction on $n$ (the I.H.\ of this induction is referred to as ``inner I.H.''):
		\begin{itemize}
		\item $n = 0$.
			In this case, the claim holds trivially.

		\item $n > 0$.
			\begin{align*}
				&\charwlpn{\guard}{C'}{f}{n}(\mTop)\\
				\cmmnt{n > 0}\\
				\qeq& \charwlp{\guard}{C'}{f}(\charwlpn{\guard}{C'}{f}{n-1}(\mTop))\\
				\cmmnt{inner I.H.}\\
				\qeq& \charwlp{\guard}{C'}{f}\left( \charwpn{\guard}{C'}{f}{n-1}(\mnull) \mmadd \charwlpn{\guard}{C'}{\mnull}{n-1}(\mTop) \right)\\
				\cmmnt{Def. of $\charwlp{\guard}{C'}{f}$}\\
				\qeq& \iverson{\neg\guard} \ivop f \mmadd \iverson{\guard} \ivop \wlp{C'}{\charwpn{\guard}{C'}{f}{n-1}(\mnull) \mmadd \charwlpn{\guard}{C'}{\mnull}{n-1}(\mTop) }\\
				\cmmnt{I.H.}\\
				\qeq& \iverson{\neg\guard} \ivop f \mmadd \iverson{\guard} \ivop \left(  \wp{C'}{\charwpn{\guard}{C'}{f}{n-1}(\mnull) \mmadd \charwlpn{\guard}{C'}{\mnull}{n-1}(\mTop)} \sadd \wlp{C'}{\mnull} \right) \\
				\cmmnt{Linearity of $\wpsymbol$}\\
				\qeq& \iverson{\neg\guard} \ivop f \mmadd \iverson{\guard} \ivop \left( \wp{C'}{\charwpn{\guard}{C'}{f}{n-1}(\mnull)} \mmadd  \wp{C'}{\charwlpn{\guard}{C'}{\mnull}{n-1}(\mTop)}  \mmadd \wlp{C'}{\mnull} \right)
				\cmmnt{Distributivity}\\
				\qeq& \iverson{\neg\guard} \ivop f \mmadd \iverson{\guard} \ivop \wp{C'}{\charwpn{\guard}{C'}{f}{n-1}(\mnull)} \\
				&\mmadd \iverson{\guard} \ivop \left( \wp{C'}{\charwlpn{\guard}{C'}{\mnull}{n-1}(\mTop)} \madd \wlp{C'}{\mnull} \right) \\
				\cmmnt{Def. of $\charwp{\guard}{C'}{f}$}\\
				\qeq& \charwpn{\guard}{C'}{f}{n}(\mnull) \mmadd \iverson{\guard} \ivop \left( \wp{C'}{\charwlpn{\guard}{C'}{\mnull}{n-1}(\mTop)}  \sadd \wlp{C'}{\mnull} \right) \\
				\cmmnt{I.H.}\\
				\qeq& \charwpn{\guard}{C'}{f}{n}(\mnull) \mmadd \iverson{\guard} \ivop \left( \wlp{C'}{\charwlpn{\guard}{C'}{\mnull}{n-1}(\mTop)} \right) \\
				\cmmnt{Def. of $\charwlp{\guard}{C'}{\mnull}$}\\
				\qeq& \charwpn{\guard}{C'}{f}{n}(\mnull) \mmadd \charwlpn{\guard}{C'}{\mnull}{n}(\mTop) ~.
			\end{align*}
			We are now in a position to conclude the proof:
			\begin{align*}
				& \wlp{\WHILEDO{\guard}{C'}}{f}\\
				\cmmnt{\cref{thm:wlp-continuous}}\\
				\qeq& \bigmeet_{n \in \Nats} \charwlpn{\guard}{C'}{f}{n}(\mTop) \\
				\cmmnt{By \eqref{eq:inner-induction-wlp-decomp}}\\
				\qeq& \bigmeet_{n \in \Nats} \left( \charwpn{\guard}{C'}{f}{n}(\mnull) \mmadd \charwlpn{\guard}{C'}{\mnull}{n}(\mTop) \right) \\
				\cmmnt{\cref{thm:chains-up-and-down}}\\
				\qeq& \bigjoin_{n \in \Nats}  \charwpn{\guard}{C'}{f}{n}(\mnull) \mmadd \bigmeet_{n \in \Nats} \charwlpn{\guard}{C'}{\mnull}{n}(\mTop)  \\
				\cmmnt{\cref{thm:wp-continuous} \& \cref{thm:wlp-continuous}}\\
				\qeq& \wp{\WHILEDO{\guard}{C'}}{f} \mmadd \wlp{\WHILEDO{\guard}{C'}}{\mnull}
			\end{align*} 
		\end{itemize}
	\end{itemize}
\end{proof}

\subsection{Proof of Theorem \ref{thm:wlp-operational-semantics} (Soundness of $\wlpsymbol$ w.r.t.\ operational semantics)}
\label{app:proof:wlp-soundness}

\begin{theorem}[\cref{thm:wlp-operational-semantics}]
	Let the monoid module $\Smodule$ over $\Monoid$ be $\omega$-\emph{bi}continuous.
    For any $\wgcl$ program $C$ and initial state $\state \in \States$,
    \begin{align*}
		\wlp{C}{\mnull}(\state) 
		\qeq \bigmeet_{n \in \Nats}
			\mbigadd_{\compPath \in \pathsOfLengthStartingIn{n}{\sosStateAbbr{C}{\state}} }
				\pathWeight(\compPath) \smop \mTop
		~.
    \end{align*}
\end{theorem}
\begin{proof}
    We use a few auxiliary definitions and lemmas that can be found below.
	First, the right hand side is well-defined by \cref{thm:trailing-paths-form-omega-chain}.
	Now, the claim is exactly \cref{thm:wlp-zero-is-olp} by \cref{def:liberal-operational-functional}.
\end{proof}

\begin{lemma}
	\label{thm:trailing-paths-form-omega-chain}
	For any $\wgcl$ program $C$ and state $\state \in \States$ the sequence $(s_n)_{n \in \Nats}$, where
	\[
		s_n \qcoloneqq \mbigadd_{\compPath \in \pathsOfLengthStartingIn{n}{\sosStateAbbr{C}{\state}}}
			\pathWeight(\compPath) \smop \mTop ~,
	\]
	is a descending $\omega$-chain.
\end{lemma}
\begin{proof}
	Let $n \in \Nats$.
	Given $\compPath \in \pathsOfLengthStartingIn{n}{\sosStateAbbr{C}{\state}}$,
	\[
		\mbigadd_{\conf \in \succConfs(\lastState(\compPath))}
			\weight{\lastState(\compPath)}{\conf} \smop \mTop
		\qnatord \mTop ~.
	\]
	As the module's scalar multiplication $\smop$ is $\omega$-continuous and thus monotone in the second argument, it follows
	\begin{align*}
		s_{n+1}
		&\qeq
		\mbigadd_{\compPath \in \pathsOfLengthStartingIn{n+1}{\sosStateAbbr{C}{\state}}}
			\pathWeight(\compPath) \smop \mTop \\
		&\qeq \mbigadd_{\compPath \in \pathsOfLengthStartingIn{n}{\sosStateAbbr{C}{\state}}}
			\pathWeight(\compPath) \smop \mbigadd_{\conf \in \succConfs(\lastState(\compPath))}
				\weight{\lastState(\compPath)}{\conf} \smop \mTop \\
		&\qnatord \mbigadd_{\compPath \in \pathsOfLengthStartingIn{n}{\sosStateAbbr{C}{\state}}}
			\pathWeight(\compPath) \smop \mTop \\
		&\qeq s_n
		~.
	\end{align*}
\end{proof}

\begin{definition}
    \label{def:liberal-operational-functional}
    The map $\olpsymbol\colon \wgcl \to \Weightings$ is defined for any $\wgcl$ program $C$ and state $\state \in \States$ via
    \[
		\olp{C}(\state)
		\qcoloneqq \bigmeet_{n \in \Nats}
			\mbigadd_{\compPath \in \pathsOfLengthStartingIn{n}{\sosStateAbbr{C}{\state}}}
				\pathWeight(\compPath) \smop \mTop ~.
    \]
	It is well-defined by \cref{thm:trailing-paths-form-omega-chain}.
    \hfill$\triangle$
\end{definition}

\begin{lemma}
	\label{thm:liberal-wgcl-functional}
	The map $\olpsymbol$ behaves like a \emph{$\wgcl$-functional} (see \cref{def:wgcl-functional}):\footnote{But it is not an (extended) weighting transformer.}
	\[
		\olp{C} \eeq \mbigadd_{\sosTrans{\sosState{C}{\state}{n}{\beta}}{a}{\sosState{C'}{\state'}{n+1}{\beta'}}} a \smop \olp{C'} ~.
	\]
\end{lemma}
\begin{proof}
	\begin{align*}
		& \olp{C}(\state) \\
		\cmmnt{Definition} \\
		\qeq& \bigmeet_{n \in \Nats} \quad
			\mbigadd_{\compPath \in \pathsOfLengthStartingIn{n}{\sosStateAbbr{C}{\state}}}
				\pathWeight(\compPath) \smop \mTop \\
		\cmmnt{\cref{thm:trailing-paths-form-omega-chain}} \\
		\qeq& \bigmeet_{\substack{n \in \Nats\\ n \geq 2}} \quad
			\mbigadd_{\conf_0 \conf_1 \ldots \conf_n \in \pathsOfLengthStartingIn{n}{\sosStateAbbr{C}{\state}}}
				(\weight{\conf_0}{\conf_1} \monop \pathWeight(\conf_1 \ldots \conf_n)) \smop \mTop \\
		\cmmnt{split $\conf_0 \conf_1 \ldots \conf_n$} \\
		\qeq& \bigmeet_{\substack{n \in \Nats\\ n \geq 2}} \quad
			\mbigadd_{\sosTrans{ \sosStateAbbr{C}{\state} }{a}{ \conf_1 }} \quad
				\mbigadd_{\conf_1 \ldots \conf_n \in \pathsOfLengthStartingIn{n-1}{\conf_1}}
					a \smop \pathWeight(\compPath) \smop \mTop \\
		\cmmnt{Distributivity} \\
		\qeq& \bigmeet_{n \in \Nats} \quad
			\mbigadd_{\sosTrans{ \sosStateAbbr{C}{\state} }{a}{ \conf_1 } } \quad
				a \smop \mbigadd_{\conf_1 \ldots \conf_n \in \pathsOfLengthStartingIn{n}{\conf_1}}
					\pathWeight(\compPath) \smop \mTop \\
		\cmmnt{$\Smodule$ is $\omega$-cocontinuous} \\
		\qeq& \mbigadd_{\sosTrans{ \sosStateAbbr{C}{\state} }{a}{ \conf_1 } } \quad
			a \smop \bigmeet_{n \in \Nats} \quad
				\mbigadd_{\conf_1 \ldots \conf_n \in \pathsOfLengthStartingIn{n}{\conf_1}}
					\pathWeight(\compPath) \smop \mTop \\
		\cmmnt{Definition} \\
		\qeq& \mbigadd_{\sosTrans{\sosStateAbbr{C}{\state}}{a}{\sosState{C'}{\state'}{1}{\beta}}}
			a \smop \olp{C'}
	\end{align*}
\end{proof}

\begin{lemma}
	\label{thm:olp-composition}
	For any $\wgcl$ programs $C_1, C_2$,
	\[
		\olp{\COMPOSE{C_1}{C_2}} \qeq \op{C_1}{ \olp{C_2} } \mmadd \olp{C_1} ~.
		\tag*{$\triangle$}
	\]
\end{lemma}
\begin{proof}
	Consider any state $\state \in \States$.
    We employ structural induction on the rules from \cref{fig:sosrules}.
    
    The following forms the induction base.
    The case $\sosTrans{ \sosState{C_1}{\state}{n}{\beta} }{a}{ \sosState{\termState}{\state'}{n+1}{\beta'} }$.
	First, $\olp{C_1} = \mnull$ as $\pathsOfLengthStartingIn{n}{\sosStateAbbr{C_1}{\state}} = \emptyset$ by \cref{fig:sosrules} for any $n \geq 2$.
	Hence,
    \begin{align*}
		& \olp{\COMPOSE{C_1}{C_2}}(\state) \\
		\cmmnt{\cref{thm:liberal-wgcl-functional}} \\
		\eeq& \mbigadd_{\sosTrans{\sosState{\COMPOSE{C_1}{C_2}}{\state}{n}{\beta}}{a}{\sosState{C'}{\state'}{n+1}{\beta'}}}
			a \smop \olp{C'}(\state') \\
		\cmmnt{\cref{fig:sosrules} (seq.\ 1)} \\
		\eeq& \mbigadd_{\sosTrans{\sosState{C_1}{\state}{n}{\beta}}{a}{\sosState{\termState}{\state'}{n+1}{\beta'}}}
			a \smop \olp{C_2}(\state') \\
		\cmmnt{$\wgcl$-functional} \\
		\eeq& \op{C_1}{\olp{C_2}}(\state) \\
		\cmmnt{$\olp{C_1} = \mnull$} \\
		\eeq& \op{C_1}{\olp{C_2}}(\state) \mmadd \olp{C_1}(\state)
    \end{align*}
    
    The following forms the induction step.
    The case $\sosTrans{ \sosState{C_1}{\state}{n}{\beta} }{a}{ \sosState{C_1'}{\state'}{n+1}{\beta'} }$.
    \begin{align*}
		& \olp{\COMPOSE{C_1}{C_2}}(\state) \\
		\cmmnt{\cref{thm:liberal-wgcl-functional}} \\
		\eeq& \mbigadd_{\sosTrans{\sosState{\COMPOSE{C_1}{C_2}}{\state}{n}{\beta}}{a}{\sosState{C'}{\state'}{n+1}{\beta'}}}
			a \smop \olp{C'}(\state') \\
		\cmmnt{\cref{fig:sosrules} (seq.\ 1)} \\
		\eeq& \mbigadd_{\sosTrans{\sosState{C_1}{\state}{n}{\beta}}{a}{\sosState{C_1'}{\state'}{n+1}{\beta'}}}
			a \smop \olp{\COMPOSE{C_1'}{C_2}}(\state') \\
		\cmmnt{Induction Hypothesis} \\
		\eeq& \mbigadd_{\sosTrans{\sosState{C_1}{\state}{n}{\beta}}{a}{\sosState{C_1'}{\state'}{n+1}{\beta'}}}
			a \smop \big( \op{C_1'}{\olp{C_2}} \mmadd \olp{C_1'} \big)(\state') \\
		\cmmnt{Distributivity} \\
		\eeq& \mbigadd_{\sosTrans{\sosState{C_1}{\state}{n}{\beta}}{a}{\sosState{C_1'}{\state'}{n+1}{\beta'}}}
			a \smop \op{C_1'}{\olp{C_2}}(\state') \\
		\mmadd& \mbigadd_{\sosTrans{\sosState{C_1}{\state}{n}{\beta}}{a}{\sosState{C_1'}{\state'}{n+1}{\beta'}}}
			a \smop \olp{C_1'}(\state') \\
		\cmmnt{$\wgcl$-functional; \cref{thm:liberal-wgcl-functional}} \\
		\eeq& \op{C_1}{\olp{C_2}{f}}(\state) \mmadd \olp{C_1}(\state)
    \end{align*}
\end{proof}

\begin{lemma}
	\label{thm:wlp-zero-is-olp}
	For any $\wgcl$ program $C$ and state $\state \in \States$, it is $\wlp{C}{\mnull} = \olp{C}$.
    \hfill$\triangle$
\end{lemma}
\begin{proof}
	We employ induction on the structure of $C$.
    
    The following forms the induction base.
    \begin{itemize}
	\item
        The program $C$ is of the form $\ASSIGN{x}{E}$.
		Then, $\pathsOfLengthStartingIn{n}{\sosStateAbbr{C}{\state}} = \emptyset$ by \cref{fig:sosrules} for any $n \geq 2$, hence
		\[
			\wlp{C}{\mnull}
			\eeq \mnull
			\eeq \olp{C}
			~.
		\]

	\item
        The program $C$ is of the form $\WEIGH{a}$.
		Then, $\pathsOfLengthStartingIn{n}{\sosStateAbbr{C}{\state}} = \emptyset$ by \cref{fig:sosrules} for any $n \geq 2$, hence
		\[
			\wlp{C}{\mnull}
			\eeq a \smop \mnull
			\eeq \mnull
			\eeq \olp{C}
			~.
		\]

    \end{itemize}
    
    The following forms the induction step.
    \begin{itemize}
	\item
        The program $C$ is of the form $\COMPOSE{C_1}{C_2}$.
		\begin{align*}
			& \wlp{C}{\mnull} \\
			\cmmnt{Def.\ of $\wlpsymbol$} \\
			\eeq& \wlp{C_1}{ \wlp{C_2}{\mnull} } \\
			\cmmnt{\cref{thm:wlp-decomp}} \\
			\eeq& \wp{C_1}{ \wlp{C_2}{\mnull} } \mmadd \wlp{C_1}{\mnull} \\
			\cmmnt{\cref{thm:wp-operational-semantics}; Induction Hypothesis} \\
			\eeq& \op{C_1}{ \olp{C_2} } \mmadd \olp{C_1} \\
			\cmmnt{\cref{thm:olp-composition}} \\
			\eeq& \olp{\COMPOSE{C_1}{C_2}}
		\end{align*}

	\item
        The program $C$ is of the form $\ITE{\guard}{C_1}{C_2}$.
		Let $\state \in \States$ be a state.
		There are the following two exclusive cases:
        \begin{enumerate}
		\item
            Case $\state \models \guard$.
            \begin{align*}
				& \wlp{C}{\mnull}(\state) \\
				\cmmnt{Def.\ of $\wlpsymbol$} \\
				\eeq& (\iverson{\guard} \ivop \wlp{C_1}{\mnull} \mmadd \iverson{\neg\guard} \ivop \wlp{C_2}{\mnull})(\state) \\
				\cmmnt{Case $\state \models \guard$} \\
				\eeq& \wlp{C_1}{\mnull}(\state) \\
				\cmmnt{\cref{fig:sosrules} (if)} \\
				\eeq& \mbigadd_{\sosTrans{\sosState{C}{\state}{n}{\beta}}{a}{\sosState{C'}{\state'}{n+1}{\beta'}}}
					a \smop \wlp{C'}{\mnull}(\state') \\
				\cmmnt{Induction Hypothesis on $C_1$} \\
				\eeq& \mbigadd_{\sosTrans{\sosState{C}{\state}{n}{\beta}}{a}{\sosState{C'}{\state'}{n+1}{\beta'}}}
					a \smop \olp{C'}(\state') \\
				\cmmnt{\cref{thm:liberal-wgcl-functional}} \\
				\eeq& \olp{C}(\state)
            \end{align*}
            
		\item
            Case $\state \not\models \guard$.
			\begin{align*}
				& \wlp{C}{\mnull}(\state) \\
				\cmmnt{Def.\ of $\wlpsymbol$} \\
				\eeq& (\iverson{\guard} \ivop \wlp{C_1}{\mnull} \mmadd \iverson{\neg\guard} \ivop \wlp{C_2}{\mnull})(\state) \\
				\cmmnt{Case $\state \not\models \neg\guard$} \\
				\eeq& \wlp{C_2}{\mnull}(\state) \\
				\cmmnt{\cref{fig:sosrules} (else)} \\
				\eeq& \mbigadd_{\sosTrans{\sosState{C}{\state}{n}{\beta}}{a}{\sosState{C'}{\state'}{n+1}{\beta'}}}
					a \smop \wlp{C'}{\mnull}(\state') \\
				\cmmnt{Induction Hypothesis on $C_2$} \\
				\eeq& \mbigadd_{\sosTrans{\sosState{C}{\state}{n}{\beta}}{a}{\sosState{C'}{\state'}{n+1}{\beta'}}}
					a \smop \olp{C'}(\state') \\
				\cmmnt{\cref{thm:liberal-wgcl-functional}} \\
				\eeq& \olp{C}(\state)
            \end{align*}

		\end{enumerate}

	\item
        The program $C$ is of the form $\BRANCH{C_1}{C_2}$.
		\begin{align*}
			& \wlp{C}{\mnull}(\state) \\
			\cmmnt{Def.\ of $\wlpsymbol$} \\
			\eeq& \wlp{C_1}{\mnull}(\state) \mmadd \wlp{C_2}{\mnull}(\state) \\
			\cmmnt{Induction Hypothesis} \\
			\eeq& \olp{C_1}(\state) \mmadd \olp{C_2}(\state) \\
			\cmmnt{\cref{fig:sosrules} (l.\ branch), (r.\ branch)} \\
			\eeq& \sbigadd_{\sosTrans{\sosState{C}{\state}{n}{\beta}}{a}{\sosState{C'}{\state'}{n+1}{\beta'}}}
				a \smop \olp{C'}(\state') \\
			\cmmnt{\cref{thm:liberal-wgcl-functional}} \\
			\eeq& \olp{C}(\state)
		\end{align*}

	\item
        The program $C$ is of the form $\WHILEDO{\guard}{C_1}$.
		We show that $\olp{C}$ is a fixed point of the loop characteristic function $\charwlp{\guard}{C}{\mnull}$ for postweight $\mnull$, where
		\begin{align}
			\charwlp{\guard}{C}{\mnull}(X)
			&\qeq \iverson{\neg\guard} \ivop \mnull \mmadd \iverson{\guard} \ivop \wlp{C_1}{X} \nonumber \\
			&\qeq \iverson{\guard} \ivop \wlp{C_1}{X} \nonumber \\
			\cmmnt{\cref{thm:wlp-decomp}} \\
			\label{eq:wlp-charfun-zero}
			&\qeq \iverson{\guard} \ivop (\wp{C_1}{X} \mmadd \wlp{C_1}{\mnull})
			~.
		\end{align}
		Let $\state \in \States$ be a state.
		There are the following two exclusive cases:
        \begin{enumerate}
		\item
            Case $\state \models \guard$.
			\begin{align*}
				& \charwlp{\guard}{C}{\mnull}(\olp{C})(\state) \\
				\cmmnt{(\ref{eq:wlp-charfun-zero}) above} \\
				\eeq& \big( \iverson{\guard} \iivop (\wp{C_1}{ \olp{C} } \mmadd \wlp{C_1}{\mnull}) \big)(\state) \\
				\cmmnt{Case $\state \models \guard$} \\
				\eeq& ( \wp{C_1}{ \olp{C} } \mmadd \wlp{C_1}{\mnull} )(\state) \\
				\cmmnt{\cref{thm:wp-operational-semantics}, Induction Hypothesis on $C_1$} \\
				\eeq& ( \op{C_1}{ \olp{C} } \mmadd \olp{C_1} )(\state) \\
				\cmmnt{\cref{thm:olp-composition}} \\
				\eeq& \olp{\COMPOSE{C_1}{C}}(\state) \\
				\cmmnt{\cref{fig:sosrules} (while)} \\
				\eeq& \sbigadd_{\sosTrans{\sosState{C}{\state}{n}{\beta}}{a}{\sosState{C'}{\state'}{n+1}{\beta'}}}
					a \smop \olp{C'}(\state') \\
				\cmmnt{\cref{thm:liberal-wgcl-functional}} \\
				\eeq& \olp{C}(\state)
			\end{align*}

		\item
            Case $\state \not\models \guard$.
			Then, $\pathsOfLengthStartingIn{n}{\sosStateAbbr{C}{\state}} = \emptyset$ by \cref{fig:sosrules} for any $n \geq 2$, hence $\olp{C}(\state) = \mnull$.
			\begin{align*}
				& \charwlp{\guard}{C}{\mnull}(\olp{C})(\state) \\
				\cmmnt{(\ref{eq:wlp-charfun-zero}) above} \\
				\eeq& \big( \iverson{\guard} \iivop (\wp{C_1}{ \olp{C} } \mmadd \wlp{C_1}{\mnull}) \big)(\state) \\
				\cmmnt{Case $\state \not\models \guard$} \\
				\eeq& \mnull \\
			\end{align*}

		\end{enumerate}
		Overall, we have $\charwlp{\guard}{C}{\mnull}(\olp{C}) = \olp{C}$.
		As $\wlp{C}{\mnull}$ is defined as \emph{greatest} fixed point, $\olp{C} \natord \wlp{C}{\mnull}$.
		Next, we show that the greatest fixed point of $\wlp{C}{\mnull}$ is at most $\olp{C}$.
		To that end, we use induction on the path length $\ell \in \Nats$.
		Denote
		\[
			\olplen{\ell}{C}
			\qcoloneqq \sbigadd_{\compPath \in \pathsOfLengthStartingIn{\ell}{\sosStateAbbr{C}{\state}}}
				\pathWeight(\compPath) \smop \mTop ~.
		\]

		The following forms the induction base.
		Let $\ell = 0$.
		Then, $\pathsOfLengthStartingIn{\ell}{\sosStateAbbr{C}{\state}} = \emptyset$ and
		\[
			\wlp{C}{\mnull} \nnatord \mTop \eeq \olplen{0}{C} ~.
		\]
		
		The following forms the induction step.
		Fully analogous to \cref{thm:wp-functional}, $\ext{\wlpsymbol}$ is a $\wgcl$-functional.
		Let $\ell \in \Nats$ such that $\wlp{C}{\mnull} \natord \olplen{\ell}{C}$, and let $\state \in \States$.
		\begin{align*}
			& \wlp{C}{\mnull}(\state) \\
			\cmmnt{$\wgcl$-functional} \\
			\eeq& \sbigadd_{\sosTrans{\sosState{C}{\state}{n}{\beta}}{a}{\sosState{C'}{\state'}{n+1}{\beta'}}}
				a \smop \wlp{C'}{\mnull}(\state') \\
			\cmmnt{Induction Hypothesis} \\
			\nnatord& \sbigadd_{\sosTrans{\sosState{C}{\state}{\ell}{\beta}}{a}{\sosState{C'}{\state'}{n+1}{\beta'}}}
				a \smop \olplen{\ell}{C'} \\
			\cmmnt{\cref{thm:liberal-wgcl-functional}} \\
			\eeq& \olplen{\ell+1}{C}(\state)
			~.
		\end{align*}
		By Definition of $\meet$, $\wlp{C}{\mnull} \natord \olp{C}$.
		Both inequalities imply $\wlp{C}{\mnull} = \olp{C}$.
    \end{itemize}
\end{proof}


\newpage
\section{Proofs of Section \ref{sec:loop-verification}}
\label{app:proofs5}

\subsection{Proof of \cref{thm:strong_term_for_state_unique_fp}}

We will use the following lemma which is a useful alternative characterization of certain termination:

\begin{lemma}
    \label{thm:koenig}
    $C$ is certainly terminating for initial state $\state \in \States$ iff there exists $b \in \Nats$ such that $\pathsOfLengthStartingIn{n}{\sosStateAbbr{C}{\sigma}} = \emptyset$ for all $n \geq b$, i.e. the length of all computation paths starting in $\sosStateAbbr{C}{\sigma}$ is bounded.
\end{lemma}
\begin{proof}
    The direction from right to left is trivial.
    For the other direction we show that the sub-tree $\compGraph'$ induced by $\succsConfs(\sosStateAbbr{C}{\state}) \coloneq \bigcup_{i \geq 0}\succConfs^i(\sosStateAbbr{C}{\state})$ is finite which implies that all paths have bounded length.
    Assume towards contradiction that $\compGraph'$ is infinite.
    By definition, $\compGraph'$ is finitely branching.
    Thus, by Kőnig's classic infinity lemma, there exists an infinite path $\conf_1 \conf_2 \ldots$ in $\compGraph'$.
    But then there also exists an infinite path starting in $\sosStateAbbr{C}{\state}$ since $\conf_1$ is reachable from there.
    This is a contradiction to the assumption that $C$ is certainly terminating for $\state$.
\end{proof}

\begin{proof}[Proof of \cref{thm:strong_term_for_state_unique_fp}]
    First, since $C'$ is universally certainly terminating we have by \cref{thm:koenig} that there exists $b \in \Nats$ such that $\pathsOfLengthStartingIn{n}{\sosStateAbbr{C}{\state'}} = \emptyset$ for all $n \geq b$ and $\state' \in \States$ and thus by \cref{thm:wlp-operational-semantics} and \cref{thm:wlp-decomp}, it holds for all $g \in \Weightings$ that
    \[
		\wlp{C'}{g} = \wp{C'}{g} ~.
    \]
    This implies that the $\wpsymbol$- and $\wlpsymbol$-characteristic functions of the loop $C = \WHILEDO{\guard}{C'}$ are equal, i.e.
	\begin{align*}
		& \charwp{\guard}{C'}{f} \\
		\eeq& \lam{X} \iverson{\neg\guard} \ivop f \mmadd \iverson{\guard} \ivop \wp{C'}{X} \\
		\eeq& \lam{X} \iverson{\neg\guard} \ivop f \mmadd \iverson{\guard} \ivop \wlp{C'}{X} \\
		\eeq& \charwlp{}{}{f} ~.
	\end{align*}
    Similarly, since $C$ is certainly terminating on $\state$, we also have
    \[
		\wlp{C'}{f}(\state) = \wp{C'}{f}(\state)
    \]
    and thus
    \[
        (\gfpop \charwp{\guard}{C'}{f})(\state) \qeq (\lfpop \charwp{\guard}{C'}{f})(\state) ~.
    \]
    Since $I_1$ and $I_2$ are fixed points of $\charwp{\guard}{C'}{f}$, 
    \[
		\lfpop \charwp{\guard}{C'}{f} \nnatord I_1, I_2 \nnatord \gfpop \charwp{\guard}{C'}{f}
    \]
    and thus $I_1(\sigma) = I_2(\sigma)$.
\end{proof}

\newpage
\section{Annotated Programs}

\subsection{Ski Rental}
\label{sec:app_ski_rental}

	%
	\begin{adjustbox}{max width=\textwidth}
		\centering
		
		\begin{minipage}{0.45\textwidth}
			\begin{align*}
				& \eqannotate{\varvaclen \sadd \varbuycost \gray{\eeq I}} \\
				%
				%
				%
				&\phiannotate{\iverson{\varvaclen = 0}\ivop \sone
					\ssadd  \iverson{\varvaclen > 0}\smul(\varvaclen 
					\sadd
					\varbuycost) 
				} \\
				&\WHILE{\varvaclen > 0} \\
				& \qquad \eqannotate{\varvaclen 
					\sadd
					\varbuycost
				} \\
				& \qquad \wpannotate{1 \smul (\varvaclen - 1)
					\ssadd
					\varbuycost
				} \\
				& \qquad \ASSIGN{\varvaclen}{\varvaclen - 1}\fatsemi \\
				& \qquad \eqannotate{1 \smul \varvaclen
					\ssadd
					\varbuycost
				} \\
				& \qquad \eqannotate{1 \smul \varvaclen \ssadd 1 \smul \varbuycost
					\ssadd
					\varbuycost
				} \\
				& \qquad \wpannotate{1 \smul (\varvaclen \sadd \varbuycost)
					\ssadd
					\varbuycost
				} \\
				& \qquad \{ \qquad \textnormal{\textcolor{gray}{(* rent *)}}  \\
				& \qquad \qquad \wpannotate{1 \smul (\varvaclen \sadd \varbuycost)} \\
				&\qquad \qquad \WEIGH{1} \\
				& \qquad \qquad \annotate{\varvaclen \sadd \varbuycost} \\
				&\qquad \} \BranchSymbol \{  \qquad \textnormal{\textcolor{gray}{(* buy *)}} \\
				& \qquad \qquad \eqannotate{\varbuycost } \\
				& \qquad \qquad \wpannotate{\varbuycost \smul \sone} \\
				& \qquad \qquad \WEIGH{\varbuycost} \fatsemi \\
				& \qquad \qquad \eqannotate{\sone} \\
				& \qquad \qquad \eqannotate{0} \\
				& \qquad \qquad \wpannotate{0 \sadd \varbuycost} \\
				&\qquad \qquad \ASSIGN{\varvaclen}{0}  \qquad \textnormal{\textcolor{gray}{(* terminate *)}}\\
				& \qquad \qquad \annotate{\varvaclen \sadd \varbuycost} \\
				&\qquad \} \\
				& \qquad \starannotate{\varvaclen \sadd \varbuycost \gray{\eeq I}} \\
				& \} \\
				& \annotate{\sone}
			\end{align*}
		\end{minipage}
		\qquad
		\begin{minipage}{0.45\textwidth}
			\begin{align*}
				&\wpannotate{\scriptstyle \iverson{\varvaclen=0} \ivop 0    \mmadd \iverson{0< \varbuycost}\ivop 
					\big( (2\varbuycost  -1) \sadd \iverson{\varvaclen \leq \varbuycost - 1}\ivop \varvaclen \big) 
				} \\
				&\ASSIGN{\varvaccount}{0}\fatsemi \\
				&  \eqannotate{\scriptstyle \iverson{\varvaclen = 0} \sadd  \iverson{\varvaccount  \geq \varbuycost} \ivop \varbuycost  \sadd  \iverson{\varvaccount  < \varbuycost} }   \\
				&\qquad \quad \quad     \phantannotate{\scriptstyle \ivop 
					\big( (2\varbuycost - \varvaccount -1) \sadd \iverson{\varvaclen \leq \varbuycost - \varvaccount - 1}\ivop \varvaclen \big)  \gray{\eeq I} } \\
				&  \phiannotate{\scriptstyle \iverson{\varvaclen = 0} \ivop 0 \sadd \iverson{\varvaclen >0} \ivop \big(  \iverson{\varvaccount +1 \geq \varbuycost} \ivop \varbuycost  \sadd  \iverson{\varvaccount +1 < \varbuycost} }   \\
				&\qquad \quad \quad     \phantannotate{\scriptstyle \ivop 
					\big( (2\varbuycost - \varvaccount -1) \sadd \iverson{\varvaclen \leq \varbuycost - \varvaccount - 1}\ivop \varvaclen \big) \big) } \\
				&\WHILE{\varvaclen > 0} \\
				& \qquad \eqannotate{\scriptstyle  \iverson{\varvaccount+1 \geq \varbuycost} \ivop \varbuycost  \sadd  \iverson{\varvaccount+1 < \varbuycost} }   \\
				&\qquad \quad \quad     \phantannotate{\scriptstyle \ivop \big(  
					(2\varbuycost - \varvaccount -1) \sadd \iverson{\varvaclen \leq \varbuycost - \varvaccount - 1}\ivop \varvaclen \big)  } \\
				& \qquad \eqannotate{\scriptstyle  \iverson{\varvaccount+1 \geq \varbuycost} \ivop \varbuycost  \sadd  \big( 1 \smul \iverson{\varvaclen=1} \ivop 0}  \\
				&\qquad \quad \quad     \phantannotate{\scriptstyle   \sadd \iverson{\varvaccount+1 < \varbuycost}\ivop 
					\big( (2\varbuycost - \varvaccount -1) \sadd \iverson{\varvaclen \leq \varbuycost - \varvaccount - 1}\ivop \varvaclen \big) \big) } \\
				& \qquad \eqannotate{\scriptstyle  \iverson{\varvaccount+1 \geq \varbuycost} \ivop \varbuycost  \sadd  \iverson{\varvaccount+1 < \varbuycost}}  \\
				&\qquad \quad  \phantannotate{\scriptstyle \ivop 1 \smul \big( \iverson{\varvaclen=1} \ivop 0 \mmadd \iverson{\varvaccount +1 \geq \varbuycost}\ivop \varbuycost \big.} \\
				&\qquad \quad \quad     \phantannotate{\scriptstyle   \sadd \iverson{\varvaccount+1 < \varbuycost}\ivop 
					\big( (2\varbuycost - \varvaccount -2) \sadd \iverson{\varvaclen \leq \varbuycost - \varvaccount - 1}\ivop (\varvaclen - 1) \big) \big) } \\
				& \qquad \wpannotate{\scriptstyle \iverson{\varvaccount+1 < \varbuycost} \ivop (1 \smul I\subst{\varvaccount}{\varvaccount+1}\subst{n}{n-1}) \sadd \iverson{\varvaccount+1 \geq \varbuycost} \ivop \varbuycost }  \\
				& \qquad \ASSIGN{\varvaclen}{\varvaclen - 1}\fatsemi \\
				& \qquad \wpannotate{\scriptstyle \iverson{\varvaccount+1 < \varbuycost} \ivop (1 \smul I\subst{\varvaccount}{\varvaccount+1}) \sadd \iverson{\varvaccount+1 \geq \varbuycost} \ivop \varbuycost }  \\
				& \qquad \ASSIGN{\varvaccount}{\varvaccount + 1}\fatsemi \\
				& \qquad \wpannotate{\scriptstyle \iverson{\varvaccount < \varbuycost} \ivop (1 \smul I) \mmadd \iverson{\varvaccount \geq \varbuycost} \ivop \varbuycost }  \\
				& \qquad \IF{\varvaccount < \varbuycost} \\
				& \qquad \wpannotate{\scriptstyle 1 \smul I	} \\
				& \qquad \qquad \WEIGH{1} \\
				& \qquad \annotate{\scriptstyle I	} \\
				& \qquad \ELSE \\
				& \qquad \eqannotate{\scriptstyle \varbuycost} \\
				& \qquad \wpannotate{\scriptstyle \varbuycost \smul 0} \\
				& \qquad \qquad \WEIGH{\varbuycost}\fatsemi \\
				& \qquad \eqannotate{\scriptstyle 0} \\
				& \qquad \wpannotate{\scriptstyle I\subst{n}{0}} \\
				&\qquad \qquad \ASSIGN{\varvaclen}{0} \\
				& \qquad \annotate{\scriptstyle I} \\
				& \qquad \} \\
				& \qquad \starannotate{\scriptstyle \iverson{\varvaclen=0} \ivop 0 \mmadd \iverson{\varvaccount \geq \varbuycost}\ivop \varbuycost } \\
				&\qquad \quad     \phantannotate{\scriptstyle   \mmadd \iverson{\varvaccount < \varbuycost}\ivop 
					\big( (2\varbuycost - \varvaccount -1) \mmadd \iverson{\varvaclen \leq \varbuycost - \varvaccount - 1}\ivop \varvaclen \big) \gray{\eeq I}
				} \\
				& \} \\
				& \annotate{\scriptsize \sone}
			\end{align*}
		\end{minipage}
	\end{adjustbox}
	%

\pagebreak
\subsection{Mutual Exclusion}
\label{sec:app_mutual_exclusion}

\begin{adjustbox}{max height=9.5cm}
    \begin{minipage}{1\textwidth}
    	\begin{align*}
    		& \succeqannotate{\scriptstyle\iverson{1 \leq \varaux \leq \varnumprocs \wedge \varprocs{\varaux} = \valw \wedge \varsem=0} \ivop \letw_\varaux^\omega} \\
    		& \succeqannotate{\scriptstyle  \letw_\varaux  \smul \iverson{1 \leq \varaux \leq \varnumprocs \wedge \varprocs{\varaux} = \valw \wedge \varsem=0} \ivop \letw_\varaux^\omega} \\
    		& \phiannotate{\scriptstyle \sbigadd_{\valchoose=1}^{N} \iverson{\varprocs{\valchoose} = \valw} \ivop (\letw_\valchoose  \smul \iverson{1 \leq \varaux \leq \varnumprocs \wedge \varprocs{\varaux} = \valw \wedge \varsem=0} \ivop \letw_\varaux^\omega)} \\
    		&\WHILE{\true} \\
    		& \qquad \wlpannotate{\scriptstyle \sbigadd_{\valchoose=1}^{N} \iverson{\varprocs{\valchoose} = \valw} \ivop (\letw_\valchoose  \smul \iverson{1 \leq \varaux \leq \varnumprocs \wedge \varprocs{\varaux} = \valw \wedge \varsem=0} \ivop \letw_\varaux^\omega)} \\
    		& \qquad \BIGBRANCH{\valchoose=1}{\varnumprocs}{\ASSIGN{\varprocid}{\valchoose}} \fatsemi \\
    		& \qquad \wlpannotate{\scriptstyle \iverson{\varprocsid = \valw} \ivop (\letw_\varprocid  \smul \iverson{1 \leq \varaux \leq \varnumprocs \wedge \varprocs{\varaux} = \valw \wedge \varsem=0} \ivop \letw_\varaux^\omega)} \\
    		& \qquad \IF{\varprocsid = \valn} \\
    		& \qquad  \qquad \wlpannotate{\scriptstyle \emptyset} \\
    		& \qquad \qquad \ASSIGN{\varprocsid}{\valw} \\
    		& \qquad  \qquad \succeqannotate{\scriptstyle \emptyset} \\
    		& \qquad  \qquad \annotate{\scriptstyle \iverson{1 \leq \varaux \leq \varnumprocs \wedge \varprocs{\varaux} = \valw \wedge \varsem=0} \ivop \letw_\varaux^\omega} \\
    		& \qquad \ELSEIF{\varprocsid = \valw} \\
    		& \qquad \qquad \succeqannotate{\scriptstyle  \letw_\varprocid  \smul \iverson{1 \leq \varaux \leq \varnumprocs \wedge \varprocs{\varaux} = \valw \wedge \varsem=0} \ivop \letw_\varaux^\omega} \\
    		& \qquad \qquad \wlpannotate{\scriptstyle \iverson{\varsem = 0} \ivop ( \letw_\varprocid  \smul \iverson{1 \leq \varaux \leq \varnumprocs \wedge\varprocs{\varaux} = \valw \wedge \varsem=0} \ivop \letw_\varaux^\omega)} \\
    		& \qquad \qquad \IF{\varsem > 0} \\
    		& \qquad \qquad \qquad  \wlpannotate{\scriptstyle \emptyset} \\
    		& \qquad \qquad \qquad \ASSIGN{\varsem}{\varsem - 1}\fatsemi \\
    		& \qquad \qquad \qquad \ASSIGN{\varprocsid}{\valc}\fatsemi \\
    		& \qquad \qquad \qquad \WEIGH{ \letc_\varprocid } \\
    		& \qquad \qquad \qquad  \succeqannotate{\scriptstyle  \emptyset}\\
    		& \qquad \qquad \qquad  \annotate{\scriptstyle \iverson{1 \leq \varaux \leq \varnumprocs \wedge \varprocs{\varaux} = \valw \wedge \varsem=0} \ivop \letw_\varaux^\omega}  \\
    		& \qquad \qquad \ELSE\\
            & \qquad \qquad \qquad \wlpannotate{\scriptstyle  \letw_\varprocid  \smul \iverson{1 \leq \varaux \leq \varnumprocs \wedge \varprocs{\varaux} = \valw \wedge \varsem=0} \ivop \letw_\varaux^\omega} \\
            & \qquad \qquad \qquad \WEIGH{ \letw_\varprocid  } \\
            & \qquad \qquad \qquad \annotate{\scriptstyle \iverson{1 \leq \varaux \leq \varnumprocs \wedge \varprocs{\varaux} = \valw \wedge \varsem=0} \ivop \letw_\varaux^\omega} \\
            & \qquad \qquad \} \\
            & \qquad\qquad   \annotate{\scriptstyle \iverson{1 \leq \varaux \leq \varnumprocs \wedge \varprocs{\varaux} = \valw \wedge \varsem=0} \ivop \letw_\varaux^\omega} \\
            & \qquad \ELSEIF{\varprocsid = \valc} \\
            & \qquad \qquad \wlpannotate{\scriptstyle \emptyset} \\
            & \qquad \qquad \ASSIGN{\varsem}{\varsem + 1}\fatsemi  \\
            & \qquad \qquad \ASSIGN{\varprocsid}{\valn}\fatsemi  \\
            & \qquad \qquad \WEIGH{ \letr_\varprocid  }  \\
            & \qquad \qquad  \succeqannotate{\scriptstyle \emptyset} \\
            & \qquad \qquad  \annotate{\scriptstyle \iverson{1 \leq \varaux \leq \varnumprocs \wedge \varprocs{\varaux} = \valw \wedge \varsem=0} \ivop \letw_\varaux^\omega} \\
            & \qquad \} \\
            & \qquad  \starannotate{\scriptstyle \iverson{1 \leq \varaux \leq \varnumprocs \wedge \varprocs{\varaux} = \valw \wedge \varsem=0} \ivop \letw_\varaux^\omega} \\
            & \} \\
            & \wlpannotate{\scriptstyle \snull}
    	\end{align*}
    \end{minipage}
\end{adjustbox}

\pagebreak
\subsection{Path Counting}
\label{app:path_counting}

\begin{figure}[th!]
    \begin{subfigure}[b]{0.4\textwidth}
        \centering
        \begin{align*}
        & \COMMENT{$\ASSIGN{res}{\texttt{[]}} \,\fatsemi$} \\
        & \COMPOSE{\ASSIGN{m}{0}}{\ASSIGN{c}{0}} \,\fatsemi \\
        & \WHILE{n > 0} \\
        & \quad \ASSIGN{n}{n-1} \\
        & \quad \{ \\
        & \quad\quad \COMMENT{$\APPEND{res}{0} \,\fatsemi$}  \\
        & \quad\quad \ASSIGN{c}{0} \\
        & \quad \} \BranchSymbol \{ \\
        & \quad\quad \COMMENT{$\APPEND{res}{1} \,\fatsemi$}  \\
        & \quad\quad \ASSIGN{c}{c+1} \,\fatsemi \\
        & \quad\quad \ASSIGN{m}{\max(m,c)} \\
        & \quad \} \\
        & \}
        \end{align*}
        \vspace{25mm}
        \caption{The program $\Cfib$ non-deterministically \enquote{creates} bitstrings of length $n$. The maximum number of consecutive 1's is stored in variable $m$.}
    \end{subfigure}
    \hfill
    \begin{subfigure}[b]{0.55\textwidth}
        \centering
        \begin{align*}
        & \eqannotate{\scriptstyle I \quad\quad\quad \gray{(\text{because } I\subst{n}{0} \eeq \iverson{m \leq 1})}} \\
        & \phiannotate{\scriptstyle \iverson{n=0} \ivop \iverson{m \leq 1} \ssadd \iverson{n > 0} \ivop I} \\
        & \WHILE{n > 0} \\
        & \quad \eqannotate{\scriptstyle \iverson{m \leq 1} \ivop (\iverson{c=0} \ivop \fib(n+2) \ssadd \iverson{c>0} \ivop \fib(n+1) ) \gray{\eeq I}} \\
        & \quad \eqannotate{\scriptstyle \iverson{m \leq 1} \ivop (\iverson{c=0} \ivop (\fib(n+1) \sadd  \fib(n)) \ssadd \iverson{c>0} \ivop \fib(n+1) )} \\
        & \quad \eqannotate{\scriptstyle \iverson{m \leq 1} \ivop (\fib(n+1) \ssadd \iverson{c=0} \ivop \fib(n))} \\
        & \quad \wpannotate{\scriptstyle \iverson{m \leq 1} \ivop \fib(n+1) \ssadd \iverson{m \leq 1 \land c=0} \ivop \fib(n)} \\
        & \quad \ASSIGN{n}{n-1} \\
        & \quad \wpannotate{\scriptstyle \iverson{m \leq 1} \ivop \fib(n+2) \ssadd \iverson{m \leq 1 \land c=0} \ivop \fib(n+1)} \\
        & \quad \{ \\
        & \quad \quad \eqannotate{\scriptstyle \iverson{m \leq 1} \ivop \fib(n+2)} \\
        & \quad \quad \wpannotate{\scriptstyle \iverson{m \leq 1} \ivop (\iverson{0=0}\ivop \fib(n+2) \sadd \iverson{0>0}\ivop \fib(n+1))} \\
        & \quad\quad \ASSIGN{c}{0} \\
        & \quad \quad \annotate{\scriptstyle \iverson{m \leq 1} \ivop (\iverson{c=0}\ivop \fib(n+2) \sadd \iverson{c>0}\ivop \fib(n+1))} \\
        & \quad \} \BranchSymbol \{ \\
        & \quad\quad \eqannotate{\scriptstyle \iverson{m \leq 1 \land c=0} \ivop \fib(n+1)} \\
        & \quad\quad \wpannotate{\scriptstyle \iverson{m \leq 1 \land c+1 \leq 1} \ivop (\iverson{c+1=0}\ivop \fib(n+2) \sadd \iverson{c+1>0}\ivop \fib(n+1))} \\
        & \quad\quad \ASSIGN{c}{c+1} \,\fatsemi \\
        & \quad\quad \eqannotate{\scriptstyle \iverson{m \leq 1 \land c \leq 1} \ivop (\iverson{c=0}\ivop \fib(n+2) \sadd \iverson{c>0}\ivop \fib(n+1))} \\
        & \quad\quad \annotate{\scriptstyle \iverson{\max(m,c) \leq 1} \ivop (\iverson{c=0}\ivop \fib(n+2) \sadd \iverson{c>0}\ivop \fib(n+1))} \\
        & \quad\quad \ASSIGN{m}{\max(m,c)} \\
        & \quad\quad \annotate{\scriptstyle \iverson{m \leq 1} \ivop (\iverson{c=0}\ivop \fib(n+2) \sadd \iverson{c>0}\ivop \fib(n+1))} \\
        & \quad \} \\
        & \quad \starannotate{\scriptstyle \iverson{m \leq 1} \ivop (\iverson{c=0}\ivop \fib(n+2) \sadd \iverson{c>0}\ivop \fib(n+1)) \gray{\eeq I} } \\
        & \} \\
        & \annotate{\scriptstyle \iverson{m\leq 1}}
        \end{align*}
        \caption{$\Cfib$ with annotations for verifying the loop invariant.\newline \newline \newline}
        \label{fig:fibAnnotated}
    \end{subfigure}
	%
\end{figure}

\pagebreak
\subsection{Knapsack}

\begin{figure}[htb]
    
    \centering
    \begin{align*}
    & \eqannotate{\scriptstyle 1 \ssadd \iverson{x \geq 8}\ivop 1 \ssadd \iverson{x \geq 13}\ivop 1} \\
    & \wpannotate{\scriptstyle \iverson{0 \leq 1 \land 0 \geq 8-x} \ivop 1\ssadd \iverson{0 = 0}\ivop 1 \ssadd \iverson{0 \leq 4 \land 0 \geq 8} \ivop 1} \\
    & \quad  \phantannotate{\scriptstyle \ssadd \iverson{0 \leq 3 \land 0 \geq 13-x}\ivop 1 \ssadd \iverson{0 \leq 2 \land 0 \geq 5}\ivop 1 \ssadd \iverson{0 \leq 6 \land 0 \geq 13}\ivop 1} \\
    &\COMPOSE{\ASSIGN{t}{0}}{\ASSIGN{r}{0}} \,\fatsemi \\
    & \eqannotate{\scriptstyle \iverson{t \leq 1 \land r \geq 8-x}\ivop 1 \ssadd \iverson{t = 0}\ivop 1 \ssadd \iverson{t \leq 4 \land r \geq 8} \ivop 1} \\
    & \quad  \phantannotate{\scriptstyle \ssadd \iverson{t \leq 3 \land r \geq 13-x}\ivop 1 \ssadd \iverson{t \leq 2 \land r \geq 5}\ivop 1 \ssadd \iverson{t \leq 6 \land r \geq 13}\ivop 1} \\
    & \wpannotate{\scriptstyle \iverson{t+2 \leq 3 \land r+5 \geq 13-x}\ivop 1 \ssadd \iverson{t+2 \leq 2 \land r+5 \geq 5}\ivop 1 \ssadd \iverson{t+2 \leq 6 \land r+5 \geq 13} \ivop 1} \\
    & \quad  \phantannotate{\scriptstyle \ssadd \iverson{t \leq 3 \land r \geq 13-x} \ivop 1\ssadd \iverson{t \leq 2 \land r \geq 5}\ivop 1 \ssadd \iverson{t \leq 6 \land r \geq 13}\ivop 1} \\
    &\BRANCH{\COMPOSE{\ASSIGN{t}{t+2}}{\ASSIGN{r}{r+5}}}{\SKIP} \,\fatsemi \\
    & \eqannotate{\scriptstyle \iverson{t \leq 3 \land r \geq 13-x} \ivop 1\ssadd \iverson{t \leq 2 \land r \geq 5} \ivop 1\ssadd \iverson{t \leq 6 \land r \geq 13}\ivop 1} \\
    & \wpannotate{\scriptstyle \iverson{t+3 \leq 2 \land r+x \geq 5}\ivop 1 \ssadd \iverson{t+3 \leq 6 \land r+x \geq 13} \ivop 1\ssadd \iverson{t \leq 2 \land r \geq 5} \ivop 1\ssadd \iverson{t \leq 6 \land r \geq 13}\ivop 1} \\
    &\BRANCH{\COMPOSE{\ASSIGN{t}{t+3}}{\ASSIGN{r}{r+x}}}{\SKIP} \,\fatsemi \\
    & \eqannotate{\scriptstyle \iverson{t \leq 2 \land r \geq 5}\ivop 1 \ssadd \iverson{t \leq 6 \land r \geq 13}\ivop 1} \\
    & \wpannotate{\scriptstyle \iverson{t+4 \leq 6 \land r+8 \geq 13}\ivop 1 \ssadd \iverson{t \leq 6 \land r \geq 13}\ivop 1} \\
    &\BRANCH{\COMPOSE{\ASSIGN{t}{t+4}}{\ASSIGN{r}{r+8}}}{\SKIP} \\
    & \annotate{\scriptstyle \iverson{t \leq 6 \land r \geq 13} \ivop 1}
    \end{align*}
    \caption{Program modelling the Knapsack problem (see \Cref{sec:knapsack}). The three branchings correspond to choosing the respective task or not. Variables $t$ and $r$ stand for time and reward, respectively.}
\end{figure}


\newpage

\section{Further Applications of Weighted Programming}

\subsection{Reasoning about $n$-th best solutions}
\label{sec:knapsack}
\begin{center}
    \vspace{-1ex}
    \begin{adjustbox}{max width=1\linewidth}
        \fbox{%
            \parbox{1.2\textwidth}{%
                \smallskip%
                \hspace*{.5em}
                \begin{minipage}{.33\linewidth}
                    \begin{tabular}{l@{\quad}l}
                        \textbf{Field:}& Discrete Optimization\\
                        \textbf{Problem:}& Knapsack Problem
                    \end{tabular}%
                \end{minipage}%
                \hfill
                \begin{minipage}{.37\linewidth}
                    \begin{tabular}{l@{\quad}l}
                        \textbf{Model:}& Optimization problem \\
                        \textbf{Semiring:}& Natural numbers
                    \end{tabular}%
                \end{minipage}%
                \hfill
                \begin{minipage}{.25\linewidth}\begin{tabular}{l@{\quad}l}
                        \textbf{Techniques:}& $\wpsymbol$
                    \end{tabular}%
                \end{minipage}%
                \smallskip%
            }%
        }%
    \end{adjustbox}
\end{center}
We apply path counting by weighted programming to \emph{quantify the ambiguity} of the non-determinism in a program. Assume we are given a nondeterministic weighted program $C$ modeling a discrete optimization problem. Each branching $\BRANCH{...}{...}$ in $C$ corresponds to a possible choice to be made by an optimizer.
Further, we assume that all program variables are either initialized explicitly by the program or read-only.
We will refer to such read-only variables as \emph{program parameters}.
Given an initial parameters state $\sigma$, the optimization goal is to reach a final state $\sigma'$ satisfying a given predicate $\guard$ that maximizes $\sigma'(r)$ for some fixed special program variable $r$ --- a \enquote{reward}, or payoff. Examples for the predicate $\guard$ include, e.g.\ thresholds on time- or energy consumption.  A solution to the optimization problem modeled by $C$ corresponds to a \emph{determinization} $C'$ obtained from $C$ by replacing every nondeterministic choice by a deterministic one.
The reward $\sigma'(r)$ achieved by $C'$ on initial state $\sigma$ is called the \emph{score} of $C'$ w.r.t.\ $\sigma$.
We can view the scores as a function $\rho \colon \States \to \Nats$ that only depends on the parameter variables.
The solution $C'$ is called \emph{valid} if for all parameters $\state$, it reaches a final state satisfying $\guard$.
Note that valid solution are not necessarily optimal.

Given some problem parameters $\sigma$, we are now interested in determining the \emph{rank} of a given solution $C'$ relative to the optimal solution implicitly encoded in $C$, e.g.\ \enquote{is $C'$ among the 3 best solutions?}.
Assume that $C'$ is valid.
Then $C'$ is among the best $n$ solutions under parameters $\state$ iff
\[
    \wp{C}{\iverson{\guard \land r \geq \rho} \ivop 1}(\state) \leq n ~.
\]

We illustrate determining ranks of solutions by $\wpsymbol$-reasoning by means of the Knapsack problem:
Suppose we operate a cloud computer earning money by completing computational tasks. There are currently 3 tasks in the queue,
which take 2, 3, and 4 hours to complete, respectively, and generate a reward of 5, $x$, and 8 euros, where $x$ is a program variable (considered an input parameter).
Our goal is to maximize the total reward earned within $6$ hours, i.e. $\guard = (t \leq 6)$.
Hence, we have to decide on a subset of the three tasks to stay within this time limit since completing all three would take 7 hours.
This optimization problem is modeled by the program $C$ in \cref{fig:knapsack}. We initialize the variables $t$ (for time) and $r$ (for reward) by $0$ and may subsequently choose a subset of the available tasks. These choices are modeled by the three nondeterministic branchings.

Choosing the first and the third task always yields a valid solution ($t \leq 6$) generating an accumulated reward of $13$ euros. This solution is modeled by the determinization $C'$ of $C$ in \cref{fig:knapsack}.
We now compute
\[
    \wp{C}{\iverson{t \leq 6 \land r \geq 13 } \ivop 1}   = 1 \sadd \iverson{x \geq 8} \ivop 1 \sadd \iverson{x \geq 13} \ivop 1 ~,
\]
where we can readily read off that $C'$ is (i) the unique optimal solution if $x < 8$, (ii) among the two best solutions if $8 \leq x < 13$, and (iii) among the three best solutions else.

\begin{figure}[t]
    \begin{minipage}{0.45\textwidth}
        \begin{align*}
        & \wpannotate{1 \ssadd \iverson{x \geq 8} \ivop 1 \ssadd \iverson{x \geq 13} \ivop 1}  \\
        %
        %
        &\COMPOSE{\ASSIGN{t}{0}}{\ASSIGN{r}{0}} \,\fatsemi \\
        %
        %
        %
        &\BRANCH{\COMPOSE{\ASSIGN{t}{t+2}}{\ASSIGN{r}{r+5}}}{\SKIP} \,\fatsemi \\
        %
        %
        %
        &\BRANCH{\COMPOSE{\ASSIGN{t}{t+3}}{\ASSIGN{r}{r+x}}}{\SKIP} \,\fatsemi \\
        %
        %
        %
        &\BRANCH{\COMPOSE{\ASSIGN{t}{t+4}}{\ASSIGN{r}{r+8}}}{\SKIP} \\
        & \annotate{\iverson{t \leq 6 \land r \geq 13} \ivop 1}
        \end{align*}
    \end{minipage}
    \qquad
    \begin{minipage}{0.45\textwidth}
        \begin{align*}
        & \phantannotate{\phantom{1 \ssadd \iverson{x \geq 8} \ivop 1 \ssadd \iverson{x \geq 13}}} \\
        %
        %
        &\COMPOSE{\ASSIGN{t}{0}}{\ASSIGN{r}{0}} \,\fatsemi \\
        %
        %
        %
        &\COMPOSE{\ASSIGN{t}{t+2}}{\ASSIGN{r}{r+5}}\,\fatsemi \\
        %
        %
        %
        &\SKIP \,\fatsemi \\
        %
        %
        %
        &\COMPOSE{\ASSIGN{t}{t+4}}{\ASSIGN{r}{r+8}} \\
        & \phantannotate{\phantom{\iverson{t \leq 6 \land r \geq 13}}}
        \end{align*}
    \end{minipage}
    \caption{Program $C$ modeling the Knapsack problem (left) and a valid solution $C'$ (right). The three branchings model the possible choices (either choosing the respective task or not). The variables $t$ and $r$ stand for time and reward, respectively. Variable $x$ is an input parameter.}
    \label{fig:knapsack}
    \centering
\end{figure}

\end{document}